\documentclass[11pt,a4paper]{article}
\usepackage{amsfonts}
\usepackage{graphicx}
\usepackage{indentfirst}
\usepackage{amsmath,amssymb,amsthm}
\usepackage{amssymb}
\usepackage{algorithm,algpseudocode}
\usepackage{latexsym}
\usepackage{natbib}
\usepackage[margin=1in,a4paper]{geometry}
\usepackage{colortbl}
\usepackage{color}
\usepackage[colorlinks,citecolor=blue,urlcolor=blue]{hyperref}
\usepackage{mathtools}
\usepackage{setspace}
\usepackage{verbatim}
\usepackage{subfigure}
\usepackage[title]{appendix}
\usepackage{multirow}

\mathtoolsset{showonlyrefs=true}

\newtheorem{proposition}{Proposition}

\newtheorem{lemma}{Lemma}

\newtheorem{assumption}{Assumption}
\newtheorem{theorem}{Theorem}
\theoremstyle{remark}
\newtheorem{remark}{Remark}

\allowdisplaybreaks[4]

\graphicspath{{../../Codes/}}

\newcommand*\diff{\mathop{}\!\mathrm{d}}

\title{Pricing American Parisian Options under General Time-Inhomogeneous Markov Models}

\author{Yuhao Liu\thanks{School of Science and Engineering, The Chinese University of Hong Kong, Shenzhen, China. Email: oxhowardliu@outlook.com} \and Nian Yang\thanks{Corresponding author. Department of Finance and Insurance, Nanjing University. Email: yangnian@nju.edu.cn} \and Gongqiu Zhang\thanks{Corresponding author. School of Science and Engineering, The Chinese University of Hong Kong, Shenzhen, China. Email: zhanggongqiu@cuhk.edu.cn.}}

\begin{document}

\maketitle

\begin{abstract}
    This paper develops general approaches for pricing various types of American-style Parisian options (down-in/-out, perpetual/finite-maturity) with general payoff functions based on continuous-time Markov chain (CTMC) approximation under general 1D time-inhomogeneous Markov models. For the down-in types, by conditioning on the Parisian stopping time, we reduce the pricing problem to that of a series of vanilla American options with different maturities and their prices integrated with the distribution function of the Parisian stopping time yield the American Parisian down-in option price. This facilitates an efficient application of CTMC approximation to obtain the approximate option price by calculating the required quantities. For the perpetual down-in cases under time-homogeneous models, significant computational cost can be reduced. The down-out cases are more complicated, for which we use the state augmentation approach to record the excursion duration and then the approximate option price is obtained by solving a series of variational inequalities recursively with the Lemke's pivoting method. We show the convergence of CTMC approximation for all the types of American Parisian options under general time-inhomogeneous Markov models, and the accuracy and efficiency of our algorithms are confirmed with extensive numerical experiments.

		

\end{abstract}

\section{Introduction}
    American Parisian option is a type of Parisian option that allows early exercise after activated or before canceled at the Parisian stopping time, e.g., the first time the underlying asset price stays below or above some barrier continuously for a certain duration of time. They can be classified according to different characteristics: (1) the barrier is upper or lower than the initial asset price: up or down; (2) the option is activated or canceled when the Parisian stopping time is triggered: in or out; (3) the maturity is finite or not: finite-maturity or perpetual; (4) payoff function: call or put. In this way, there is a total of 16 types of American Parisian options, providing a flexible structure for speculation and risk management and inheriting the merit of European-style Parisian options in less prone to market manipulation than classical Barrier options. The feature of early-exercising subject to a Parisian stopping time or other random time alike also appear in many other applications: executive options for listed companies (\cite{zhuang2023pricing}), convertible bond (\cite{chu2013pricing}), structural capital structure models (\cite{antill2019optimal}), to name a few.

    However, due to the complicated nature of American Parisian options, the research on their pricing is very limited, even compared to the European-style Parisian options which are already quite challenging (see e.g., \cite{zhang2023general}). The existing literature for American Parisian option pricing are limited in either the Black-Scholes model or jump-diffusion models with relatively simple jump structure, and call or put payoffs. As for the Black-Scholes model and call/put payoff, \cite{chesney2006american} proposed valuation formulae of various types of American Parisian options based on a decomposition approach for the down-in cases and alterative arguments for the down-out cases; \cite{zhu2015pricing} obtained a pricing formula for American Parisian down-in call options by solving the pricing partial differential equation systems; \cite{le2017pricing} developed an integral equation method for pricing American Parisian down-out call options; \cite{le2016analytical} developed a pricing approach for American Parisian up-in calls based on a ``moving window'' technique; \cite{lu2018pricing} priced American Parisian up-out call options by integral equation approach. \cite{haber1999pricing} developed a finite-difference approach for pricing American Parisian options under the Black-Scholes model. As for jump-diffusion models, an analytical maturity-excursion randomization approach is proposed by \cite{chesney2018parisian} for pricing American Parisian options under the hyper-exponential jump-diffusion models with finite jump activities. To the best of our knowledge, there is not a general pricing approach applicable to general jump-diffusions and payoff functions for American Parisian options.
    
    In this paper, we develop efficient methods that are applicable to general one-dimensional time-inhomogeneous Markov models for pricing various types of American Parisian options with general payoffs and establish their convergence. Our idea is to approximate the underlying model by a continuous-time Markov chain (CTMC). In recent years, CTMC approximation has been very successful in developing efficient option pricing approaches under very general stochastic models; its applications include European and barrier options (\cite{mijatovic2013continuously} and \cite{cui2021pricing}), American options (\cite{eriksson2015american}),  Asian options(\cite{cai2015general}, \cite{song2018computable} and \cite{cui2018single}), Parisian options (\cite{zhang2023general}), lookback options (\cite{zhang2023lookback}), drawdown options (\cite{zhang2021pricing} and \cite{zhang2023drawdown}), etc.
	
	We consider a 1D time-inhomogeneous Markov model $X$ living in an interval $I \subseteq \mathbb{R}$ whose infinitesimal generator $\mathcal{G}_t$ takes the following form:
	\begin{align}
		\mathcal{G}_t g(x) &:= \lim_{s \downarrow t} \frac{\mathbb{E}[g(X_s)|X_t=x] - g(x)}{s - t}  \\
        &= \mu(t,x) g'(x) + \frac{1}{2} \sigma^2(t,x) g''(x) + \int_{\mathbb{R}} \big( g(x+z) - g(x) - z g'(x) 1_{\{ z \le 1 \}} \big) \nu(t, x, \diff z),	
	\end{align}
	for $g\in C_c^2(I)$ (twice continuously differentiable with compact support in $I$). Here the drift $\mu(t,x)$, diffusion coefficient $\sigma(t,x)$ and jump measure $\nu(t,x,\diff y)$ are time-dependent, providing flexibility in capturing time-varying factors such as economic cycles, policy changes, etc. To develop CTMC approximation, we follow \cite{zhang2023drawdown} to embed $X$ into a regime-switching time-homogeneous Markov process with the regime capturing the time evolution and approximate it with a CTMC whose generator matrix approximates $\partial_t + \mathcal{G}_t$. Then we perform American Parisian option pricing under the approximating CTMC model.

    \begin{itemize}
        \item The down-in cases.  As for the down-in American Parisian options, at the exercise time $\tau$, the discounted payoff can be written as:
	\begin{align}
		e^{-r\tau}  f( X_\tau) 1_{\{  \tau_{L, D}^- \le \tau \le T, \tau < \infty \}},
	\end{align}
	where $\tau_{L, D}^- := \inf\{ t \ge 0: t - g_{L,t}^- \ge D \}$ is the Parisian stopping time with $\ g_{L,t}^- :=    \sup\{ s \le t: X_s \ge L \}$, $r$ is the risk-free rate,  $T \in  (0, \infty]$ is the maturity and $f(\cdot)$ is the payoff function. Once $\tau_{L, D}^-$ is activated, the option becomes a vanilla American option whose pricing function can be found by standard numerical routines such as the Lemke's pivoting method (\cite{lemke1968complementary}). Given that found, the remained is to find the joint distribution of the Parisian stopping time $\tau_{L, D}^-$ and the state of $X$ at that time. Under the approximating CTMC model, we derive linear systems for these quantities and find that they can be solved efficiently with backward recursion. Especially, when $T = \infty$ (the perpetual case) and the model is time-homogeneous, the recursion can be avoided, leading to significant complexity reduction.

        \item The down-out cases. Compared with the down-in case, the down-out case is much more complicated. At the exercise time $\tau$, the down-out American Parisian option delivers a discounted payoff
	\begin{align}
		e^{-r\tau}  f( X_\tau) 1_{\{  \tau \le \tau_{L, D}^-, \tau  \le T, \tau < \infty \}}.
	\end{align}
	We can see that the effects of Parisian stopping time and maturity on early-exercise are strongly coupled and it is difficult to disentangle them. Thus we utilize a state augmentation approach to add an additional state that records the current duration of excursion below the barrier $L$ at any time $t$ and formulate a system of variational inequalities for the option prices under the approximating CTMC model. We solve the variational inequalities by the Lemke's pivoting method and obtain the approximate option price.
    \end{itemize}

Analyzing the convergence of approximate American Parisian option prices is a challenging task. In the down-in cases, the difficulties of convergence analysis mainly come from the co-existence of Parisian barriers and the early exercise feature. Inspired by the property that the down-in American Parisian option becomes a plain vanilla American option once the Parisian time is triggered, we first prove the convergence of the value function of the American Parisian option conditional on the Parisian stopping time, which is the value function of a vanilla American option, and then show the weak convergence of the Parisian stopping time under the approximating CTMC model; combining these results implies the convergence of the approximate down-in American Parisian option prices. The down-out cases are even more challenging, mainly due to the strong interaction between early exercise and the Parisian stopping time which makes it not applicable to decompose the American Parisian option prices like the down-in cases. We address the difficulties of convergence analysis by carefully constructing some intermediate quantities for the down-out cases.

 The rest of the paper is organized as follows. In Section \ref{sec:coctmc}, we introduce the construction of approximating CTMC. In Section \ref{sec:dictmc}, we present the analysis of down-in American Parisian options with CTMC approximation, and the down-out cases are dealt with in Section \ref{sec:doctmc}. Section \ref{sec:ca} performs the convergence analysis of the proposed approximations. The results of numerical experiments are included in Section \ref{sec:ne}, and Section \ref{sec:con} concludes this paper. Supplementary proofs are provided in Appendix \ref{app:proof}.

\section{ CTMC Approximation}
         \label{sec:coctmc}
    The construction of CTMC approximation for 1D time-inhomogeneous Markov models can be found in \cite{pistorius2009continuously},  \cite{ding2021markov}, and \cite{zhang2023drawdown}. \cite{pistorius2009continuously} and \cite{ding2021markov} build an approximating time-inhomogeneous chain with generator matrices in the time intervals after disretization by matching the infinitesimal behaviors of the original processes,  while \cite{zhang2023drawdown} embed the original 1D time-inhomogeneous Markov process $X$ into a 2D time-homogeneous Markov process with an additional dimension capturing time evolution, which is then approximated by a regime-switching CTMC where the regime process approximates the natural clock. In this paper, we find it convenient to include such an additional approximation for analyzing the distribution of Parisian stopping time and hence adopt the construction method proposed by \cite{zhang2023drawdown}.
    
    \subsection{ Embedding the Time Dimension}
	To construct the approximating CTMC, we first embed $X$ into a time-homogeneous process $(\xi, X_\xi)$ with $\xi_t = \xi_0 + t$. Then the infinitesimal generator of this process can be found as
	\begin{align}
		\mathcal{G}^{\xi,X} g(t,x) &= \lim_{\varepsilon \downarrow 0} \frac{\mathbb{E}_{t,x}[g(\xi_\varepsilon, X_{\xi_\varepsilon}) - g(t,x)]}{\varepsilon} \\
        &= \partial_t g(t, x) + \mathcal{G}_t g(t, x) \\
		&= \partial_t g(t,x) + \mu(t,x) \partial_x g(t, x) + \frac{1}{2} \sigma^2(t,x) \partial_{xx} g(t,x) \\
		&\quad + \int_{\mathbb{R}} \big( g(t,x+z) - g(t,x) - z \partial_x g(t,x) 1_{\{ z \le 1 \}} \big) \nu(t, x, \diff z),
	\end{align}
	where $\mathbb{E}_{t,x}[\cdot] = \mathbb{E}[\cdot| \xi_0 = t, X_{\xi_0} = x]$.

    \subsection{Construction of CTMC}
	We first discretize the state space of $(\xi, X_\xi)$ as $\mathbb{T} \times \mathbb{S}$ with
	\begin{align}
		&\mathbb{T} = \{ t_i = i\delta_t: i = 0, 1, 2, \ldots \},\\
		&\mathbb{S} = \{ y_0, y_1, \ldots, y_n \}.
	\end{align}
	We denote by $\mathbb{S}^o = \mathbb{S} \backslash \{ y_0, y_n \}$ the set of interior states of $X$.  The grid size of $\mathbb{S}$ is measured as the maximum distance between neighboring states,
	\begin{align}
		\delta_x = \max_{x \in \mathbb{S}^-} \delta^+ x = \max_{x \in \mathbb{S}^+} \delta^- x, 
	\end{align}
	where
	\begin{align}
		& \mathbb{S}^+ = \mathbb{S} \backslash \{ y_0 \},\ \mathbb{S}^- = \mathbb{S} \backslash \{ y_n \},\\
		& x^+ = \min_{y > x, y \in \mathbb{S}} y,\ x^- = \max_{y < x, y \in \mathbb{S}} y,\\
		& \delta^+ x = x^+ - x,\ \delta^- x = x - x^-.
	\end{align}
	
	We let $\mathbb{T} \times \mathbb{S}$ be the state space of the approximating CTMC $(\zeta, Y)$ whose generator can be constructed by approximating the derivatives and integral in $\mathcal{G}^{\xi,X} g(t,x)$. For $(t, x) \in \mathbb{T} \times \mathbb{S}$,
	\begin{align}
		\partial_t g(t, x) \approx \frac{g(t+\delta_t, x) - g(t, x)}{\delta_t}.
	\end{align}
	And for $(t,x) \in \mathbb{T} \times \mathbb{S}^o$, the spatial derivatives can be approximated as
	\begin{align}
		\partial_x g(t, x) \approx \nabla g(t, x),\\
		\partial_{xx} g(t, x) \approx \Delta g(t, x),
	\end{align}
	where $\delta x = (\delta^+ x + \delta^- x)/2$ and,
	\begin{align}
		& \nabla^\pm g(t,x) = \pm \frac{g(t,x^\pm) - g(t,x)}{\delta^\pm x},\\
		& \nabla g(t, x) = \frac{\delta^+ x \nabla^- g(t,x) + \delta^-x \nabla^+ g(t,x)}{2\delta x},\\
		&\Delta g(t, x) = \frac{\nabla^+ g(t,x) - \nabla^- g(t,x)}{\delta x}.
	\end{align}
	For the integral in $\mathcal{G}^{\xi,X} g(t,x)$, we can approximate the small jumps as an additional diffusion component and the remained integral as a discrete summation,
	\begin{align}
		&\int_{\mathbb{R}} \big( g(t,x+z) - g(t,x) - z \partial_x g(t,x) 1_{\{ z \le 1 \}} \big) \nu(t, x, \diff z) \\
		&\approx \frac{1}{2} \overline{\sigma}^2(t, x) \partial_{xx} g(t,x) - \overline{\mu}(t, x) \partial_x g(t,x) + \int_{\mathbb{R} \backslash (I_x-x)} \big( g(t,x+z) - g(t,x)  \big)  \nu(t, x, \diff z) \\
		&\approx \frac{1}{2} \overline{\sigma}^2(t, x) \partial_{xx} g(t,x) - \overline{\mu}(t, x) \partial_x g(t,x) + \sum_{y \in \mathbb{S} \backslash \{ x \}} \big( g(t,y) - g(t,x)  \big) \overline{\nu}(t, x, y),
	\end{align}
	where $\overline{\nu}(t, x, y) = \int_{I_y - x} \nu(t, x, \diff z)$, and,
	\begin{align}
		& I_y = [y-\delta^-y/2, y+\delta^+y/2),\ y \in \mathbb{S}^o,\\
		& I_y = (-\infty,y+\delta^+y/2),\ y = y_0, \\
		& I_y = [y-\delta^-y/2, \infty),\ y = y_n,
	\end{align}
	and,
	\begin{align}
		\overline{\sigma}^2(t,x) = \int_{I_x-x} z^2 \nu(t, x, \diff z),\ \overline{\mu}(t,x) = \sum_{y \in \mathbb{S} \backslash\{x\}}(y-x) \int_{I_{y}-x} 1_{\{|z|\le 1\}} \nu(t, x, \diff z).
	\end{align}
	Putting the approximations above together, we get an approximation $\mathcal{G}^{\zeta, Y}$ for $\mathcal{G}^{\xi, X}$ and set it as the generator of CTMC $(\zeta, Y)$.
	\begin{align}
		\mathcal{G}^{\xi, X}g(t,x) &\approx \mathcal{G}^{\zeta, Y}  g(t,x)\\
		&= \frac{g(t+\delta_t, x) - g(t, x)}{\delta_t} + \big(\mu(t, x) - \overline{\mu}(t, x) \big) \nabla g(t, x) + \frac{1}{2} \big(\sigma^2(t, x) + \overline{\sigma}^2(t, x) \big) \Delta g(t,x) \\
		&+ \sum_{y \in \mathbb{S} \backslash \{ x \}} \big( g(t,x+y) - g(t,x)  \big) \overline{\nu}(t, x, y),
	\end{align}
	for $(t, x) \in \mathbb{T} \times \mathbb{S}^o$, while for $(t, x) \in \mathbb{T} \times \{ y_0, y_n \}$, we set
	\begin{align}
		\mathcal{G}^{\zeta, Y}  g(t,x) = \frac{g(t+\delta_t, x) - g(t, x)}{\delta_t}.
	\end{align}
	This means the spatial boundary states $y_0$ and $y_n$ are set as absorbing.
	
	We can get the transition rates from the formula of $\mathcal{G}^{\zeta, Y}$. Let $G_{(t,x) \to (s,y)}$ be the transition rate from state $(t,x)$ to $(s,y)$ for $(t,x) \ne (s,y)$. Then for $(t, x) \in \mathbb{T} \times \mathbb{S}^o$,
	\begin{align}
		& G_{(t,x) \to (t + \delta_t,x)} = \frac{1}{\delta_t},\\
		& G_{(t,x) \to (t, x^+)}  = \big( \mu(t,x) - \overline{\mu}(t,x) \big) \frac{\delta^-x }{2\delta^+x\delta x} + \frac{1}{2\delta^+ x\delta x} \big(\sigma^2(t, x) + \overline{\sigma}^2(t, x) \big) +  \overline{\nu}(t, x, x^+),\\
		& G_{(t,x) \to (t, x^-)} = -\big( \mu(t,x) - \overline{\mu}(t,x) \big) \frac{\delta^+x }{2\delta^-x\delta x}+ \frac{1}{2\delta^- x\delta x} \big(\sigma^2(t, x) + \overline{\sigma}^2(t, x) \big) + \overline{\nu}(t, x, x^-),\\
		& G_{(t, x) \to (t, y)} = \overline{\nu}(t, x, y),\ y \ne x^\pm, x.
	\end{align}
	Other transition rates among different states are all zero, and,
	\begin{align}
		G_{(t,x) \to (t,x)} = -\sum_{(s,y) \in (\mathbb{T} \times \mathbb{S}) \backslash \{(t,x)\}} G_{(t,x) \to (s,y)}.
	\end{align}
	We put the transition rates between spatial states into the matrix ${\pmb G}_t^Y$ with elements
	\begin{align}
		{\pmb G}_t^Y(x, y) = G_{(t,x) \to (t, y)},\ x, y \in \mathbb{S}.
	\end{align}
	The corresponding generator form is
	\begin{align}
		\mathcal{G}_t^Y g(t, x) = \sum_{y \in \mathbb{S}} {\pmb G}_t^Y(x, y) g(t, y).
	\end{align}
	Then $\mathcal{G}^{\zeta, Y}$ can be rewritten as
	\begin{align}
		\mathcal{G}^{\zeta, Y}g(t, x) &= \frac{g(t+\delta_t,x) - g(t,x)}{\delta_t} + \mathcal{G}_t^Y g(t, x) \\ 
		&= \frac{g(t+\delta_t,x) - g(t,x)}{\delta_t} + \sum_{y \in \mathbb{S}} {\pmb G}_t^Y(x, y) g(t,y).
	\end{align}
	
	In the following sections, we first develop the pricing methods of American Parisian options under the approximating CTMC model and then show the convergence of approximate prices to exact ones under the continuous model.

     \begin{remark}
         In the time-homogeneous case, we note that the spatial transition rates of $Y$ are not modulated by $\zeta$ and hence $Y$ alone is a time-homogeneous CTMC.
     \end{remark}

\section{Down-In American Parisian Options under the CTMC Approximation}
        \label{sec:dictmc}
        In this section, we focus on perpetual and finite-maturity down-in American Parisian options under the approximating CTMC model. We first present the pricing approach for the finite-maturity case and then show how complexity reduction can be achieved for the perpetual case.  
 \subsection{The Finite-Maturity Case}
	In this subsection, we assume that $T < \infty$.  In the CTMC model $(\zeta, Y)$, $Y$ approximates the spatial dynamics of $X$ and $\zeta$ approximates the natural clock such that $\zeta_t \approx \zeta_0 + t = t$ with $\zeta_0 = 0$. Therefore, the discounted payoff at exercise can be approximated as
	\begin{align}
		e^{-r\tau}  f( X_\tau) 1_{\{ \tau_{L, D}^{-} \le \tau \le T\}} \approx e^{-r\zeta_{\tau}}  f( Y_\tau) 1_{\{ \tau_{L, D}^{-} \le \tau, \zeta_{\tau} \le T \}}.
	\end{align}
    The reason to add such an approximation is that we can do conditioning on $\zeta_{\tau_{L, D}^-}$ instead of $\tau_{L, D}^-$  which eases the preceding derivations as the former has a discrete distribution while the latter is a continuous random variable. Assuming that the current duration of excursion is $0$, we have the approximate down-in American Parisian option price at time $t$:
	\begin{align}
		C_{fi}(t,x) := \sup_{\tau \in \mathcal{T}_t}\mathbb{E}_{t,x} \big[ e^{-r(\zeta_{\tau} - t)}  f( Y_\tau) 1_{\{ \tau_{L, D}^- \le \tau, \zeta_{\tau} \le T \}} \big],
	\end{align}
	where $\mathbb{E}_{t,x}[\cdot] = \mathbb{E}_{t,x}[\cdot | \zeta_t = t, Y_t = x]$ and $\mathcal{T}_t$ is the set of stopping times taking value in $[t,\infty]$. With a light abuse of notation, we still use $\tau_{L, D}^-$ to denote the Parisian stopping time defined from $Y$.

    We note that $C_{fi}(t, x)$ is essentially a perpetual American option pricing problem with a Parisian stopping time constraint. The option is activated at $\tau_{L, D}^-$ and after that it becomes a vanilla American option. By the strong Markovian property of CTMC, we have that
    \begin{align}
        C_{fi}(t, x) = \mathbb{E}_{t, x} \Big[ e^{-r(\zeta_{\tau_{L, D}^-} - t)} \sup_{\tau \in \mathcal{T}_{\tau_{L, D}^-}}\mathbb{E}\big[ e^{-r(\zeta_\tau - \zeta_{\tau_{L, D}^-})} f(Y_\tau) 1_{\{ \zeta_\tau \le T \}} | \zeta_{\tau_{L, D}^-}, Y_{\tau_{L, D}^-} \big] \Big],
    \end{align}
    where the inner side is the price of a vanilla American option. Integrating it w.r.t. the distribution of $(\zeta_{\tau_{L, D}^-}, Y_{\tau_{L, D}^-})$ yields a decomposition of $C_{fi}(t, x)$ as follows.
	\begin{theorem}\label{fdi}
		$C_{fi}(t,x)$ can be decomposed as,
		\begin{align}
			C_{fi}(t,x) = \sum_{(s,y) \in \mathbb{T} \times \mathbb{S}} e^{-r(s-t)} h(t,x; s,y) c_f(s,y),
		\end{align}	
		where
		\begin{align}
			& h(t,x; s,y) = \mathbb{P}_{t,x}\big[  \zeta_{\tau_{L,D}^-} =s, Y_{\tau_{L,D}^-} = y \big], \\
			& c_f(s, y) = \sup_{\tau \in \mathcal{T}_s} \mathbb{E}_{s,y} \big[ e^{-r(\zeta_\tau - s)} f(Y_\tau) 1_{\{ \zeta_{\tau} \le T \}} \big].
		\end{align}
	\end{theorem}

    For $h(t,x; s, y)$, we derive a linear system for it as follows.
	\begin{proposition}\label{prop:htx-linear-system}
		$h(t, x; s, y) = 0$ for $s < t$. And for $(t,x), (s,y) \in \mathbb{T} \times \mathbb{S}$ with $s \ge t$,
		\begin{align}
			h(t, x; s, y) &= 1_{\{ x < L \}} \sum_{z \ge L, s' \ge t} \mathbb{P}_{t,x} \big[  \tau_L^+ < D, Y_{\tau_L^+} = z,\zeta_{\tau_L^+} = s' \big] h(s',z; s,y) \\
			&\quad + 1_{\{ x < L \}}  \mathbb{P}_{t,x} \big[  \tau_L^+ \ge D, Y_D = y, \zeta_D = s \big] \\
			&\quad + 1_{\{ x \ge L \}} \sum_{z < L, s' \ge t} \mathbb{P}_{t,x} \big[ \zeta_{\tau_L^-} = s', Y_{\tau_L^-} = z \big] h(s', z; s, y),
		\end{align}
		where $\tau_L^+ = \inf\{ t \ge 0: Y_t \ge L \}$ and $\tau_L^- = \inf\{ t \ge 0: Y_t < L \}$.
	\end{proposition}

    \begin{proposition}\label{prop:cfi-linear-system}
		The discounted option price $\widetilde{C}_{fi}(t,x) = e^{-rt} C_{fi}(t, x)$ satisfies the following linear system:
		\begin{align}
			\widetilde{C}_{fi}(t,x) &= 1_{\{ x < L \}} \sum_{z \ge L, s \ge t} \mathbb{P}_{t,x} \big[  \tau_L^+ < D, Y_{\tau_L^+} = z,\zeta_{\tau_L^+} = s \big] \widetilde{C}_{fi}(s,z) \\
			&\quad + 1_{\{ x < L \}} \sum_{(s,y) \in \mathbb{T}\times \mathbb{S}, s \ge t}  \mathbb{P}_{t,x} \big[  \tau_L^+ \ge D, Y_D = y,\zeta_D = s \big] e^{-rs} c_f(s, y) \\
			&\quad + 1_{\{ x \ge L \}} \sum_{z < L, s \ge t} \mathbb{P}_{t,x} \big[ \zeta_{\tau_L^-} = s, Y_{\tau_L^-} = z \big] \widetilde{C}_{fi}(s, z). 
		\end{align}
	\end{proposition}

	Proposition \ref{prop:cfi-linear-system} implies a recursive approach to calculate $\widetilde{C}_{fi}(t,x)$.
	\begin{itemize}
		\item For $t = T^+ := \min\{s \in \mathbb{T}: s> T\}$, $\widetilde{C}_{fi}(T^+, x) = 0$.
		
		\item For $t = T^+-\delta_t, T^+ - 2\delta_t, \ldots$, given $\widetilde{C}_{fi}(s,x)$ for all $s > t$, we solve the linear system for $\widetilde{C}_{fi}(t, x)$ (a reformulation of Proposition \ref{prop:cfi-linear-system}):
		\begin{align}\label{eq:cfi-linear-system}
			\widetilde{C}_{fi}(t,x) &= 1_{\{ x < L \}} \sum_{z \ge L} h^+(D, t, x, z) \widetilde{C}_{fi}(t,z) + 1_{\{ x \ge L \}} \sum_{z < L} h^-(t, x, z) \widetilde{C}_{fi}(t, z) \\
			&\quad + 1_{\{ x < L \}} u^+(D, t, x)  + 1_{\{ x < L \}} v(D, t, x)  + 1_{\{ x \ge L \}} u^-(t, x), 
		\end{align}
		where 
		\begin{align}
			& v(D, t, x) = \sum_{(s,y) \in \mathbb{T}\times \mathbb{S}}  \mathbb{P}_{t,x} \big[  \tau_L^+ \ge D, Y_D = y,\zeta_D = s \big] e^{-rs} c_f(s, y),\\
			& u^-(t, x) = \sum_{z < L, s > t} \mathbb{P}_{t,x} \big[ \zeta_{\tau_L^-} = s, Y_{\tau_L^-} = z \big] \widetilde{C}_{fi}(s, z),\\
			& h^+(D, t, x, z) = \mathbb{P}_{t,x} \big[  \tau_L^+ < D, Y_{\tau_L^+} = z,\zeta_{\tau_L^+} = t \big],\\
			& h^-(t, x, z) = \mathbb{P}_{t,x} \big[ \zeta_{\tau_L^-} = t, Y_{\tau_L^-} = z \big],\\
			& u^+(D, t, x) = \sum_{z \ge L, s > t} \mathbb{P}_{t,x} \big[  \tau_L^+ < D, Y_{\tau_L^+} = z,\zeta_{\tau_L^+} = s \big] \widetilde{C}_{fi}(s,z).
		\end{align}
	\end{itemize}
	
	

Next we analyze each of the quantities involved in the recursion above.

\subsubsection{Analysis of $c_f(s, y)$}

We note that $c_f(s, y)$ is the price of a perpetual vanilla American option written on $(\zeta, Y)$ jointly. By the general theory of optimal stopping, $c_f(s,y)$ satisfies the variational inequalities for all $(s,y) \in \mathbb{T} \times \mathbb{S}$ (\cite{eriksson2015american}),
	\begin{align}
		\begin{cases}
			\mathcal{G}^{\zeta, Y} c_f(s, y) \le 0,\\
			c_f(s,y) \ge f(y) 1_{\{ s \le T \}},\\
			\big(\mathcal{G}^{\zeta, Y} c_f(s, y)\big) \big( c_f(s,y) - f(s, y) 1_{\{ s \le T \}} \big) = 0.
		\end{cases}
	\end{align}
	The system above can be solved recursively.
	\begin{itemize}
		\item First, $c_f(T^+, y) = 0$.
		\item For $s = T^+-\delta_t, T^+-2\delta_t,\ldots $, given that $c_f(s+\delta_t, y)$ has been found for all $y$. $c_f(s,y)$ can be found by solving the following system (where we write the action of $\mathcal{G}^{\zeta, Y}$ explicitly) with the Lemke's pivoting method (\cite{lemke1968complementary}).
		\begin{align}
			\begin{cases}
				\sum_{y \in \mathbb{S}} {\pmb G}_s^Y c_f(s,y) - \frac{1}{\delta_t} c_f(s,y) \le -\frac{1}{\delta_t} c_f(s+\delta_t, y),\\
				c_f(s, y) \ge f(y),\\
				\big( \sum_{y \in \mathbb{S}} {\pmb G}_s^Y c_f(s,y) - \frac{1}{\delta_t} c_f(s,y)  + \frac{1}{\delta_t} c_f(s+\delta_t, y) \big) \big( c_f(s,y) - f(y) \big) = 0.
			\end{cases}
		\end{align}
		Let
		\begin{align}
			{\pmb c}_f(s) = \big( c_f(s,y) \big)_{y \in \mathbb{S}},\ {\pmb f} = \big( f(y) \big)_{y \in \mathbb{S}}.
		\end{align}
		Then the system can be written in a vectorized form
		\begin{align}\label{eq:variational-ineq-finite-down-in}
			\min \big( ({\pmb I} - {\pmb G}_t^Y \delta_t) {\pmb c}_f(s) - {\pmb c}_f(s+\delta_t), {\pmb c}_f(s) - {\pmb f}   \big) = ({\pmb 0})_{n+1},
		\end{align}
  where $(\pmb{0})_{n+1}$ is an all-zero vector in $\mathbb{R}^{n+1}$ and ${\pmb I}$ is the identity matrix in $\mathbb{R}^{(n+1)\times (n+1)}$.
		The linear complementary problem (LCP) can be formulated as
		\begin{align}\label{LCP-finite-down-in}
		(\operatorname{LCP}(A, \psi))\left\{\begin{array}{l}
		A \overline{C}+\psi \geq \mathbf{0}, \\
		\overline{C} \geq \mathbf{0}, \\
		\overline{C}^{\top}(A \overline{C}+\psi)=0,
		\end{array}\right.
		\end{align}
		where $A={\pmb I} - {\pmb G}_t^Y \delta_t$, $\overline{C}={\pmb c}_f(s) - {\pmb f}$, and $\psi=A{\pmb f}-{\pmb c}_f(s+\delta_t)$.
	\end{itemize}

\subsubsection{Analysis of $v(D, t, x)$, $u^-(t, x)$, $h^+(D, t, x, z)$, $h^-(t, x, z)$ and $u^+(D, t, x)$}

In this part, we derive ODE and linear systems for $v(D, t, x)$, $u^-(t, x)$, $h^+(D, t, x, z)$, $h^-(t, x, z)$, and $u^+(D, t, x)$ which imply analytical solutions for them in terms of matrix exponentiation or inversion. The derivations of those equations are shown in Appendix \ref{app:proof}. In the following,  ${\pmb I}_L^+ = \operatorname{diag}\big((1_{x\ge L})_{x \in \mathbb{S}}\big)$ and ${\pmb I}_L^- = {\pmb I} - {\pmb I}_L^+$. We also use the symbol $(A)_{k \times l}$ to represent a new matrix composed of the first $k$ rows and the first $l$ columns of the matrix $A$.

\begin{proposition}\label{prop:v-system}
     $v(D, t, x)$ satisfies the following ODE system:
     \begin{align}
		\begin{cases}
			\partial_D v(D, t, x) = \frac{v(D, t+\delta_t, x) - v(D, t, x)}{\delta_t} + \mathcal{G}_t^Y v(D, t, x),\quad t \le T, D>0, x < L,\\
			v(D, t, x)= 0,\quad D>0, x \ge L,\\
                v(D, t, x)= 0,\quad t>T, \\
                v(0,t,x) = e^{-rt} c_f(t, x).
		\end{cases}
	\end{align}
 \end{proposition}

Next we show the procedures of calculating $v(D, t, x)$. Let $m$ be an integer in $\{0, 1, 2, \ldots, n\}$ such that $y_m<L$, $y_{m+1}\ge L$, $(\pmb{0})_{n-m}$ be the all-zero vector in $\mathbb{R}^{n-m}$, and
	\begin{align}
         & \overline{{\pmb v}}(D, t)=(v(D,t,x))_{x\in \{y_0, y_1, \cdots, y_m\}},\\
         & {\pmb v}(D, t)=(v(D,t,x))_{x\in \mathbb{S}},\\
		&\overline{{\pmb c}}_f(t)=\big( e^{-rt} c_f(t, x)  \big)_{x\in \{y_0, y_1, \cdots, y_m\}}.
	\end{align}
	Then, for $x<L$, we have the solution
	\begin{align}\label{eq:vDt-integral-recursion}
		\overline{\pmb v}(D, t) = \exp\big( (\overline{\pmb G}_t^Y - \overline{\pmb I}/\delta_t) D \big)\overline{\pmb c}_f(t) + \int_{0}^{D} \exp\big((\overline{\pmb G}_t^Y - \overline{\pmb I}/\delta_t) (D-s) \big) \frac{1}{\delta_ t}\overline{\pmb v}(s, t+\delta_t) \diff s,
	\end{align}
 where $\overline{{\pmb G}}_t^Y=({\pmb G}_t^Y)_{(m+1)\times (m+1)}$ and $\overline{\pmb I} = {\pmb I}_{(m+1)\times (m+1)}$.
 
Let $\pmb M=\overline{{\pmb G}}_t^Y - {\pmb I}/\delta_t$. According to Proposition \ref{prop:v-system},  $v(D,t,x)$ can be calculated by recursively evaluating the integrals in \eqref{eq:vDt-integral-recursion} from $T^+-\delta_t$. We get that, for $t = T^+ - (j+1)\delta_t$ with $j=0, 1,2,\cdots, \frac{T^+}{\delta_t}-1$,
\begin{align}
\overline{\pmb v}(D, t) = \sum_{k=0}^{j} (\frac{1}{\delta_t})^ke^{D\pmb{M}}\frac{D^k}{k!}\overline{\pmb c}_f(t+k\delta_t).
\end{align}

For $x\ge L$, we have $v(D, t, x)=0$. Then
 \begin{align}
     {\pmb v}(D, t)=\begin{bmatrix}
  \overline{\pmb{v}}(D,t)\\
(\pmb{0})_{n-m}
\end{bmatrix}\in \mathbb{R}^{n+1}.
 \end{align}

\begin{proposition}\label{prop:h^+-system}
$h^+(D,t,x,z)$ satisfies the following ODE system:
    \begin{align}
		\begin{cases}
			\partial_D h^+(D,t,x, z)= (\mathcal{G}_t^Y - 1/\delta_t) h^+(D,t,x, z),\quad x < L, t\le T, D>0,\\
			h^+(D, t, x, z) = 1_{\{ x = z \}},\quad  x \ge L, t\le T, D>0, \\
             h^+(0, t, x, z) = 0.
		\end{cases}
	\end{align}
$h^+(D, t, x, z)$ can be decomposed into two parts:
	\begin{align}
		h^+(D, t, x, z) = h^+_1(t, x, z) - h^+_2(D, t, x, z).
	\end{align}
	Here $h^+_1(t, x, z)$ satisfies,
	\begin{align}
		\begin{cases}
			(\mathcal{G}_t^Y - 1/\delta_t) h^+_1(t, x, z) = 0,\quad x < L, t \le T, \\
			h^+_1(t, x, z) = 1_{\{ x=z \}},\quad x \ge L, t \le T. 
		\end{cases}
	\end{align}
	And $h^+_2(D, t, x, z)$ satisfies
	\begin{align}
		\begin{cases}
			\partial_D h^+_2(D, t, x, z) = (\mathcal{G}_t^Y - 1/\delta_t) h^+_2(D, t, x, z),\quad x < L, t \le T, D > 0,\\
			h^+_2(D, t, x, z) = 0,\quad x \ge L, t \le T, D > 0,\\
			h^+_2(0, t, x, z) = h^+_1(t, x, z).
		\end{cases}
	\end{align}
 \end{proposition}

The solutions to the equations of Proposition \ref{prop:h^+-system} can be given in closed-form as follows. 
	\begin{align}
		& {\pmb H}_1^+(t) := \big(h^+_1(t, x, z)\big)_{x, z \in \mathbb{S}} = \big( {\pmb I}_L^- ({\pmb I}/\delta_t -  {\pmb G}^Y_t) + {\pmb I}_L^+ \big)^{-1} {\pmb I}_L^+,\\
		& {\pmb H}_2^+(D, t) :=  \big( h^+_2(D, t, x, z) \big)_{x, z \in \mathbb{S}} = \exp\big( {\pmb I}_L^- ({\pmb G}^Y_t - {\pmb I/\delta_t}) D \big) {\pmb I}_L^-{\pmb H}_1^+(t),\\
		& {\pmb H}^+(D, t) := \big( h^+(D, t, x, z) \big)_{x, z \in \mathbb{S}} = {\pmb H}_1^+(t) - {\pmb H}_2^+(D, t).
	\end{align}

\begin{proposition}\label{prop:h^-system}
    $h^-(t,x, z)$ satisfies the following linear system:
    \begin{align}
		\begin{cases}
			\mathcal{G}_t^Y h^-(t,x,z)-\frac{h^-(t,x,z)}{\delta_t} = 0,\quad x \ge L, t\le T\\
			h^-(t, x, z) = 1_{\{ x = z \}} ,\quad x < L, t\le T. 
		\end{cases}
	\end{align} 
\end{proposition}

Then $h^-(t,x, z)$ is calculated as follows:
	\begin{align}
		{\pmb H}^-(t) := \big( h^-(t, x, z) \big)_{x, z \in \mathbb{S}} = \big( {\pmb I}_L^+ ({\pmb I}/\delta_t - {\pmb G}_t^Y) + {\pmb I}_L^- \big)^{-1} {\pmb I}_L^-.
	\end{align}

\begin{remark}
    In Propositions \ref{prop:h^+-system} and \ref{prop:h^-system}, if $Y$ is a time-homogeneous CTMC, then $h^+(D,t,x,z)$ and $h^-(t,x,z)$ are independent of $t$ and hence we only need to do the calculations above once for all $t$.
\end{remark}

\begin{proposition}\label{prop:u^+-system}
    $u^+(D,t,x)$ satisfies the following equations:
    \begin{align}
		\begin{cases}
			 \partial_D u^+(D,t,x) = \frac{1}{\delta_t} \sum_{z \ge L} h^+(D, t+\delta_t, x, z) \widetilde{C}_{fi}(t+\delta_t, z) \\
                    \qquad\qquad\qquad+ \big( u^+(D,t+\delta_t,x) - u^+(D,t,x) \big) /\delta_t \\
                    \qquad\qquad\qquad+ \mathcal{G}_t^Y u^+(D,t,x),\quad x < L, t\le T, D>0,\\
			u^+(D, t, x) = 0 ,\quad t > T, \\
              u^+(D, t, x) = 0 ,\quad x \ge L, \\
             u^+(0, t, x) = 0.
		\end{cases}
	\end{align}
\end{proposition}

 Let
	\begin{align}
		&\overline{{\pmb u}}^+(D, t) :=\big( u^+(D, t, x) \big)_{x \in \{y_0, y_1, \cdots, y_m \}},\\
            &\pmb{u}^+(D,t)=\big( u^+(D, t, x) \big)_{x \in \mathbb{S}}, \\
		&\overline{{\pmb H}}_1^+(t):=({\pmb H}_1^+(t))_{(m+1)\times (n+1)},\\
        & \widetilde{\pmb{C}}_{fi}(t):=\big(\widetilde{C}_{fi}(t,x) \big)_{x \in \mathbb{S}}.
	\end{align}
Then, for $x<L$, the solution is given by
	\begin{align}\label{eq:uDt-integral-recursive}
		\overline{{\pmb u}}^+(D, t) = \int_0^D \exp\big( (\overline{{\pmb G}}_t^Y - \overline{\pmb I}/\delta_t) (D - s)\big) \frac{1}{\delta_t} \big( \overline{{\pmb H}}^+(s, t+\delta_t) \widetilde{\pmb C}_{fi}(t+\delta_t) + \overline{{\pmb u}}^+(s, t+\delta_t) \big)  \diff s,
	\end{align}
	where $\overline{{\pmb H}}^+(s,t)=\overline{{\pmb H}}_1^+(t)-\exp((\overline{{\pmb G}}_t^Y-\overline{\pmb I}/\delta_t)s)\overline{{\pmb H}}_1^+(t)$. $\overline{{\pmb u}}^+(D, t)$ can be found by evaluating the integral in \eqref{eq:uDt-integral-recursive} recursively from $T$. Let $(\pmb{0})_{(m+1)\times (m+1)}$ be the all-zero matrix in $\mathbb{R}^{(m+1)\times (m+1)}$. We get that, for $t = T^+ - (j+1)\delta_t$ with $j=0, 1,2,\cdots, \frac{T^+}{\delta_t}-1$,
	\begin{align}
		\overline{{\pmb u}}^+(D, t) = \sum_{k=0}^{j} \pmb{C}_k(D)\overline{{\pmb H}}_1^+(t)\widetilde{\pmb C}_{fi}(t+k\delta_t),
	\end{align}
 where 
 \begin{align}
\pmb{C}_k(D)=\begin{cases}(\frac{1}{\delta_t})^k\big((-1)^k(\pmb{M}^{-1})^k+\sum_{m=1}^{k+1}(-1)^{k+m}e^{D\pmb{M}}\frac{D^{(m-1)}}{(m-1)!}(\pmb{M}^{-1})^{k+1-m}  \big), k>0,\\
(\pmb{0})_{(m+1)\times (m+1)}, k=0.
\end{cases}
 \end{align}
 
For $x\ge L$, we have that $u^+(D, t, x)=0$. Then
 \begin{align}
     \pmb{u}^+(D,t)=\begin{bmatrix}
  \overline{\pmb{u}}^+(D,t)\\
(\pmb{0})_{n-m}
\end{bmatrix}\in \mathbb{R}^{n+1}.
 \end{align}

\begin{proposition}\label{prop:u^-system}
    $u^-(t,x)$ satisfies the following linear system:
    \begin{align}
		\begin{cases}
			\frac{1}{\delta_t}\sum_{z < L} h^-(t+\delta_t, x, z) \widetilde{C}_{fi}(t+\delta_t, z) + \frac{u^-(t+\delta_t,x)-u^-(t,x)}{\delta_t} + \mathcal{G}^Y_t u^-(t,x)= 0,\quad x \ge L, t\le T\\
			u^-(t,x)= 0,\quad x < L, \\
                u^-(t,x)= 0,\quad t>T.
		\end{cases}
	\end{align}
Let
	\begin{align}
		& {\pmb u}^-(t) = \big( u^-(t,x) \big)_{x \in \mathbb{S}},\\
		& \widetilde{\pmb C}_{fi}(t+\delta_t) = \big( \widetilde{C}_{fi}(t + \delta_t, z) \big)_{z \in \mathbb{S}}.
	\end{align}
\end{proposition}
The solution is given by
\begin{align}
			{\pmb u}^-(t) = \big( {\pmb I}^-_L + {\pmb I}^+_L ({\pmb I} - {\pmb G}^Y_t \delta_t) \big)^{-1} {\pmb I}^+_L \big( {\pmb u}^-(t+\delta_t) + {\pmb H}^-(t) \widetilde{\pmb C}_{fi}(t + \delta_t)  \big).
		\end{align} 
Noting that ${\pmb u}^-(T^+) = ({\pmb 0})_{n+1}$, ${\pmb u}^-(t)$ can be calculated recursively from $t = T^+$.

\subsubsection{The Algorithm and Complexity Analysis}

After we obtain all the required quantities, we can calculate the approximate finite-maturity down-in American Parisian option prices by a recursive procedure implied by the Proposition \ref{prop:cfi-linear-system}, which we summarize in Algorithm \ref{algo:finite-down-in-Amer-Pari-option-price}. 

Next we analyze the time complexity of our algorithm. For jump models, constructing $\pmb{G}_t^Y\in \mathbb{R}^{(n+1)\times (n+1)}$ and $\pmb{M}\in \mathbb{R}^{(m+1)\times (m+1)}$ costs $\mathcal{O}(n^2)$ (note that $m$ should be proportional to $n$) and the calculations of the quantities that involve matrix exponentiation, multiplication or inversion cost $\mathcal{O}(n^3)$. For the LCP \eqref{LCP-finite-down-in}, constructing it costs $\mathcal{O}(n^2)$ mainly due to the complexity of calculating $(\pmb{I}-\pmb{G}_t^Y\delta_t)\pmb{f}$, and by \cite{cottle2009linear}, solving the LCP for $\overline{\pmb{c}}_f(t)$ with the Lemke's pivoting method requires $\mathcal{O}(n)$ pivot steps. Therefore, the total complexity for computing $\widetilde{\pmb{C}}_{fi}(t)$ in each recursion step is $\mathcal{O}(n^3)$ for jump models.

As for the diffusion case, the computational cost can be significantly reduced.  Define $\pmb{1}_y=(1_{\{x=y\}})_{x\in \mathbb{S}}$, $L^+=\min \{z\ge L: z\in \mathbb{S}\}$ and $L^-=\max \{z< L: z\in \mathbb{S}\}$. Since the process $Y$ is a birth-and-death process, $1_{\{x<L\}}h^+(D,t,x,z)\neq 0$ only when $z=L^+$ and $1_{\{x\ge L\}}h^-(t,x,z)\neq 0$ only when $z=L^-$. So \eqref{eq:cfi-linear-system} reduces to
\begin{align}\label{eq:cfi-diffusion-linear-system}
\widetilde{C}_{fi}(t,x) &= 1_{\{ x < L \}}  h^+(D, t, x, L^+) \widetilde{C}_{fi}(t,L^+) + 1_{\{ x \ge L \}}  h^-(t, x, L^-) \widetilde{C}_{fi}(t, L^-) \\
			&\quad + 1_{\{ x < L \}} u^+(D, t, x)  + 1_{\{ x < L \}} v(D, t, x)  + 1_{\{ x \ge L \}} u^-(t, x). 
\end{align}
Setting $x=L^+$ and $x=L^-$ respectively in \eqref{eq:cfi-diffusion-linear-system}, we have
\begin{align}
    \quad & \widetilde{C}_{fi}(t,L^+)= h^-(t, L^+, L^-) \widetilde{C}_{fi}(t, L^-)+ u^-(t, L^+),\\
    & \widetilde{C}_{fi}(t,L^-)=  h^+(D, t, L^-, L^+) \widetilde{C}_{fi}(t,L^+)+ u^+(D, t, L^-)+v(D, t, L^-).
\end{align}
Solving the linear system above for $\widetilde{C}_{fi}(t,L^+)$ and $\widetilde{C}_{fi}(t,L^-)$ yields
\begin{align}
    &\widetilde{C}_{fi}(t,L^-)=\frac{h^+(D,t,L^-,L^+)u^-(t,L^+)+u^+(D,t,L^-)+v(D,t,L^-)}{1-h^+(D,t,L^-,L^+)h^-(t,L^+,L^-)}\label{eq:Cfi-L^-}, \\
    &\widetilde{C}_{fi}(t,L^+)=h^-(t,L^+,L^-)\widetilde{C}_{fi}(t,L^-)+u^-(t, L^+)\label{eq:Cfi-L^+}.
\end{align}
The quantities involved above can be computed as follows:
\begin{align}
    &\pmb{h}_1^+(t)=\big(h_1^+(t,x,L^+)\big)_{x\in \mathbb{S}}=\big(\pmb{I}_L^-(\pmb{I}/\delta_t-\pmb{G}_t^Y)+\pmb{I}_L^+\big)^{-1}\pmb{1}_{L^+},\\
    &\pmb{h}_2^+(t)=\big(h_2^+(D,t,x,L^+)\big)_{x\in \mathbb{S}}=\exp\big(\pmb{I}^-_L(\pmb{G}_t^Y-\pmb{I}/\delta_t)D\big)\pmb{I}_L^-\pmb{h}_1^+(t),\\
    &\pmb{h}^+(D,t)=\big(h^+(D,t,x,L^+)\big)_{x\in \mathbb{S}}=\pmb{h}_1^+(t)-\pmb{h}_2^+(t),\\
    &\pmb{h}^-(t)=\big(h^-(t,x,L^-)\big)_{x\in \mathbb{S}}=\big(\pmb{I}^+_L(\pmb{I}/\delta_t-\pmb{G}_t^Y)+\pmb{I}^-_L \big)^{-1}\pmb{1}_{L^-}.
\end{align}
It follows that 
\begin{align}
   & h^+(D,t,L^-,L^+)=\pmb{1}_{L^-}^\top\pmb{h}^+(D,t),\\
   & h^-(t,L^+,L^-)=\pmb{1}_{L^+}^\top\pmb{h}^-(t).
\end{align}
Moreover,
\begin{align}
    &u^-(t,L^+)=\pmb{1}_{L^+}^\top\pmb{u}^-(t),\\
    &u^+(D,t,L^-)=\pmb{1}_{L^-}^\top\pmb{u}^+(D,t),\\
    &v(D,t,L^-)=\pmb{1}_{L^-}^\top\pmb{v}(D,t).
\end{align}
Then $\widetilde{C}_{fi}(t,L^-)$ and $\widetilde{C}_{fi}(t,L^+)$ can be found by substituting the above quantities back into \eqref{eq:Cfi-L^-} and \eqref{eq:Cfi-L^+}. The approximate down-in American Parisian option prices can be computed as follows:
\begin{align}
    \widetilde{\pmb{C}}_{fi}(t)=\pmb{I}_L^-\pmb{h}^+(D,t)\widetilde{C}_{fi}(t,L^+)+\pmb{I}_L^+\pmb{h}^-(t)\widetilde{C}_{fi}(t,L^-)+\pmb{I}_L^{-}\pmb{u}^+(D,t)+\pmb{I}_L^{-}\pmb{v}(D,t)+\pmb{I}_L^{+}\pmb{u}^{-}(t).
\end{align}
The costs of constructing $\pmb{G}_t^Y$ and $\pmb{M}$ are $\mathcal{O}(n)$. Since $\pmb{G}_t^Y$ is a tridiagonal matrix, the calculations of $\pmb{h}_1^+(t)$, $\pmb{h}^-(t)$ and $\pmb{u}^-(t)$ cost $\mathcal{O}(n)$. For any tridiagonal matrix $\pmb{A}\in \mathbb{R}^{n\times n}$ and vector $b\in \mathbb{R}^n$, the matrix exponentiation $\exp{(\pmb{A})}b$ can be approximated by
\begin{equation}
    \exp{(\pmb{A})}b \approx (I-\pmb{A}/k)^{-k}b,  \label{eq:exp-approximation}
\end{equation}
which costs $\mathcal{O}(kn)$ according to \cite{zhang2023general}. Then calculating $\pmb{h}_2^+(D,t)$ and $\pmb{v}(D,t)$ costs $\mathcal{O}(kn)$ using the approximation in \eqref{eq:exp-approximation}. As for $\pmb{u}^+(D,t)$, we need to do the following calculation:
\begin{align}
    \overline{\pmb{u}}^+(D,t)=&\sum_{k=1}^{j} \bigg((\frac{1}{\delta_t})^k\big((-1)^k(\pmb{M}^{-1})^k\\
    &+\sum_{m=1}^{k+1}(-1)^{k+m}e^{D\pmb{M}}\frac{D^{(m-1)}}{(m-1)!}(\pmb{M}^{-1})^{k+1-m}  \big)\bigg)\overline{{\pmb h}}_1^+(t)\big(\pmb{1}_{L^+}^\top\widetilde{\pmb C}_{fi}(t+k\delta_t)\big),
\end{align}
where $\overline{{\pmb h}}_1^+(t)=\big(h_1^+(t,x,L^+)\big)_{x\in \{y_0, y_1, \cdots, L^-\}}$ and $j=(T^+-t)/\delta_t$. The calculation of $\overline{\pmb{u}}^+(D,t)$ involves solving the system $\pmb{M}x=b_1$ and calculating $\exp(D\pmb{M})b_2$ given some vectors $b_1, b_2\in \mathbb{R}^{m+1}$ and the tridiagonal matrix $\pmb{M}\in \mathbb{R}^{(m+1)\times (m+1)}$, which cost $\mathcal{O}(n)$ and $\mathcal{O}(kn)$ respectively. In summary, computing $\pmb{u}^+(D,t)$ costs $\mathcal{O}(kn)$. Constructing the LCP \eqref{LCP-finite-down-in} costs $\mathcal{O}(n)$ as $\pmb{I}-\pmb{G}_t^Y\delta_t$ tridiagonal and by \cite{cottle2009linear}, solving \eqref{LCP-finite-down-in} costs $\mathcal{O}(n)$. Considering all the costs analyzed above, the total complexity for calculating $\widetilde{\pmb{C}}_{fi}(t)$ in each recursion step is $\mathcal{O}(kn)$ in the diffusion case.

Considering that the largest number of recursion steps for the time variable is $n_t=T^+/\delta_t$ in Algorithm \ref{algo:finite-down-in-Amer-Pari-option-price}, the overall complexity for calculating the price of a finite-maturity down-in American Parisian option is $\mathcal{O}(n_tn^3)$ in general and $\mathcal{O}(n_tkn)$ in the diffusion case.

\begin{algorithm}[htbp]
\caption{Calculate $\widetilde{C}_{fi}(s,x)=\sup_{\tau \in \mathcal{T}_s}\mathbb{E}_{s,x} \big[ e^{-r\zeta_{\tau}}  f( Y_\tau) 1_{\{ \tau_{L, D}^- \le \tau, \zeta_{\tau} \le T \}} \big]$}
\label{algo:finite-down-in-Amer-Pari-option-price}
\begin{algorithmic}[1]
\Require $D, L, n, x, f(\cdot), T, \delta_t, r$ 
\State $\widetilde{\pmb{C}}_{fi}(T+\delta_t)\gets (\pmb{0})_{n+1},  \pmb{c}_f(T+\delta_t)\gets (\pmb{0})_{n+1}, \overline{\pmb{u}}^+(D,T+\delta_t)\gets (\pmb{0})_{m+1}$ 
\State $ \pmb{u}^-(T+\delta_t)\gets (\pmb{0})_{n+1}, \overline{\pmb{v}}(D,T+\delta_t)\gets (\pmb{0})_{m+1}, \pmb{I}_{L}^+\gets \text{diag}((1_{x\ge L})_{x\in \mathbb{S}}), \pmb{I}_{L}^-\gets \text{diag}((1_{x< L})_{x\in \mathbb{S}}) $
\For{$j\in \{0,1,2,\cdots, \frac{T^+-s}{\delta_t}-1\}$}
\State $t=T^+-(j+1)\delta_t$
\State calculate $\pmb{G}_t^Y,\overline{\pmb{G}}_t^Y$
\State $\pmb M\gets \overline{{\pmb G}}_t^Y - \overline{\pmb I}/\delta_t$
\State calculate $\pmb{M}^{-1}$
        \State  solve $\min \big( ({\pmb I} - {\pmb G}_t^Y \delta_t) {\pmb c}_f(t) - {\pmb c}_f(t+\delta_t), {\pmb c}_f(t) - {\pmb f}   \big) = {\pmb 0}$
        \State ${\pmb H}_1^+(t) \gets \big( {\pmb I}_L^- ({\pmb I}/\delta_t -  {\pmb G}^Y_t) + {\pmb I}_L^+ \big)^{-1} {\pmb I}_L^+$
        \State ${\pmb H}_2^+(D, t) \gets \exp\big( {\pmb I}_L^- ({\pmb G}^Y_t - {\pmb I/\delta_t}) D \big) {\pmb I}_L^-{\pmb H}_1^+(t)$
        \State $\pmb{H}^+(D,t)\gets \pmb{H}_1^+(t)-\pmb{H}_2^+(D,t)$
        \State ${\pmb H}^-(t) \gets \big( {\pmb I}_L^+ ({\pmb I}/\delta_t - {\pmb G}_t^Y) + {\pmb I}_L^- \big)^{-1} {\pmb I}_L^-$
        \State obtain $\overline{{\pmb H}}_1^+(t), \overline{\pmb c}_f(t)$
        \State $\overline{{\pmb u}}^+(D, t) \gets \sum_{k=0}^{j} \pmb{C}_k(D) \overline{{\pmb H}}_1^+(t)\widetilde{\pmb C}_{fi}(t+k\delta_t)
        $
        \State ${\pmb u}^-(t) \gets \big( {\pmb I}^-_L + {\pmb I}^+_L ({\pmb I} - {\pmb G}^Y_t \delta_t) \big)^{-1} {\pmb I}^+_L \big( {\pmb u}^-(t+\delta_t) + {\pmb H}^-(t) \widetilde{\pmb C}_{fi}(t + \delta_t)  \big)$
        \State $\overline{\pmb v}(D, t) \gets \sum_{k=0}^{j} (\frac{1}{\delta_t})^ke^{D\pmb{M}}\frac{D^k}{k!}\overline{\pmb c}_f(t+k\delta_t)$
        \State obtain $\pmb{u}^+(D, t), \pmb{v}(D, t) $
        \State $\widetilde{\pmb{C}}_{fi}(t)=(\pmb{I}-\pmb{I}^{-}_L\pmb{H}^{+}(D,t)\pmb{I}_L^{+}-\pmb{I}_L^{+}\pmb{H}^{-}(t)\pmb{I}_L^{-})^{-1}(\pmb{I}_L^{-}\pmb{u}^+(D,t)+\pmb{I}_L^{-}\pmb{v}(D,t)+\pmb{I}_L^{+}\pmb{u}^{-}(t))$
      \EndFor
    \State \textbf{return} $\widetilde{\pmb{C}}_{fi}(s)$
\end{algorithmic}
\end{algorithm}

\subsection{The Perpetual and Time-Homogeneous Case}
	
	When the down-in American Parisian option is perpetual and the underlying model is time-homogeneous, the pricing problem can be significantly simplified. Unlike the finite-maturity or time-inhomogeneous case, we do not approximate the exercise time $\tau$ as $\zeta_\tau$ in the discounted payoff, and hence the pricing problem under the CTMC model can be written as follows:
	\begin{align}
		C_{pi}(x) := \sup_{\tau \in \mathcal{T}}\mathbb{E}_x \big[ e^{-r\tau}  f( Y_\tau) 1_{\{ \tau \ge \tau_{L, D}^-, \tau < \infty \}} \big],
	\end{align}
	where $\mathbb{E}_x[\cdot] = \mathbb{E}[\cdot | Y_0 = x]$ and $\mathcal{T}$ is the set of stopping times taking value in $[0,\infty]$.
	
	As the generator is independent of $t$, by construction, the transition rates of $Y$ are now independent of $\zeta$. Therefore, $Y$ itself is a CTMC and the generator $\mathcal{G}^Y $ is independent of $t$ with a corresponding transition rate matrix denoted as ${\pmb G}^Y$.
	
	Conditional on the activation time $\tau_{L, D}^-$ and $Y_{\tau_{L, D}^-}$, we can decompose $C_{pi}(x)$ as follows.
	\begin{theorem}\label{pdi}
		The perpetual down-in American Parisian option price can be decomposed as
		\begin{align}
			C_{pi}(x) = \sum_{y \in \mathbb{S}} h_p(r, x; y) c_p(y),
		\end{align}
		where
		\begin{align}
			&h_p(r,x;y) = \mathbb{E}_x \Big[ e^{-r\tau_{L,D}^-} 1_{\big\{ Y_{\tau_{L,D}^-} = y\big\}} \Big],\\
			&c_p(y) = \sup_{\tau \in \mathcal{T}} \mathbb{E}_y [e^{-r\tau} f(Y_\tau)].
		\end{align}	
	\end{theorem}

    \subsubsection{Analysis of $c_p(y)$}
    
	We note that $c_p(y)$ is the price of a vanilla American option on $Y$ and it can be found by solving a system of variational inequalities. For all $y \in \mathbb{S}$,
	\begin{align}
		\begin{cases}
			r c_p(y) - {\mathcal{G}}^Y c_p(y) \ge 0,\\
			c_p(y) \ge f(y),\\
			\big( r c_p(y) - \mathcal{G}^Y c_p(y) \big) \big( c_p(y) - f(y) \big) = 0.
		\end{cases}	
	\end{align}
	Let ${\pmb c}_p = \big( c_p(y) \big)_{y \in \mathbb{S}}$. Then the system above can be written in a vectorized form:
	\begin{align}\label{eq:variational-ineq-perpetual-down-in}
		\min\big( (r{\pmb I} - {\pmb G}^Y) {\pmb c}_p, {\pmb c}_p - {\pmb f}\big) = {\pmb 0}.
	\end{align}
	The linear complementary problem can be formulated as
	\begin{align}\label{LCP:perpetual-down-in}
	(\operatorname{LCP}(A, \psi))\left\{\begin{array}{l}
	A \overline{C}+\psi \geq \mathbf{0}, \\
	\overline{C} \geq \mathbf{0}, \\
	\overline{C}^{\top}(A \overline{C}+\psi)=0,
	\end{array}\right.
	\end{align}
	with $A=r{\pmb I} - {\pmb G}^Y$, $\overline{C}={\pmb c}_p - {\pmb f}$, and $\psi=A{\pmb f}$.
	It can be solved by the Lemke's pivoting method.

    \subsubsection{Analysis of $h_p(r,x;y)$}
 
	We follow \cite{zhang2023general} to calculate $h_p(r,x;y)$. By \cite{zhang2023general} Theorem 2.2, $h_p(r,x;y)$ satisfies a linear system:
	\begin{align}
		h_p(r, x; y)& = 1_{\{ x < L \}} e^{-rD} v_p(D, x; y) + 1_{\{ x < L \}} \sum_{z \ge L}  h_p^+(r, D, x; z) h_p(r, z; y) \label{eq:hqxy}\\
		&\quad + 1_{\{x \ge L\}} \sum_{z < L}  h_p^-(r, x; z) h_p(r, z; y). \label{eq:hqxy-2}
	\end{align}	
	Here $v_p$, $h_p^+$ and $h_p^-$ are the first passage quantities of $Y$ defined as follows:
	\begin{align}
		&v_p(D, x; y) := \mathbb{E}_x\left[ 1_{\{ \tau_L^+ \ge D \}} 1_{\{ Y_D = y \}} \right],\\
		&h^+_p(r, D, x; z) := \mathbb{E}_x\left[ e^{-r \tau_L^+} 1_{\{ \tau_L^+ < D \}} 1_{\{ Y_{\tau_L^+} = z   \}}\right],\\
		&h^-_p(r, x; z) := \mathbb{E}_x\left[ e^{-r \tau_L^-} 1_{\{ Y_{\tau_L^-} = z \}} \right].
	\end{align}
	As shown in \cite{zhang2023general} Proposition 2.4, $v_p$, $h_p^+$ and $h_p^-$ admit closed-form solutions in terms of matrix-vector operations as follows:
	\begin{align}
		&{\pmb V}_p := \left( v_p(D, x; y) \right)_{x, y \in \mathbb{S}} = \exp\left( {\pmb I}_L^- {\pmb G}^Y D \right)\pmb{I}_L^-, \label{eq:v-matrix}\\
		&{\pmb U}^+_{1}(r) := (r{\pmb I}_L^- - {\pmb I}_L^- {\pmb G}^Y + {\pmb I}_L^+ )^{-1} {\pmb I}_L^+, \label{eq:u1-matrix}\\
		&{\pmb U}^+_{2}(r)  := e^{-rD}\exp\left( {\pmb I}_L^- {\pmb G}^Y  D \right) {\pmb I}_L^- {\pmb U}^+_1(r), \label{eq:u2-matrix} \\
		&{\pmb U}_p^+(r) := \left( h^+_p(r,D, x; z) \right)_{x, z \in \mathbb{S}} =  {\pmb U}^+_1(r) - {\pmb U}^+_2(r),\\
		&{\pmb U}^-_p(r) := \left( h^-_p(r, x; z) \right)_{x, z \in \mathbb{S}} = (r{\pmb I}_L^+ - {\pmb I}_L^+ {\pmb G}^Y + {\pmb I}_L^- )^{-1} {\pmb I}_L^-. \label{eq:um-matrix}
	\end{align}
	Let ${\pmb U}_p(r) = \pmb I^-_L {\pmb U}_p^+(r) + \pmb I_L^+ {\pmb U}_p^-(r)$. Then
	\begin{align}
		{\pmb H}_p(r) = \left( h_p(r, x; y) \right)_{x, y \in \mathbb{S}}= e^{-rD} ({\pmb I} - {\pmb U}_p(r))^{-1}\pmb{I}_L^-{\pmb V}_p. \label{eq:H-matrix}
	\end{align}

\subsubsection{The Algorithm and Complexity Analysis}

Combining the analysis above,  the option price can be found as,
	\begin{align}
		{\pmb C}_{pi} &= \big( C_{pi}(x) \big)_{x \in \mathbb{S}} = {\pmb H}_p(r) {\pmb c}_p = e^{-rD} ({\pmb I} - {\pmb U}_p(r))^{-1}\pmb{I}_L^-{\pmb V}_p {\pmb c}_p.
	\end{align}
We summarize the algorithm in Algorithm \ref{algo:perpetual-down-in-Amer-Pari-option-price}. The complexity analysis can be conducted similarly to the finite-maturity case except that we avoid the iteration over time. Therefore, the overall complexity of computing the price of a perpetual down-in American Parisian option is $\mathcal{O}(kn)$ for diffusion models if the approximation \eqref{eq:exp-approximation} is used and $\mathcal{O}(n^3)$ for jump models.  

\begin{algorithm}[htbp]
\caption{Calculate $C_{pi}(x) := \sup_{\tau \in \mathcal{T}}\mathbb{E}_x \big[ e^{-r\tau}  f( Y_\tau) 1_{\{ \tau \ge \tau_{L, D}^-, \tau < \infty \}} \big]$}
\label{algo:perpetual-down-in-Amer-Pari-option-price}
\begin{algorithmic}[1]
\Require $D, L, n, x, f(\cdot), r,  \pmb{G}^Y$ 
\State $\pmb{I}_{L}^+\gets \text{diag}((1_{x\ge L})_{x\in \mathbb{S}}), \pmb{I}_{L}^-\gets \text{diag}((1_{x< L})_{x\in \mathbb{S}})$
\State solve $\min\big( (r{\pmb I} - {\pmb G}^Y) {\pmb c}_p, {\pmb c}_p - {\pmb f}\big) = {\pmb 0}$
\State ${\pmb V}_p \gets  \exp\left( {\pmb I}_L^- {\pmb G}^Y D \right)\pmb{I}_L^-$
\State ${\pmb U}^+_{1}(r) \gets (r{\pmb I}_L^- - {\pmb I}_L^- {\pmb G}^Y + {\pmb I}_L^+ )^{-1} {\pmb I}_L^+$
\State ${\pmb U}^+_{2}(r)  \gets e^{-rD}\exp\left( {\pmb I}_L^- {\pmb G}^Y  D \right) {\pmb I}_L^- {\pmb U}^+_1(r)$
\State ${\pmb U}_p^+(r) \gets  {\pmb U}^+_1(r) - {\pmb U}^+_2(r)$
\State ${\pmb U}^-_p(r) \gets  (r{\pmb I}_L^+ - {\pmb I}_L^+ {\pmb G}^Y + {\pmb I}_L^- )^{-1} {\pmb I}_L^-$
\State ${\pmb U}_p(r) \gets \pmb I^-_L {\pmb U}_p^+(r) + \pmb I_L^+ {\pmb U}_p^-(r)$
\State ${\pmb C}_{pi} \gets  \exp{(-rD)} ({\pmb I} - {\pmb U}_p(r))^{-1}\pmb{I}_L^-{\pmb V}_p {\pmb c}_p$
    \State \textbf{return} $\pmb{C}_{pi}(x)$
\end{algorithmic}
\end{algorithm}

\section{Down-Out American Parisian Options under the CTMC Approximation}
        \label{sec:doctmc}
 
A down-out American Parisian option is canceled when the Parisian stopping time is triggered. In this case, the early exercise has a strong interaction with the down-out event, which makes the problem even more complicated than the down-in case. Therefore, we adopt the state augmentation approach to deal with this situation. To be more specific, we add an additional state $G_t = t - g_{L, t}^-$ to record the current duration of excursion of $Y$ below $L$ up to time $t$. The infinitesimal generator of $(\zeta_t,  G_t, Y_t)$  can be written as,
	\begin{align}
		&\lim_{\varepsilon \downarrow 0}\frac{\mathbb{E}[g(\zeta_\varepsilon, G_\varepsilon, Y_\varepsilon)|\zeta_0 = t, G_0 = d, Y_0 = x] - g(t, d, x)}{\varepsilon} \\
		&= \partial_t g(t, d, x) + \partial_d g(t, d, x) 1_{\{ x < L \}} +  \mathcal{G}_t^Y g(t, d, x),\qquad (t, d, x) \in \mathbb{S}^{\zeta, G , Y},
	\end{align}
	where the second term considers that $G_t$ would be reset to zero once $Y$ crosses $L$ from below, and 
    \begin{align}
       \mathbb{S}^{\zeta, G , Y} = \Big( \mathbb{T} \times [0, \infty) \times  \big( \mathbb{S} \cap (-\infty, L) \big) \Big) \cup \Big( \mathbb{T} \times  \{ 0 \} \times  \big( \mathbb{S} \cap [L, \infty) \big) \Big)
    \end{align}
    is the state space of the process $(\zeta_t,  G_t, Y_t)$.  Then we can approximate the joint process $(\zeta_t, G_t, Y_t)$ with a CTMC $(\zeta_t, D_t, Y_t)$. For $G_t$, we discretize $[0, \infty)$ with a uniform grid
	\begin{align}
		\mathbb{D} = \{ d_i = i\delta_d: i = 0, 1, 2, \ldots \},
	\end{align}
	where $\delta_d$ is the step size. The state space of $(\zeta_t, D_t,  Y_t)$ can be constructed as
	\begin{align}
		\mathbb{S}^{\zeta, D, Y} = \Big( \mathbb{T} \times \mathbb{D} \times  \big( \mathbb{S} \cap (-\infty, L) \big) \Big) \cup \Big( \mathbb{T} \times  \{ 0 \} \times  \big( \mathbb{S} \cap [L, \infty) \big) \Big).
	\end{align} 
	The generator of $(\zeta_t, D_t,  Y_t)$ approximates that of $(\xi_t,  G_t, X_t)$ and can be obtained by further approximate $\partial_d$ with a finite forward difference following the generator of $(\zeta_t, Y_t)$:
	\begin{align}
		\mathcal{G}^{\zeta, D, Y} g(t, d, x) &= \frac{g(t+\delta_t, d, x) - g(t, d, x)}{\delta_t} +  \widetilde{\mathcal{G}}_t^{D, Y} g(t, d, x),\ (t, d,  x) \in \mathbb{S}^{\zeta, D, Y},
	\end{align}
	where
	\begin{align}
		\widetilde{\mathcal{G}}_t^{D,Y} g(t, d, x) := \frac{g(t, d+\delta_d, x) - g(t, d, x)}{\delta_d} 1_{\{ x < L \}} + \sum_{y \in \mathbb{S}} {\pmb G}_t^Y(x,y) g(t, d1_{\{y < L\}}, y).
	\end{align}
 
\subsection{The Finite-Maturity Case}

Similar to the finite-maturity case for down-in American Parisian options, we first approximate the exercise time $\tau$ by $\zeta_\tau - \zeta_0 = \zeta_\tau$ with $\zeta_0 = 0$, and approximate $\tau_{L, D}^-$ with $\widetilde{\tau}_{L, D}^- := \inf\{s \ge 0: D_s \ge D\}$. The resulting approximate discounted payoff function is as follows:
	\begin{align}
		e^{-r\zeta_{\tau}} f(Y_\tau) 1_{\{ \tau \le \widetilde{\tau}_{L, D}^-, \zeta_\tau \le T \}}.
	\end{align}
In this way, the finite-maturity down-out American Parisian option price at time $t$ is approximated as
	\begin{align}
		C_{fo}(t, d, x) = \sup_{\tau \in \mathcal{T}} \mathbb{E}_{t, d, x} \big[ e^{-r(\zeta_{\tau}-t)}f(Y_\tau) 1_{\{ \tau \le \widetilde{\tau}_{L, D}^-, \zeta_\tau \le T \}} \big],
	\end{align}
	where $\mathbb{E}_{t, d, x}[\cdot] = \mathbb{E}[\cdot|\zeta_t = t, D_t = d, Y_t = x]$. We note $C_{fo}(t, d, x)$ is the price of a perpetual American option written on $(\zeta, D, Y)$ and hence it satisfies the system of variational inequalities:
	\begin{align}
		\begin{cases}
			\frac{C_{fo}(t+\delta_t, d, x) - C_{fo}(t, d, x)}{\delta_t} +
			\widetilde{\mathcal{G}}_t^{D,Y} C_{fo}(t, d, x) - r C_{fo}(t, d, x) \le 0,\\
			C_{fo}(t,x,d) \ge f(x) 1_{\{ t \le T, d \le D \}},\\
			\text{equality holds for either of the above.}
		\end{cases}
	\end{align}
	To construct a vectorized form of the above, we can order the joint states of $(D_t, Y_t)$ as
	\begin{align}
		\mathbb{S}^{D, Y} &= \{ (d_0, y_0), (d_0, y_1), \ldots, (d_0, y_n), (d_1, y_0), (d_1, y_1), \ldots, (d_1, L^-), \\
  &\qquad\qquad\qquad\qquad\qquad\ldots, (D^+, y_0 ), \ldots, (D^+, L^- )  \},
	\end{align}
	where $D^+ = \min\{ d \in \mathbb{D}: d > D \}$.
	Then the values of $C_{fo}(t, d, x)$ at these states are stacked together into a vector.
	\begin{align}
		{\pmb C}_{fo}(t) = \big( C_{fo}(t, d, x) 1_{\{ d \le D \}} \big)_{(d,x) \in \mathbb{S}^{D, Y}}.
	\end{align}
	Let $\widetilde{\pmb G}_t^{D,Y}$ be the matrix representing the effect of $\widetilde{\mathcal{G}}_t^{D, Y}$. 
    For $(d, x)$ such that $d< D^+$ and $x < L$, the corresponding row is as follows.
	\begin{align}
		& \widetilde{\pmb G}_t^{D,Y}(d, x; d, x) = -\frac{1}{\delta_d} + {\pmb G}_t^Y(x, x),\\
		& \widetilde{\pmb G}_t^{D,Y}(d, x; d+\delta_d, x) = \frac{1}{\delta_d}, \\
		& \widetilde{\pmb G}_t^{D,Y}(d, x; d, y) = {\pmb G}_t^Y(x, y),\ y < L,\ y \ne x,\\
		& \widetilde{\pmb G}_t^{D,Y}(d, x; 0, y) = {\pmb G}_t^Y(x, y),\ y \ge L.
	\end{align}
	For $(d, x)$ such that $d=0$ and $x \ge L$, the corresponding row is as follows.
	\begin{align}
		& \widetilde{\pmb G}_t^{D,Y}(0, x; 0, y) = {\pmb G}_t^Y(x, y).
	\end{align}
	For $(d, x)$ such that $d=D^+$ and $x < L$, the corresponding row is as follows.
	\begin{align}
		& \widetilde{\pmb G}_t^{D,Y}(D^+, x; d', y) = 0, \qquad (d', y) \in \mathbb{S}^{D, Y}.
	\end{align}
	This means the boundary at $d=D^+$ is set as absorbing, which is consistent with the down-out feature as the option is canceled once this boundary is crossed. All other elements of $\widetilde{\pmb G}_t^{D,Y}$ are set as zero.

    Let $\widetilde{\pmb I}$ be an identity matrix and $\widetilde{\pmb 0}$ be an all-one vector of the same dimension as the number of states in $\mathbb{S}^{D, Y}$. Then the variational inequalities can be written in a vectorized form:
	\begin{align}\label{eq:finite-down-out-Amer-Pari-variational-inequality}
		\min\Big( \big( (r\delta_t + 1) \widetilde{\pmb I} -  \widetilde{\pmb G}_t^{D,Y}\delta_t  \big) {\pmb C}_{fo}(t) - {\pmb C}_{fo}(t+\delta_t),  {\pmb C}_{fo}(t) - \widetilde{\pmb f} \Big) = \widetilde{\pmb 0}.
	\end{align}
	The linear complementary problem can be formulated as
	\begin{align}\label{LCP:finite-down-out-case}
	(\operatorname{LCP}(A, \psi))\left\{\begin{array}{l}
	A \overline{C}+\psi \geq \widetilde{\pmb 0}, \\
	\overline{C} \geq \widetilde{\pmb 0}, \\
	\overline{C}^{\top}(A \overline{C}+\psi)=0
	\end{array}\right.
	\end{align}
	where $A=(r\delta_t + 1) \widetilde{\pmb I} -  \widetilde{\pmb G}_t^{D,Y}\delta_t $, $\overline{C}={\pmb C}_{fo}(t) - \widetilde{\pmb f}$, and $\psi=A\widetilde{\pmb f}-{\pmb C}_{fo}(t+\delta_t)$ with $\widetilde{\pmb f} = (f(x))_{(d, x) \in \mathbb{S}^{D, Y}}$. The equations can be solved recursively from ${\pmb C}_{fo}(T^+) = \widetilde{\pmb 0}$ and for each $t = T^+-\delta_t, \ldots, 0$, the Lemke's pivoting method can be applied to solve for ${\pmb C}_{fo}(t)$; the largest number of iterations is $n_t=T^+/\delta_t$. 

\begin{remark}
    The matrix $\widetilde{\pmb G}_t^{D,Y}$ is an $(n_d n)\times (n_d n)$ matrix with $n_d = D^+/\delta_d+1$. We note that the transition rate matrix in the $D$ dimension is sparse, and hence the computational costs of calculating $\big((r\delta_t + 1) \widetilde{\pmb I} -  \widetilde{\pmb G}_t^{D,Y}\delta_t\big)$ times a vector are $\mathcal{O}(n_d n)$ and $\mathcal{O}(n_d n^2)$ for the diffusion and jump cases respectively in the LCP \eqref{LCP:finite-down-out-case}. So the overall complexity of calculating the price of a finite-maturity down-out American Parisian option is $\mathcal{O}(n_t n_d n)$ for diffusion models and $\mathcal{O}\big(n_t n_d n^2\big)$ for jump models.
\end{remark}
 	
	\subsection{The Perpetual and Time-Homogeneous Case}
	
	In the time-homogeneous case, $\widetilde{\mathcal{G}}_t^{D,Y} = \widetilde{\mathcal{G}}^{D,Y}$ and $\widetilde{\pmb G}_t^{D,Y} = \widetilde{\pmb G}^{D,Y}$ are independent of $t$. The price of a down-out perpetual American Parisian option can be approximated as:
	\begin{align}
		C_{po}(d, x) = \sup_{\tau \in \mathcal{T}}\mathbb{E}_{d,x}\big[ e^{-r\tau} f(Y_\tau) 1_{\{ \tau \le \widetilde{\tau}_{L, D}^- \}} \big], 
	\end{align}
	where $\mathbb{E}_{d, x}[\cdot] = \mathbb{E}[\cdot|D_0 = d, Y_0 = x]$. Then $C_{po}(d,x)$ satisfies a system of variational inequalities:
	\begin{align}
		\begin{cases}
			\widetilde{\mathcal{G}}^{D, Y} C_{po}(d, x) - rC_{po}(d,x) \le 0,\\
			C_{po}(d,x) \ge f(x) 1_{\{ d \le D \}},\\
			\text{equality holds for either of the above.}
		\end{cases}
	\end{align}
	Let
	\begin{align}
		{\pmb C}_{po} = \big( C_{po}(d, x) 1_{\{ d \le D \}} \big)_{(d,x) \in \mathbb{S}^{D, Y}}.
	\end{align}
	Then the variational inequalities can be written in a vectorized form:
	\begin{align}
		\min\big( (r \widetilde{\pmb I} - \widetilde{\pmb G}^{D, Y}) {\pmb C}_{po}, {\pmb C}_{po}  - \widetilde{\pmb f}  \big) = \widetilde{\pmb 0},
	\end{align}
	which can be formulated as a linear complementary problem:
	\begin{align}
	(\operatorname{LCP}(A, \psi))\left\{\begin{array}{l}
	A \overline{C}+\psi \geq \mathbf{0}, \\
	\overline{C} \geq \mathbf{0} \\
	\overline{C}^{\top}(A \overline{C}+\psi)=0
	\end{array}\right.
	\end{align}
	where $A=r \widetilde{\pmb I} - \widetilde{\pmb G}^{D, Y} $, $\overline{C}={\pmb C}_{po}- \widetilde{\pmb f}$, and $\psi=A\widetilde{\pmb f}$. It can be solved with Lemke's pivoting method.

\begin{remark}
    With the same arguments as the finite-maturity down-out case, the overall complexity of calculating the price of a perpetual down-out American Parisian option is $\mathcal{O}(n_d n)$ for diffusion models and $\mathcal{O}(n_d n^2)$ for jump models.
\end{remark}

\section{Convergence Analysis}
 \label{sec:ca}
In this section, we analyze the convergence of down-in/-out perpetual/finite-maturity American Parisian option prices as the approximating CTMC converges. To this end, we impose the following assumptions for both the down-in and -out cases, and additional assumptions of convergence for these two cases are presented in Sections \ref{sec:conv-down-in} and \ref{sec:conv-down-out} respectively.

\begin{assumption}\label{assump:payoff-lipschitz}
The payoff function $f(\cdot)$ is Lipschitz-continuous.
\end{assumption}

\begin{assumption}\label{assump:moment-quadratic-variation}
    The first moment of quadratic variations of $X$  and $Y$ as functions of $t \ge 0$,
    \[
    t\rightarrow \mathbb{E}\big[ \langle X  \rangle_t \big],\  t\rightarrow \mathbb{E}\big[\langle Y \rangle_t \big] 
    \]
    are Lipschitz-continuous with Lipschitz constants independent of $n$ and $\delta_t$.
\end{assumption}

\begin{assumption}\label{assump:payoff-growth}
There is a constant $C > 0$ independent of $n$ and $\delta_t$ such that,
\begin{align}
&\mathbb{E}_{x}\Big(\sup _{t \in[0, \infty)} e^{-r t}\big|f(X_{t})\big|\Big) \le C,\quad \lim _{t \rightarrow \infty} e^{-r t} f(X_{t})=0 \quad a.s.,\\
&\mathbb{E}_{x}\Big(\sup _{t \in[0, \infty)} e^{-r t}\big|f(Y_{t})\big|\Big) \le C,\quad \lim _{t \rightarrow \infty} e^{-r t} f(Y_{t})=0 \quad a.s..
\end{align}
\end{assumption}

\subsection{Down-In American Parisian Options}\label{sec:conv-down-in}

For the pricing of down-in American Parisian options, we reduce the problem to calculating the expectation of prices of a vanilla American option with an initial time-spatial state $(\tau_{L, D}^-, Y_{\tau_{L, D}^-})$. Therefore, we first establish the weak convergence of the joint state under $Y$ to that under $X$ following \cite{zhang2023general}. To differentiate the Parisian stopping times defined from these two processes $Y$ and $X$, we add superscripts $Y$ and $X$ to $\tau_{L, D}^-$ respectively.    

 Let $D(\mathbb{R})$ be the space of right-continuous function with left limits mapping $[0, \infty) \to \mathbb{R}$. And for any $\omega \in D(\mathbb{R})$, let $\sigma_0^+(\omega):=0$ and
		\[
		\begin{aligned}
		\sigma_{k}^{-}(\omega) & :=\inf \left\{t>\sigma_{k-1}^{+}(\omega): \omega_{t}<L\right\}, \\
		\sigma_{k}^{+}(\omega) & :=\inf \left\{t>\sigma_{k}^{-}(\omega): \omega_{t} \geq L\right\}, \quad k \geq 1 . \\
		\end{aligned}
		\] 
  We also construct the following subsets of $D(\mathbb{R})$ following \cite{zhang2023general}:
		\[
		\begin{aligned}
		V & :=\left\{\omega \in D(\mathbb{R}): \lim _{k \rightarrow \infty} \sigma_{k}^{ \pm}(\omega)=\infty\right\}, \\
		W & :=\left\{\omega \in D(\mathbb{R}): \sigma_{k}^{+}(\omega)-\sigma_{k}^{-}(\omega) \neq D \text { for all } k \geq 1\right\}, \\
		U^{+} & :=\left\{\omega \in D(\mathbb{R}): \inf \left\{t \geq u: \omega_{t} \geq L\right\}=\inf \left\{t \geq u: \omega_{t}>L\right\} \text { for all } u \geq 0\right\}, \\
		U^{-} & :=\left\{\omega \in D(\mathbb{R}): \inf \left\{t \geq u: \omega_{t} \leq L\right\}=\inf \left\{t \geq u: \omega_{t}<L\right\} \text { for all } u \geq 0\right\} .
		\end{aligned}
		\]
\begin{assumption}\label{assump:x-path-regularity}
    Suppose that the paths $X$ are right-continuous with left limits a.s., and $\mathbb{P}[X \in V]=\mathbb{P}[X \in W]=\mathbb{P}\left[X \in U^{+}\right]=\mathbb{P}\left[X \in U^{-}\right]=1$. Moreover, $\mathbb{P}[\tau_{L, D}^{X, -}=t]=0$ for any $t > 0$.
\end{assumption}
  
  \begin{proposition}[\cite{zhang2023general} Theorem 3.3]\label{prop:parisian-stopping-time-weak-convergence}
    Suppose Assumption \ref{assump:x-path-regularity} holds and $Y \Rightarrow X$ as $n \to \infty$ and $\delta_t \to 0$. Then for any continuous function $f(\cdot)$ such that for any $t \ge 0$ there exists constant $C_t$ independent of $n$ and $\delta_t$ satisfying that $\mathbb{E}[|f(Y_t)|] \le C_t$, we have,	   
	   \[
	   \mathbb{E}\left[1_{\left\{\tau_{L, D}^{Y,-} \leq t\right\}} f\left(Y_{t}\right)\right] \longrightarrow \mathbb{E}\left[1_{\left\{\tau_{L, D}^{X, -} \leq t\right\}} f\left(X_{t}\right)\right].
	   \] 
  \end{proposition}

\subsubsection{The Perpetual and Time-Homogeneous Case}

We denote by $\mathrm{F}^X$ and $\mathrm{F}^Y$ the filtrations generated by the processes $X$ and $Y$, respectively. We introduce the following intermediate quantities:
\[
c^{\hat{\delta}}(x)=\sup_{\tau \in \mathcal{T}^X}\mathbb{E}_{x}\Big[e^{-r\tau^{\hat{\delta}}}f(X_{\tau^{\hat{\delta}}})\Big],
\]
\[
c(x)=\sup_{\tau \in \mathcal{T}^X}\mathbb{E}_{x}\Big[e^{-r\tau}f(X_{\tau})\Big],
\]
where $\mathcal{T}^X$ is the set of $\mathrm{F}^X$ stopping times taking values in $[0,\infty]$ (the payoff should be understood as zero when the stopping time takes value $\infty$ and this convention also applies elsewhere in this paper), $\tau^{\hat{\delta}}$ is the stopping time taking the form $\tau^{\hat{\delta}}=\inf \left \{t\ge \tau:t\in \mathbb{T}^{\hat{\delta}}  \right \}  $ with $\mathbb{T}^{\hat{\delta}}=\left \{i\hat{\delta}: i\in \mathbb{N}  \right \} $, $c^{\hat{\delta}} (x)$ and $c(x)$ are the value functions of a perpetual Bermudan option and that of a perpetual American option under original dynamic $X$, respectively. 

We consider the value function of a perpetual Bermudan option under $Y$, denoted by $c_p^{\hat{\delta}}(x)$:
\[
c_p^{\hat{\delta}}(x)=\sup_{\tau\in \mathcal{T}^Y}\mathbb{E}_x\Big[e^{-r\tau^{\hat{\delta}}}f( Y_{\tau^{\hat{\delta}}} )  \Big],
\]
where $\mathcal{T}^Y$ is the set of $\mathrm{F}^Y$ stopping times taking values in $[0,\infty]$. Besides, we introduce some new intermediate quantities with finite maturities:
\[
\widetilde{c}^{\hat{\delta}}_f(s,x)=\sup_{\tau\in \mathcal{T}^Y_{s,T}}\mathbb{E}_{s,x}\Big[ e^{-r(\tau^{\hat{\delta}}-s)}f(Y_{\tau^{\hat{\delta}}})  \Big],
\]
\[
c^{\hat{\delta}} (s,x)=\sup _{\tau \in \mathcal{T}^X_{s,T}} \mathbb{E}_{s, x}\left[e^{-r\left(\tau^{\hat{\delta}}-s\right)} f\left(X_{\tau^{\hat{\delta}}}\right) \right],
\]
where $\mathcal{T}^Y_{s,T}$ and $\mathcal{T}^X_{s,T}$ are the sets of $\mathrm{F}^Y$ and $\mathrm{F}^X$ stopping times taking values in $[s,T] \cup \{\infty\}$, respectively.

We introduce the following lemma.
\begin{lemma}\label{lem:perpetual-American-option-intermedia}
    Suppose Assumption \ref{assump:payoff-growth} holds. Then for any $\epsilon>0$, there exists a $T_\epsilon$, such that for any $T>T_\epsilon$ and $n\in \mathbb{N}$, we have
    \begin{equation} 
    \Big|c^{\hat{\delta}}(0,x)-c^{\hat{\delta}}(x)\Big|< \epsilon,
    \end{equation}
    \begin{equation} 
    \Big|\widetilde{c}^{\hat{\delta}}_f(0,x)-c_p^{\hat{\delta}}(x)\Big|< \epsilon.
    \end{equation}
\end{lemma}
With Lemma \ref{lem:perpetual-American-option-intermedia}, we have the following result.
\begin{proposition}\label{prop:perpetual-down-in-Amer-option-convergence}
    In the perpetual and time-homogeneous case, suppose that Assumptions \ref{assump:payoff-lipschitz}, \ref{assump:moment-quadratic-variation} and \ref{assump:payoff-growth} hold. Then, as $n\rightarrow \infty$,
    \[
    c_p(x)\rightarrow c(x). 
    \]
\end{proposition}
The proofs of Lemma \ref{lem:perpetual-American-option-intermedia} and Proposition \ref{prop:perpetual-down-in-Amer-option-convergence} can be found in Appendix \ref{app:proof}. With Proposition \ref{prop:parisian-stopping-time-weak-convergence} and Proposition \ref{prop:perpetual-down-in-Amer-option-convergence}, we have the following result:
\begin{theorem}\label{thm:perpetual-down-in-Amer-Pari-option-convergence}
    In the perpetual and time-homogeneous case, suppose that Assumptions \ref{assump:payoff-lipschitz}, \ref{assump:moment-quadratic-variation}, \ref{assump:payoff-growth} and \ref{assump:x-path-regularity} hold. Then, as $n\rightarrow \infty$,
    \[
    C_{pi}(x)\rightarrow \overline{C}_{pi}(x),
    \]
    where $\overline{C}_{pi}(x)$ is the value function of the perpetual down-in American Parisian option under the original dynamic $X$: 
    \[
    \overline{C}_{pi}(x)=\sup_{\tau \in \mathcal{T}^X}\mathbb{E}_x\big[e^{-r\tau}f(X_{\tau})1_{\{ \tau\ge \tau_{L, D}^{X,-} \}}\big].
    \]
\end{theorem}
The proof of Theorem \ref{thm:perpetual-down-in-Amer-Pari-option-convergence} is provided in Appendix \ref{app:proof}.
  
\subsubsection{The Finite-Maturity Case}
Let $\mathrm{F}^{\zeta, Y}$ be the filtration generated by the process $(\zeta, Y)$. We introduce the value function of a finite-maturity American option under the original dynamic $X$:
\[
c(s,x)=\sup _{\tau \in \mathcal{T}^X_{s,T}} \mathbb{E}_{s, x}\left[e^{-r\left(\tau-s\right)} f\left(X_{\tau}\right) \right],
\]
where $\mathcal{T}^X_{s,T}$ is the set of $\mathrm{F}^X$ stopping times taking values in $[s,T] \cup \{\infty
\}$. We also consider the value function of a finite-maturity Bermudan option under $Y$: 
\[
c_f^{\hat{\delta}} (s,y)=\sup_{\tau \in \mathcal{T}^{\zeta,Y}_{s}} \mathbb{E}_{s,y} \Big[ e^{-r(\zeta_{\tau^{\hat{\delta}}}- s)} f(Y_{\tau^{\hat{\delta}}})1_{\{ \zeta_{\tau}\le T \}} \Big],
\]
where $\mathcal{T}^{\zeta,Y}_{s}$ is the set of $\mathrm{F}^{\zeta, Y}$ stopping times taking values in $[s,\infty]$.

These intermediate quantities are used to prove the following result.  
\begin{proposition}\label{prop:finite-American-option-converge}
    Suppose Assumptions \ref{assump:payoff-lipschitz}, \ref{assump:moment-quadratic-variation} and \ref{assump:payoff-growth} hold. Then as $\delta_t \rightarrow 0$ and $n\rightarrow \infty$,
    \[
    c_f(s,x)\rightarrow c(s,x).
    \]
\end{proposition}
The proof of Proposition \ref{prop:finite-American-option-converge} is provided in Appendix \ref{app:proof}. With Propositions \ref{prop:parisian-stopping-time-weak-convergence} and \ref{prop:finite-American-option-converge}, we have the following result.
\begin{theorem}\label{thm:finite-down-in-Amer-Pari-option-converge}
    Suppose Assumptions \ref{assump:payoff-lipschitz}, \ref{assump:moment-quadratic-variation}, \ref{assump:payoff-growth} and \ref{assump:x-path-regularity} hold. Then as $\delta_t \rightarrow 0$ and $n\rightarrow \infty$,
    \[
    C_{fi}(t,x)\rightarrow \overline{C}_{fi}(t,x),
    \]
    where $\overline{C}_{fi}(t,x)$ is the value function of the finite maturity down-in American Parisian option under the original dynamic $X$ with maturity $T$:
    \[
    \overline{C}_{fi}(t,x)=\sup_{\tau \in \mathcal{T}_{t,T}^X}\mathbb{E}_{t,x} \big[ e^{-r(\tau - t)}  f( X_{\tau}) 1_{\{\tau\ge \tau_{L, D}^{X,-} \}} \big].
    \]
\end{theorem}
The proof of Theorem \ref{thm:finite-down-in-Amer-Pari-option-converge} is provided in Appendix \ref{app:proof}.

\subsection{Down-Out American Parisian Options}\label{sec:conv-down-out}

In the down-out case, we need to analyze the weak convergence of $\widetilde{\tau}_{L, D}^{Y,-}$, which is the Parisian stopping time under the CTMC process $(\zeta, D, Y)$. $\widetilde{\tau}_{L, D}^{X,-}$ represents the Parisian stopping time under $(\xi, G, X)$. 

\begin{assumption}\label{assump:g-regularity}
    Suppose that $(\zeta, D, Y)$ weakly converges to $(\xi, G, X)$ as $\delta_t\rightarrow 0$, $\delta_d\rightarrow 0$ and $n\rightarrow \infty$, $\inf\{ t \ge 0: G_t > D \} = \inf\{ t \ge 0: G_t \ge D \}$ a.s., and $\mathbb{P}[\widetilde{\tau}^{X,-}_{L, D} = t] = 0$ for any $t > 0$. 
\end{assumption}

\begin{proposition}\label{prop:down-out-parisian-time-weak-convergence}
    Suppose Assumption \ref{assump:g-regularity} holds. Then for any continuous function $f(\cdot)$ such that for any $t \ge 0$ there exists constant $C_t$ independent of $\delta_t, \delta_d$ and $n$ satisfying that $\mathbb{E}[|f(Y_t)|] \le C_t$, we have,	   
	   \[
	   \mathbb{E}\left[1_{\left\{\widetilde{\tau}_{L, D}^{Y,-} \leq t\right\}} f\left(Y_{t}\right)\right] \longrightarrow \mathbb{E}\left[1_{\left\{\widetilde{\tau}_{L, D}^{X,-} \leq t\right\}} f\left(X_{t}\right)\right].
	   \] 
\end{proposition}

The proof of Proposition \ref{prop:down-out-parisian-time-weak-convergence} is similar to the proof of Theorem 3.3 in \cite{zhang2023general}, so we neglect the proof here.

\subsubsection{The Finite-Maturity Case}

Let $\widetilde{\mathrm{F}}^X$ be the filtration generated by the process $(G, X)$. We introduce the following quantity under the original dynamic $X$:
\[
\overline{C}^{\hat{\delta}}_{fo} (s,d,x)=\sup _{\tau \in \widetilde{\mathcal{T}}^X_{s,T}} \mathbb{E}_{s,d, x}\left[e^{-r\left(\tau^{\hat{\delta}}-s\right)} f\left(X_{\tau^{\hat{\delta}}}\right) 
 1_{\{ \tau\le \widetilde{\tau}_{L, D}^{X, -} \}} \right],
\]
where $\widetilde{\mathcal{T}}^X_{s,T}$ is the set of $\widetilde{\mathrm{F}}^X$ stopping times taking value in $[s,T] \cup \{\infty\}$, and $\tau^{\hat{\delta}}$ is the stopping time taking the form $\tau^{\hat{\delta}}=\inf \left \{t\ge \tau:t\in \mathbb{T}^{\hat{\delta}}  \right \}$ with $\mathbb{T}^{\hat{\delta}}=\{i\hat{\delta}:i\in\mathbb{N}\}$ for some stopping time $\tau \in \widetilde{\mathcal{T}}^X_{s,T}$.

Let $\widetilde{\mathrm{F}}^{\zeta,Y}$ be the filtration generated by the process $(\zeta, D, Y)$. We consider the following intermediate quantity under $(\zeta, D, Y)$: 
\[
C_{fo}^{\hat{\delta}} (s,d,x)=\sup_{\tau \in \widetilde{\mathcal{T}}^{\zeta, Y}_{s}} \mathbb{E}_{s,d,x} \big[ e^{-r(\zeta_{\tau^{\hat{\delta}}} - s)} f(Y_{\tau^{\hat{\delta}}})1_{\{ \zeta_{\tau}\le T \}}1_{\{ \tau\le \widetilde{\tau}_{L, D}^{Y, -} \}} \big],
\]
where $\widetilde{\mathcal{T}}^{\zeta, Y}_{s}$ is the set of $\widetilde{\mathrm{F}}^{\zeta,Y}$ stopping times taking values in $[s,\infty]$. With Proposition \ref{prop:down-out-parisian-time-weak-convergence}, we have the following result.

\begin{theorem}\label{thm:finite-down-out-Amer-Pari-option-converge}
    Suppose Assumptions \ref{assump:payoff-lipschitz}, \ref{assump:moment-quadratic-variation}, \ref{assump:payoff-growth} and \ref{assump:g-regularity} hold. Then as $n\rightarrow \infty$, $\delta_t\rightarrow 0$ and $\delta_d\rightarrow 0$,
    \[
    C_{fo}(s,d,x)\rightarrow \overline{C}_{fo}(s,d,x),
    \]
    where $\overline{C}_{fo}(s,d,x)$ is the value function of the finite-maturity down-out American Parisian option under the original joint process $(\xi,G,X)$:
    \[
\overline{C}_{fo}(s,d,x)=\sup _{\tau \in \widetilde{\mathcal{T}}^X_{s,T}} \mathbb{E}_{s, d, x}\left[e^{-r\left(\tau-s\right)} f\left(X_{\tau}\right) 1_{\{ \tau\le \widetilde{\tau}_{L, D}^{X,-} \}} \right].
\]
\end{theorem}

The proof of Theorem \ref{thm:finite-down-out-Amer-Pari-option-converge} is provided in Appendix \ref{app:proof}.

\subsubsection{The Perpetual and Time-Homogeneous Case}
We introduce the following intermediate quantity under the process $(G,X)$:
\[
\overline{C}_{po}^{\hat{\delta}}(d,x)=\sup _{\tau \in \widetilde{\mathcal{T}}^X}\mathbb{E}_{d, x}\big[e^{-r\tau^{\hat{\delta}}} f(X_{\tau^{\hat{\delta}}})1_{\{ \tau\le \widetilde{\tau}_{L, D}^{X, -} \}} \big],
\]
where $\widetilde{\mathcal{T}}^X$ is the set of $\widetilde{\mathrm{F}}^X$ stopping times taking values in $[0,\infty]$, and $\tau^{\hat{\delta}}$ is the stopping time taking the form $\tau^{\hat{\delta}}=\inf \left \{s\ge \tau:s\in \mathbb{T}^{\hat{\delta}}  \right \}  $ with $\mathbb{T}^{\hat{\delta}}=\left \{i{\hat{\delta}}: i\in \mathbb{N}  \right \} $ for some stopping time $\tau \in \widetilde{\mathcal{T}}^X$. Let $\widetilde{\mathrm{F}}^Y$ be the filtration generated by the process $(D,Y)$. We also consider the following intermediate quantities under the process $(D,Y)$: 
\[
C_{po}^{\hat{\delta}}(d,x)=\sup_{\tau \in \widetilde{\mathcal{T}}^Y} \mathbb{E}_{d,x} \big[ e^{-r\tau^{\hat{\delta}}} f(Y_{\tau^{\hat{\delta}}})1_{\{ \tau\le \widetilde{\tau}_{L, D}^{Y, -} \}} \big],
\]
\[
\widetilde{C}^{\hat{\delta}}_{fo}(s,d,x)=\sup_{\tau\in \widetilde{\mathcal{T}}^Y_{s,T}}\mathbb{E}_{s,d,x}\big[e^{-r(\tau^{\hat{\delta}}-s)}f(Y_{\tau^{\hat{\delta}}})1_{\{\tau\le \widetilde{\tau}^{Y, -}_{L,D}\}}  \big],
\]
where $\widetilde{\mathcal{T}}^Y$ and $\widetilde{\mathcal{T}}^Y_{s,T}$ are the sets of $\widetilde{\mathrm{F}}^Y$ stopping times taking values in $[0,\infty]$ and $[s,T] \cup \{ \infty \}$ respectively.

We first establish the following lemma.
\begin{lemma}\label{lem:perpetual-down-out-American-option-intermedia}
    Suppose Assumption \ref{assump:payoff-growth} holds. Then for any $\epsilon>0$, there exists a $T_\epsilon$, such that for any $T>T_\epsilon$, $n$, $\delta_d$ and $\hat{\delta}$, we have
    \[
    \big| \overline{C}_{fo}(0,d,x)-\overline{C}^{\hat{\delta}}_{po}(d,x)\big|<\epsilon,
    \]
    \[
    \big| \widetilde{C}^{\hat{\delta}}_{fo}(0,d,x)-C^{\hat{\delta}}_{po}(d,x)\big|<\epsilon.
    \]
\end{lemma}

With Lemma \ref{lem:perpetual-down-out-American-option-intermedia}, we have the following result.
\begin{theorem}\label{thm:perpetual-down-out-Amer-Pari-option-converge}
    In the perpetual and time-homogeneous case, suppose Assumptions \ref{assump:payoff-lipschitz}, \ref{assump:moment-quadratic-variation}, \ref{assump:payoff-growth} and \ref{assump:g-regularity} hold. Then as $\delta_d\rightarrow 0$ and $n\rightarrow \infty$, 
    \[
    C_{po}(d, x)\rightarrow \overline{C}_{po}(d, x)
    \]
    where $\overline{C}_{po}(d, x)$ is the value function of the perpetual down-out American Parisian option under the original process $(G, X)$:
    \[
\overline{C}_{po}(d,x)=\sup _{\tau \in \widetilde{\mathcal{T}}^X} \mathbb{E}_{d, x}\big[e^{-r\tau} f(X_{\tau})1_{\{ \tau\le \widetilde{\tau}_{L, D}^{X, -} \}} \big].
\]
\end{theorem}
The proofs of Lemma \ref{lem:perpetual-down-out-American-option-intermedia} and Theorem \ref{thm:perpetual-down-out-Amer-Pari-option-converge} are provided in Appendix \ref{app:proof}.

\section{Numerical Experiments}
 \label{sec:ne}

	We consider the pricing of perpetual/finite-maturity American Parisian down-in/-out call options under the following models to assess the efficiency and accuracy of our algorithms.
	
	\begin{itemize}
		\item The Black-Scholes (BS) model: $dX_t = (r_f - d) X_t dt + \sigma X_t dW_t,\ t \ge 0$. The parameters considered are shown in Table \ref{tab:BS-parameter-numerical-experiment}.
		
		\item The Kou's double-exponential jump-diffusion model (\cite{kou2004option}):
		\begin{equation}
			\frac{dX_t}{X_{t-}} = (r_f - d - \lambda\zeta) dt + \sigma dW_t + d\left( \sum_{i = 1}^{N_t} (V_i - 1) \right),
		\end{equation}
		where $N_t$ is a Poisson process with intensity $\lambda$, $\{V_i, i\ge1\}$ is a sequence of i.i.d. random variables with the density of $\ln V_i$ given by $f_{\ln V_i}(y) = p^+ \eta^+ e^{-\eta^+ y} 1_{\{y \ge 0\}} + p^- \eta^- e^{\eta^- y} 1_{\{ y < 0\}}$, and
		$\zeta = \mathbb{E}[V_i] - 1 = \frac{p^+\eta^+}{\eta^+ - 1} + \frac{p^-\eta^-}{\eta^-+1} - 1$. We set $r_f=0.05$, $d=0$, $\sigma = 0.30$, $\lambda = 3.0$, $\eta^+ = \eta^- = 10$, and $p^+ = p^- = 0.5$.
		The Kou model is a popular jump-diffusion model with finite jump activity.
		
		\item The Variance Gamma (VG) model  (\cite{madan1998variance}): 
		\begin{equation}
			X_t = X_0 \exp\left( (r_f - d)t + Z_t + \omega t \right),
		\end{equation}
		where $\omega = \ln(1 - \theta\nu - \sigma^2\nu/2)/\nu$, $Z_t = \theta \gamma_t(1; \nu) + \sigma W_{\gamma_t(1; \nu)}$ and $\gamma_t(1; \nu)$ is a gamma process with variance rate $\nu$ and $\mathbb{E}[\gamma_1(1; \nu)]=1$. We set $r_f=0.05$, $d=0$,  $\sigma = 0.1213$, $\nu = 0.1686$, and $\theta = -0.1436$. 
		The VG model is a popular pure-jump model with infinite jump activity. 
	\end{itemize}

    To remove the oscillations in convergence and enable the application of Richardson extrapolation to accelerate convergence, we adopt the piecewise uniform (PU) grid construction of \cite{zhang2023general} to place $L$ on the grid and strike $K$ exactly in the midway between two neighbouring grid points: suppose that $L<K$, we let
	\begin{align}\label{eq:grid-construction}
	\begin{aligned}
	\mathbb{S}_{P U}= & \{y_0+ih_{1}: 0 \leq i \leq n_{1}\} \cup \{L+(1+i) h_{2}:0 \leq i \leq n_{2}-2 \} \\
    	& \cup \{K+h_2+ih_3: 0 \leq i \leq n_{3}\},
	\end{aligned}
	\end{align}
	where $h_1=(L-y_0)/n_1$, $h_2=(K-L)/n_2$, $h_3=(y_n-K-h_2)/n_3$ and $n_1+n_2+n_3=n$. The construction is similar if $K<L$. In the numerical experiments, we fix small values of $\delta_t$ and $\delta_d$, and mainly explore the convergence against $n$. For the perpetual down-in American Parisian call options under the Black-Scholes model, we use the closed-form pricing formula derived by \cite{chesney2006american} as the benchmark to calculate the pricing errors of CTMC approximation; while in other cases, we obtain the benchmarks from CTMC approximation with very fine grids, which are accurate to the fourth decimal place and match the results from Monte Carlo simulation with a large number of replications or the numerical results from existing literature.

 \begin{table}[htbp]
\begin{tabular}{ccccccccccc}
\hline
American Parisian option & \multicolumn{10}{c}{Parameter}                                      \\ \cline{2-11} 
                         & $S_0$ & $K$ & $L$ & $D$            & $\sigma$ & $r_f$ & $d$ & $\delta_t$ & $\delta_d$ & $T$ \\ \hline
Perpetual down-in        & 90    & 95  & 90  & $1/12$ & 0.3      & 0.1   & 0.05 &  &  &  \\
Perpetual down-out       & 90    & 95  & 90  & $1/12$ & 0.3      & 0.1   & 0.05 &  &1/120 &   \\
Finite-maturity down-in  & 90    & 95  & 90  & $1/12$ & 0.3      & 0.05  & 0 &1/60&   & 1   \\
Finite-maturity down-out & 105   & 100 & 95  & $1/15$ & 0.4      & 0.06  & 0.1 &1/60 &1/150  & 1   \\ \hline
\end{tabular}
\centering
\caption{Parameters for the numerical experiments under the Black-Scholes model.}
\label{tab:BS-parameter-numerical-experiment}
\end{table}

 Table \ref{tab:error-BS-model} and Figure \ref{fig:error-BS-model} show the convergence behaviour of perpetual and finite-maturity American Parisian down-in and -out call options under the BS model. The model and contract parameters are summarized in Table \ref{tab:BS-parameter-numerical-experiment}. We observe convergence as the number of grid points increases and the Richardson extrapolation results in a substantial error reduction. Overall, the down-in cases take much less time than the down-out ones and we achieve a satisfactory level of relative error in less than $1$ minute for all the four types of options under the BS model.

\begin{table}[htbp]
\resizebox{\linewidth}{!}{
\begin{tabular}{cccccccccc}
\hline
American Parisian option type                  & Grid & Benchmark & CTMC & Abs.Err. & Rel.Err. & Time/s & Extra. & Abs.Err. & Rel.Err. \\ \hline
\multirow{5}{*}{Perpetual down-in}        &      129  &26.3239 &25.9747  & 0.3492  & 1.33$\%$&0.01 & &   &     \\
				&161  &26.3239 &26.0946  & 0.2293 &0.87$\%$& 0.02 &26.3096 &0.0142 &0.05$\%$ \\
				&193  &26.3239 &26.1658  & 0.1581  & 0.60$\%$&0.02 &26.3287 &0.0048 &0.02$\%$         \\
				&225 &26.3239 &26.2087  & 0.1152 & 0.43$\%$&0.04 &26.3282 &0.0043&0.02$\%$      \\ 
				&257 &26.3239 &26.2346  & 0.0893 & 0.34$\%$&0.08 &26.3196 &0.0043&0.02$\%$          \\ \hline
\multirow{5}{*}{Perpetual down-out}       &   661  &10.3882 &10.4574  & 0.0692  & 0.67$\%$&0.30 & &   &     \\
    			&793  &10.3882 &10.4341  & 0.0459 &0.44$\%$& 0.48 &10.3809 &0.0072 &0.07$\%$ \\
    			&925  &10.3882 &10.4217  & 0.0335  & 0.32$\%$&0.60 &10.3872 &0.0009 &0.01$\%$         \\
    			&1057 &10.3882 &10.4137  & 0.0255 & 0.24$\%$&0.93 &10.3875 &0.0007&0.00$\%$      \\ 
    			&1189 &10.3882 &10.4083  & 0.0201 &0.19$\%$&1.24 &10.3879 &0.0002&0.00$\%$         \\ \hline
\multirow{5}{*}{Finite-maturity down-in}  &   177  &3.3483 &3.3169  & 0.0314  & 0.93$\%$&5.62 & &   &     \\
				&193  &3.3483 &3.3230  & 0.0253 &0.75$\%$& 6.24 &3.3552 &0.0069 &0.20$\%$ \\
				&209  &3.3483 &3.3275  & 0.0208  & 0.62$\%$&8.83 &3.3535 &0.0052 &0.15$\%$         \\
				&225 &3.3483 &3.3309  & 0.0174 & 0.52$\%$&11.38 &3.3522 &0.0039&0.11$\%$      \\ 
				&241 &3.3483 &3.3333 & 0.0150 &0.44$\%$&13.04 &3.3495 &0.0012&0.03$\%$       \\ \hline
\multirow{5}{*}{Finite-maturity down-out} &   265  &13.5126 &13.6015  & 0.0889  & 0.66$\%$&1.41 & &   &     \\
    			&397  &13.5126 &13.5501  & 0.0375 &0.28$\%$& 3.41 &13.5088 &0.0038 &0.03$\%$ \\
    			&529  &13.5126 &13.5332  & 0.0206  & 0.15$\%$&7.18 &13.5114 &0.0012 &0.00$\%$         \\
    			&661 &13.5126 &13.5256  & 0.0130 & 0.09$\%$&14.85 &13.5121 &0.0005&0.00$\%$      \\ 
    			&793 &13.5126 &13.5216 & 0.0090 &0.06$\%$&24.04 &13.5125 &0.0001&0.00$\%$       \\ \hline
\end{tabular}}
\caption{The absolute and relative errors of CTMC approximation for American Parisian option pricing under the BS model.}
\label{tab:error-BS-model}
\centering
\end{table}

\begin{figure}[htbp]
	\small
	\centering
	\subfigure{
		\begin{minipage}[t]{0.5\textwidth}
			\centering
			\includegraphics[width=\linewidth]{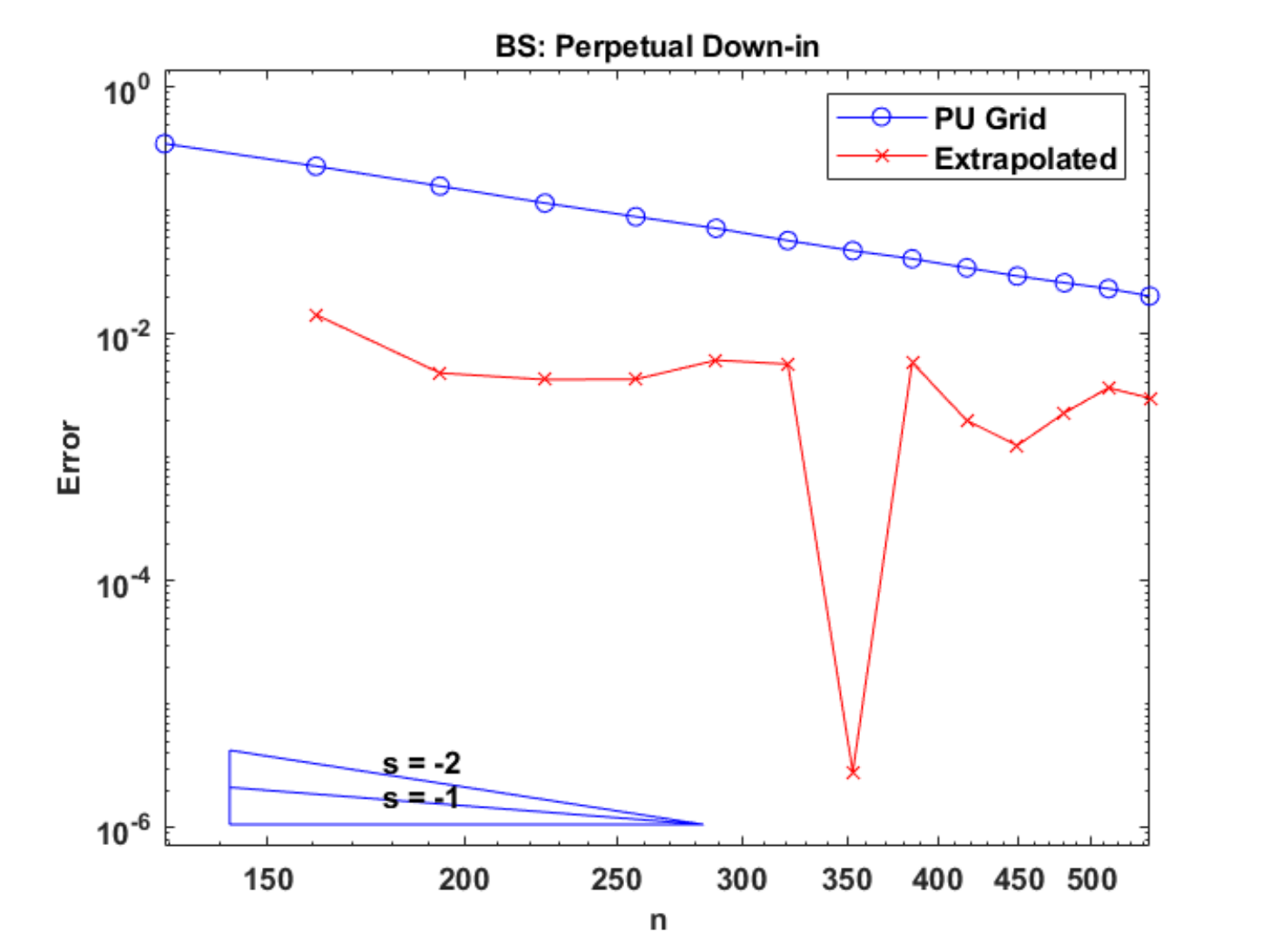}
		\end{minipage}%
	}%
	\subfigure{
		\begin{minipage}[t]{0.5\textwidth}
			\centering
			\includegraphics[width=\linewidth]{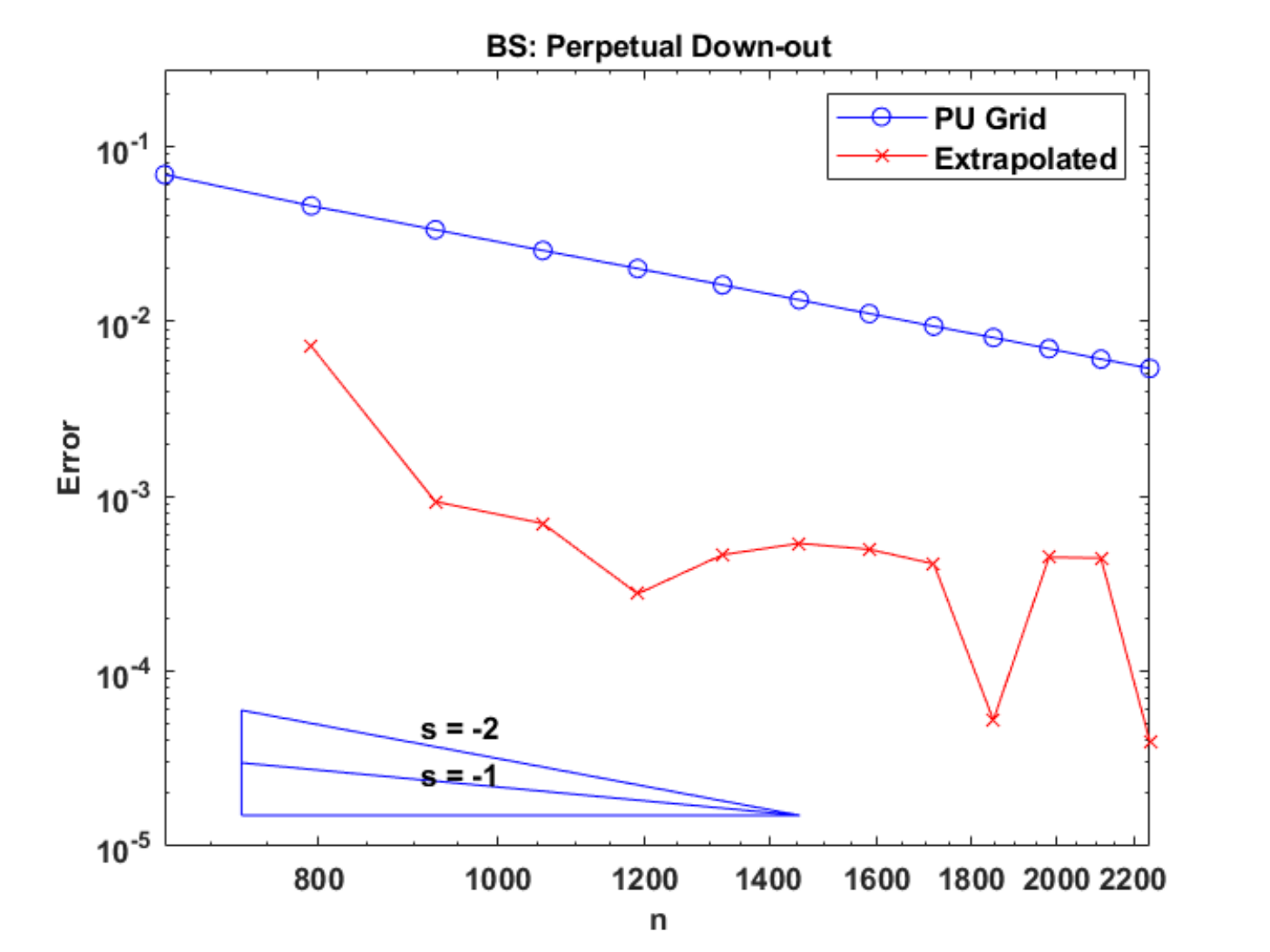}
		\end{minipage}%
	}%
	
	\subfigure{
		\begin{minipage}[t]{0.5\textwidth}
			\centering
			\includegraphics[width=\linewidth]{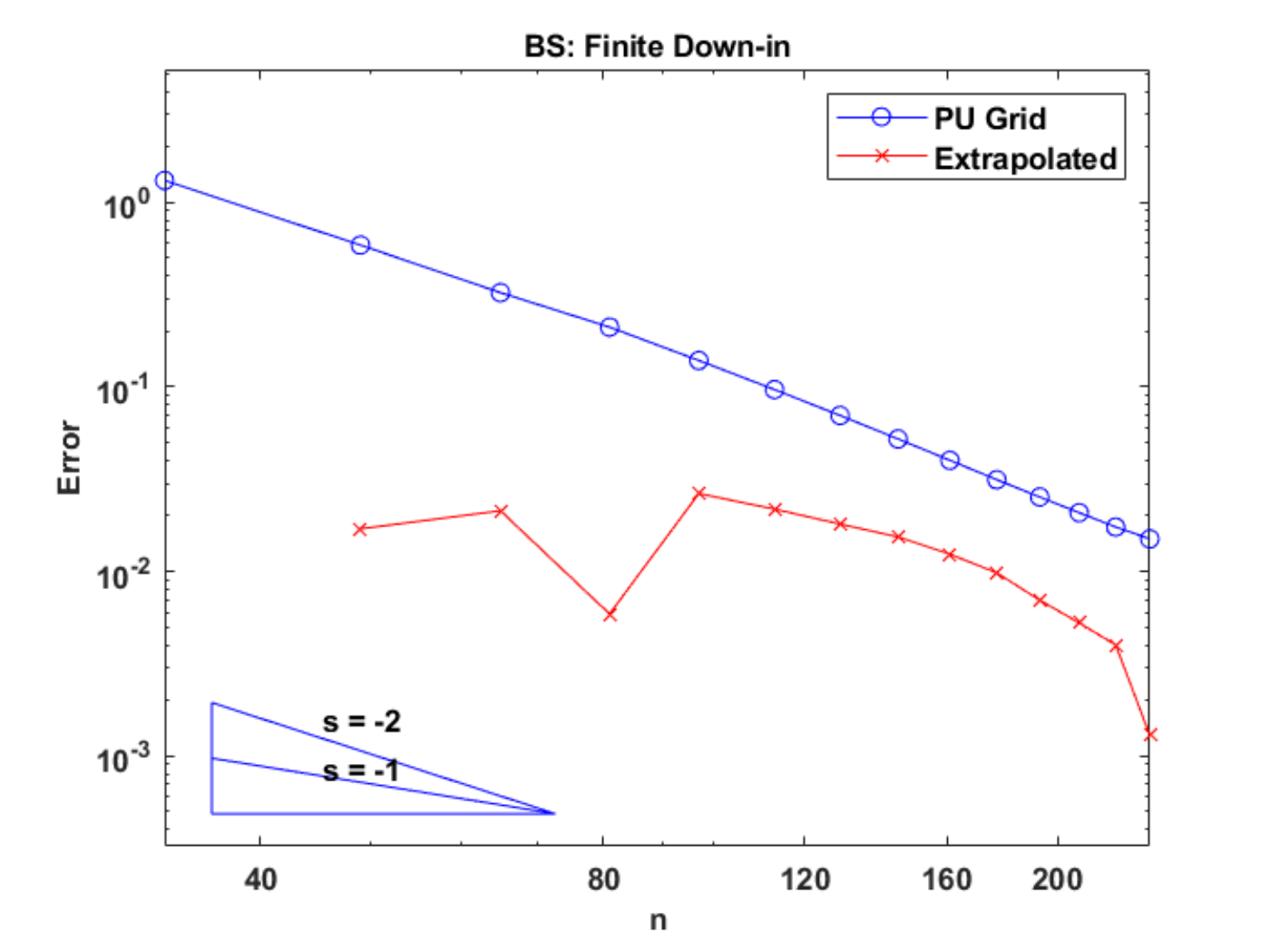}
		\end{minipage}%
	}%
	\subfigure{
		\begin{minipage}[t]{0.5\textwidth}
			\centering
			\includegraphics[width=\linewidth]{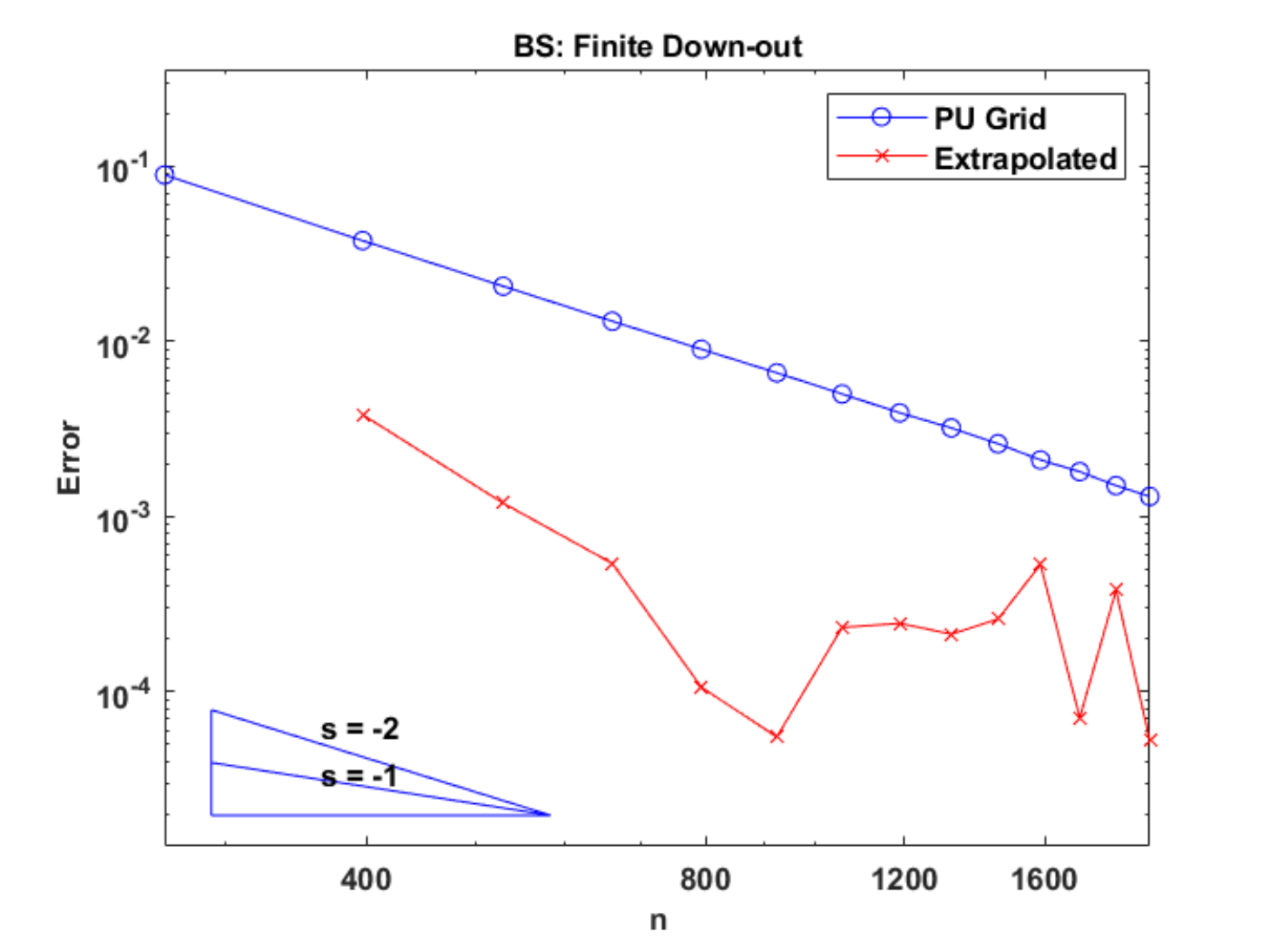}
		\end{minipage}%
	}%
	\centering
	\caption{Convergence of the CTMC approximation for perpetual/finite-maturity down-in/-out American Parisian options under the BS model.}
	\label{fig:error-BS-model}
\end{figure}

   For the Kou's and VG models, we set $S_0=90$, $K=95$, $L=90$ and $D=1/12$ for the perpetual cases, and $S_0=90$, $K=95$, $L=90$, $D=1/12$ and $T=1$ for the finite maturity cases in our numerical investigations. We set $\delta_t=1/60$ for the finite maturity cases and $\delta_d=1/120$ for the down-out cases. Tables \ref{tab:error-Kou-model}, \ref{tab:error-VG-model} and Figures \ref{fig:error-Kou-model}, \ref{fig:error-VG-model} present the results, from which we see that our approach still performs well under these models, and the Richardson extrapolation is still effective in reducing errors.

\begin{table}[htbp]
\resizebox{\linewidth}{!}{
\begin{tabular}{cccccccccc}
\hline
American Parisian option type                 & Grid & Benchmark & CTMC & Abs.Err. & Rel.Err. & Time/s & Extra. & Abs.Err. & Rel.Err. \\ \hline
\multirow{5}{*}{Perpetual down-in}        &   97  &65.0695 &64.7315  & 0.3380  & 0.52$\%$&0.04 & &   &     \\
   			&129  &65.0695 &64.8809  & 0.1886 &0.29$\%$& 0.05 &65.0752 &0.0057 &0.01$\%$ \\
   			&161  &65.0695 &64.9492  & 0.1203  & 0.18$\%$&0.06 &65.0716 &0.0021 &0.00$\%$         \\
   			&193 &65.0695 &64.9862  & 0.0833 & 0.13$\%$&0.08 &65.0708 &0.0013&0.00$\%$      \\ 
   			&225 &65.0695 &65.0085 & 0.0610 &0.09$\%$&0.10 &65.0706 &0.0011&0.00$\%$      \\ \hline
\multirow{5}{*}{Perpetual down-out}       &    397  &15.3456 &15.6182  & 0.2726  & 1.78$\%$&0.30 & &   &     \\
			 &529  &15.3456 &15.4965  & 0.1509 &0.98$\%$& 0.63 &15.3395 &0.0060 &0.04$\%$ \\
			 &661  &15.3456 &15.4415  & 0.0959  & 0.62$\%$&1.39 &15.3435 &0.0021 &0.01$\%$         \\
			 &793 &15.3456 &15.4119  & 0.0663& 0.43$\%$&2.49 &15.3445 &0.0011&0.01$\%$      \\ 
			 &925 &15.3456 &15.3941 & 0.0485&0.32$\%$&3.79 &15.3447 &0.0008&0.00$\%$       \\ \hline
\multirow{5}{*}{Finite-maturity down-in}  &    161  &4.7907 &4.7502  & 0.0405  & 0.84$\%$&4.96 & &   &     \\
			&177  &4.7907 &4.7583  & 0.0324 &0.67$\%$&5.36 &4.7971 &0.0064 &0.13$\%$ \\
			&193 &4.7907 &4.7642  & 0.0265  & 0.55$\%$&5.91 &4.7954 &0.0047 &0.09$\%$         \\
			&209 &4.7907 &4.7685  & 0.0222& 0.46$\%$&7.38 &4.7934 &0.0027&0.05$\%$      \\ 
			&225 &4.7907 &4.7716 & 0.0191&0.40$\%$&8.34 &4.7911 &0.0004&0.01$\%$      \\ \hline
\multirow{5}{*}{Finite-maturity down-out} &   529  &9.0537 &9.1261  & 0.0724 &0.79$\%$& 8.62 & & & \\
			&595  &9.0537 &9.1118  & 0.0581  & 0.64$\%$&10.50 &9.0578 &0.0041 &0.04$\%$         \\
			&661 &9.0537 &9.1013  & 0.0476& 0.52$\%$&13.85 &9.0564 &0.0027&0.03$\%$      \\ 
			&727 &9.0537 &9.0934 & 0.0397&0.43$\%$&17.65 &9.0557 &0.0020&0.02$\%$\\
			&793  &9.0537 &9.0873  & 0.0336  & 0.37$\%$&28.02 & 9.0551&0.0014   &  0.01$\%$       \\ \hline
\end{tabular}}
\centering
\caption{The absolute and relative errors of CTMC approximation for American Parisian option pricing under the Kou's model.}
\label{tab:error-Kou-model}
\end{table}

\begin{figure}[htbp]
	\small
	\centering
	\subfigure{
		\begin{minipage}[t]{0.5\textwidth}
			\centering
			\includegraphics[width=\linewidth]{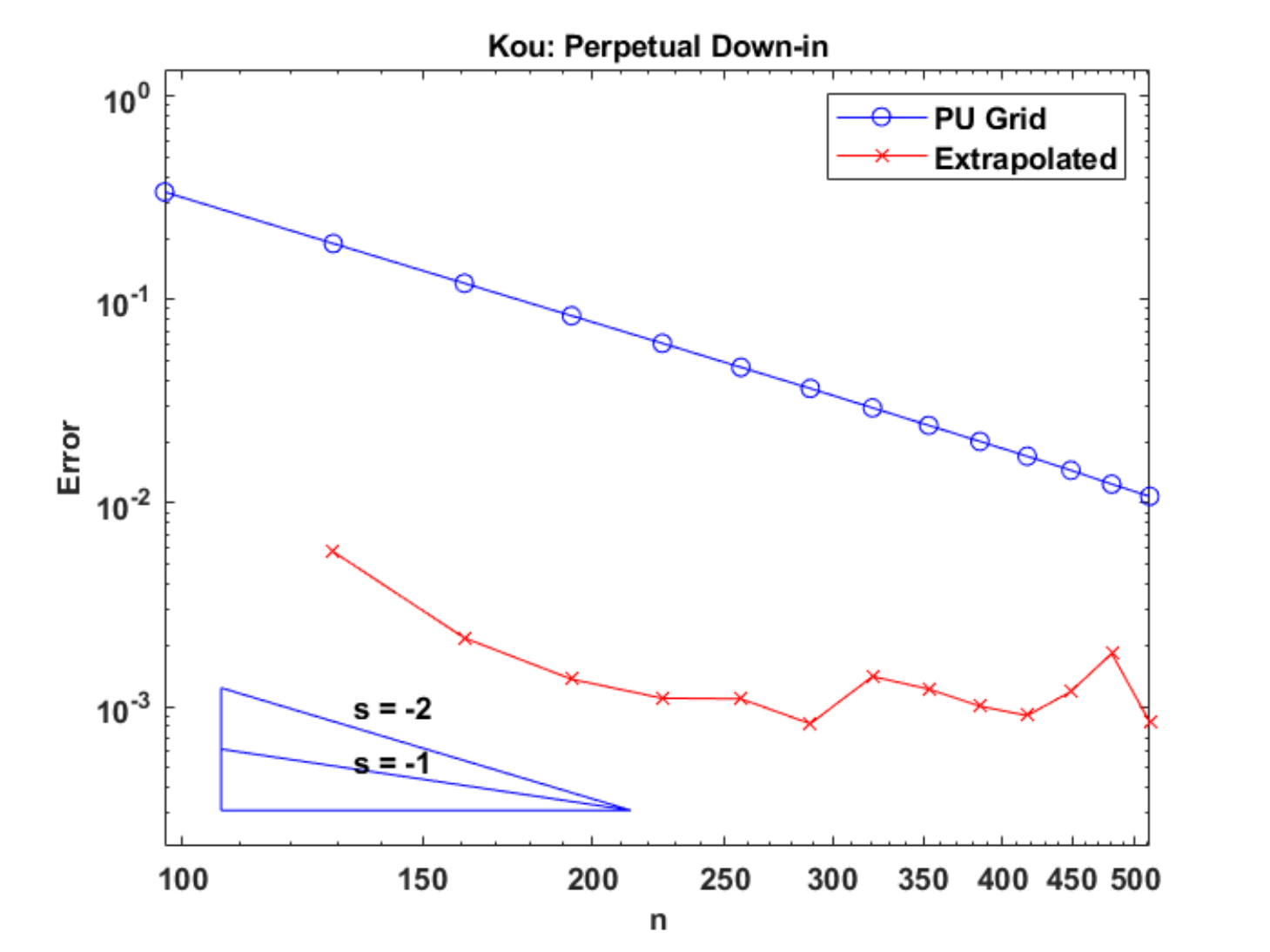}
		\end{minipage}%
	}%
	\subfigure{
		\begin{minipage}[t]{0.5\textwidth}
			\centering
			\includegraphics[width=\linewidth]{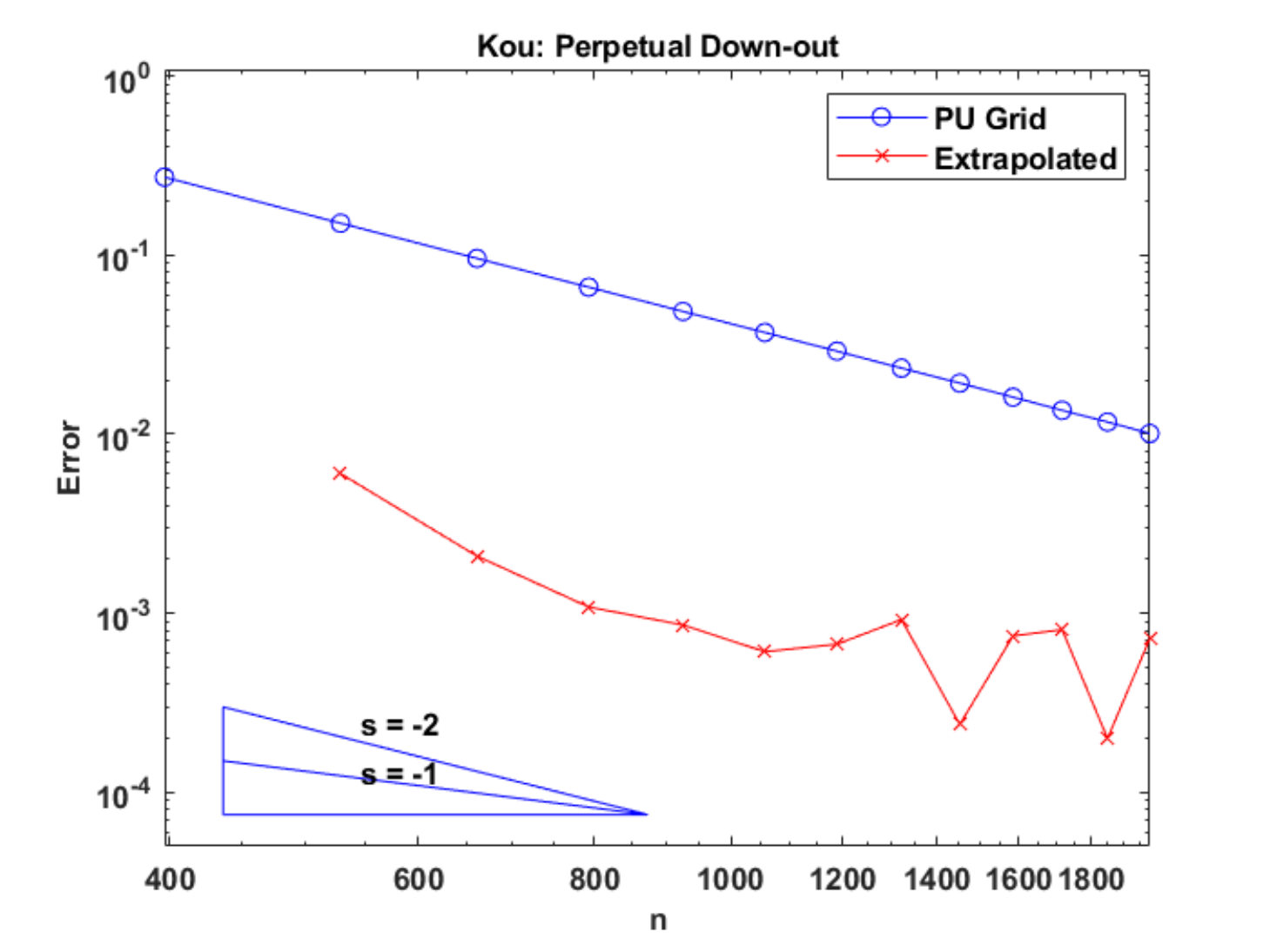}
		\end{minipage}%
	}%
	
	\subfigure{
		\begin{minipage}[t]{0.5\textwidth}
			\centering
			\includegraphics[width=\linewidth]{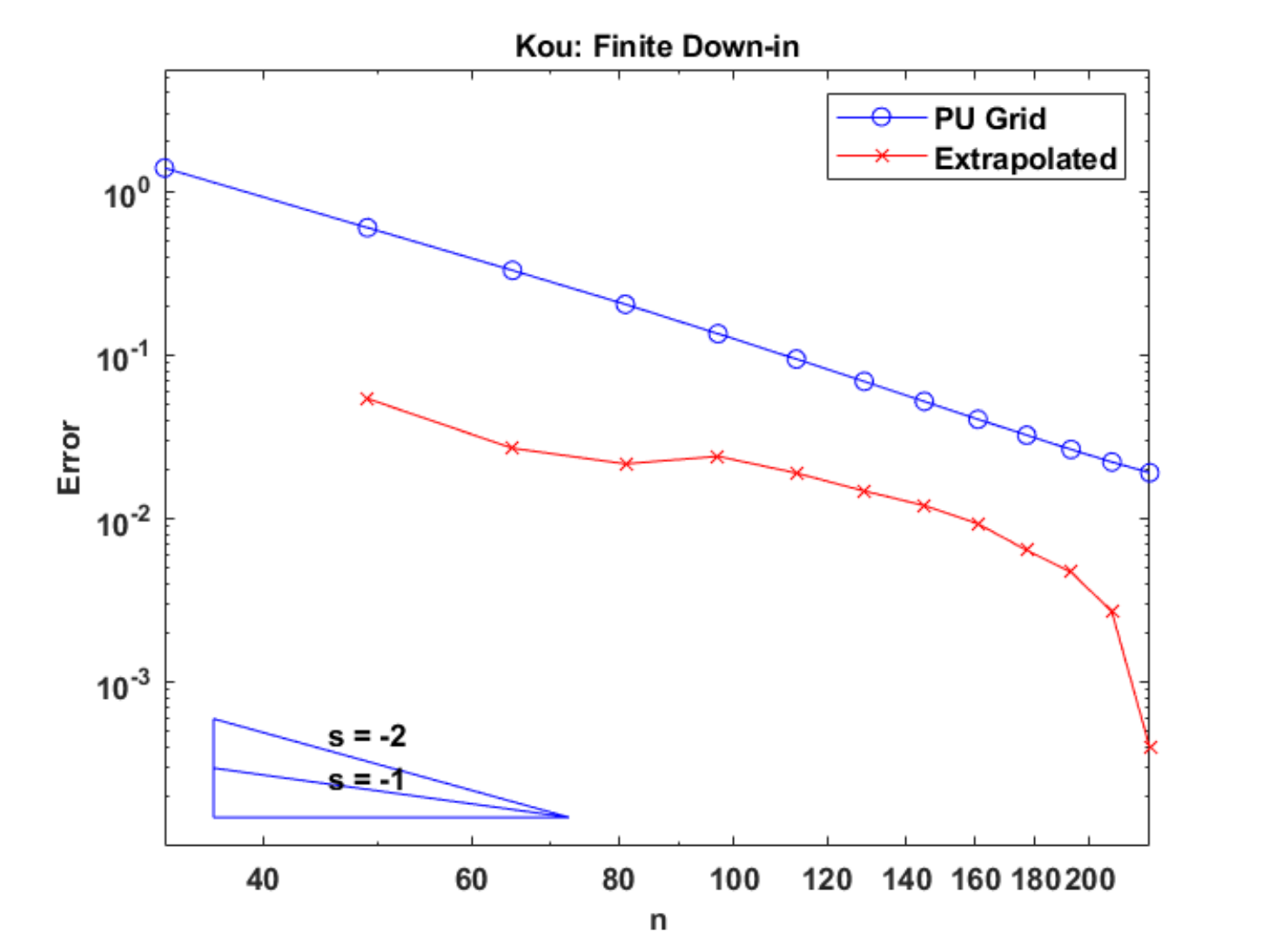}
		\end{minipage}%
	}%
	\subfigure{
		\begin{minipage}[t]{0.5\textwidth}
			\centering
			\includegraphics[width=\linewidth]{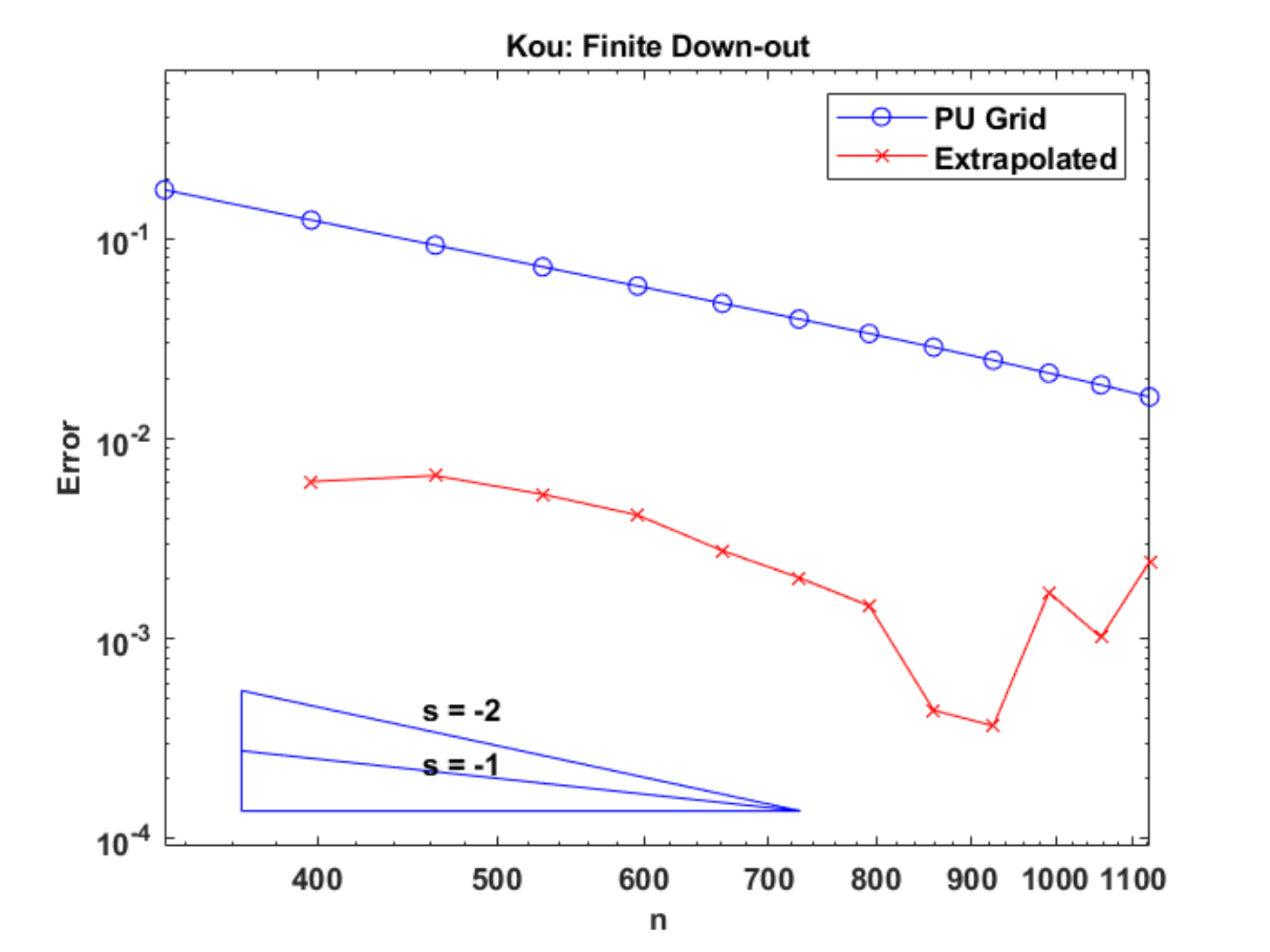}
		\end{minipage}%
	}%
	\centering
	\caption{Convergence of the CTMC approximation for perpetual/finite-maturity down-in/-out American Parisian option under the Kou's model.}
	\label{fig:error-Kou-model}
\end{figure}

\begin{table}[htbp]
\resizebox{\linewidth}{!}{
\begin{tabular}{cccccccccc}
\hline
American Parisian option type                & Grid & Benchmark & CTMC & Abs.Err. & Rel.Err. & Time/s & Extra. & Abs.Err. & Rel.Err. \\ \hline
\multirow{5}{*}{Perpetual down-in}        &   353  &52.4163 &52.5464  & 0.1301  & 0.24$\%$&7.04 & &   &     \\
			&385  &52.4163 &52.5314  & 0.1151 &0.22$\%$& 7.28 &52.4522 &0.0359 &0.07$\%$ \\
			&417  &52.4163 &52.5183  & 0.1020  & 0.19$\%$&8.93 &52.4426 &0.0263 &0.05$\%$         \\
			&449 &52.4163 &52.5070  & 0.0907 & 0.17$\%$&10.44 &52.4361 &0.0198&0.03$\%$      \\ 
			&481 &52.4163 &52.4972 & 0.0809 &0.15$\%$&11.61 &52.4308 &0.0145&0.02$\%$      \\ \hline
\multirow{5}{*}{Perpetual down-out}       &    1849  &20.7958 &20.7064  & 0.0894  & 0.43$\%$&37.05 & &   &     \\
				&1915  &20.7958 &20.7114  & 0.0844 &0.40$\%$& 40.47 &20.7802 &0.0156 &0.07$\%$ \\
				&1981 &20.7958 &20.7160  & 0.0798  & 0.38$\%$&45.73 &20.7816 &0.0142 &0.07$\%$         \\
				&2047 &20.7958 &20.7203  & 0.0755 & 0.36$\%$&49.45 &20.7837 &0.0120&0.06$\%$      \\ 
				&2113 &20.7958 &20.7244 & 0.0714 &0.34$\%$&50.98 &20.7869 &0.0088&0.04$\%$       \\ \hline
\multirow{5}{*}{Finite-maturity down-in}  &   353 &1.1137 &1.0847 & 0.029 &2.60$\%$&33.13 & & &\\
			&385  &1.1137 &1.0877  & 0.0260  &2.33$\%$&36.13 &1.1150 & 0.0013  &  0.12$\%$  \\
			&417  &1.1137 &1.0901  & 0.0236 &2.12$\%$& 41.65 &1.1139&0.0002 &0.02$\%$ \\
			&449 &1.1137 &1.0921  & 0.0216  & 1.94$\%$&44.74 &1.1136 &0.0001 &0.01$\%$         \\
			&481 &1.1137 &1.0938  & 0.0199 & 1.78$\%$&53.20 &1.1136 &0.0001&0.01$\%$        \\ \hline
\multirow{5}{*}{Finite-maturity down-out} &   1123  &3.5011 &3.8158  & 0.3147  &8.99$\%$&78.94 & &   &     \\
			&1189  &3.5011 &3.7957  & 0.2946 &8.41$\%$& 86.71 &3.5123 &0.0112 &0.32$\%$ \\
			&1255 &3.5011 &3.7777  & 0.2766  & 7.90$\%$&97.16 &3.5089 &0.0078 &0.22$\%$         \\
			&1321 &3.5011 &3.7614  & 0.2603& 7.43$\%$&111.58 &3.5044 &0.0033&0.09$\%$      \\ 
			&1387 &3.5011 &3.7467 & 0.2456 &7.01$\%$&136.32 &3.5027 &0.0016&0.04$\%$      \\ \hline
\end{tabular}}
\centering
\caption{The absolute errors and relative errors of CTMC approximation for American Parisian option pricing under the VG model.}
\label{tab:error-VG-model}
\end{table}

\begin{figure}[!h]
	\small
	\centering
	\subfigure{
		\begin{minipage}[t]{0.5\textwidth}
			\centering
			\includegraphics[width=\linewidth]{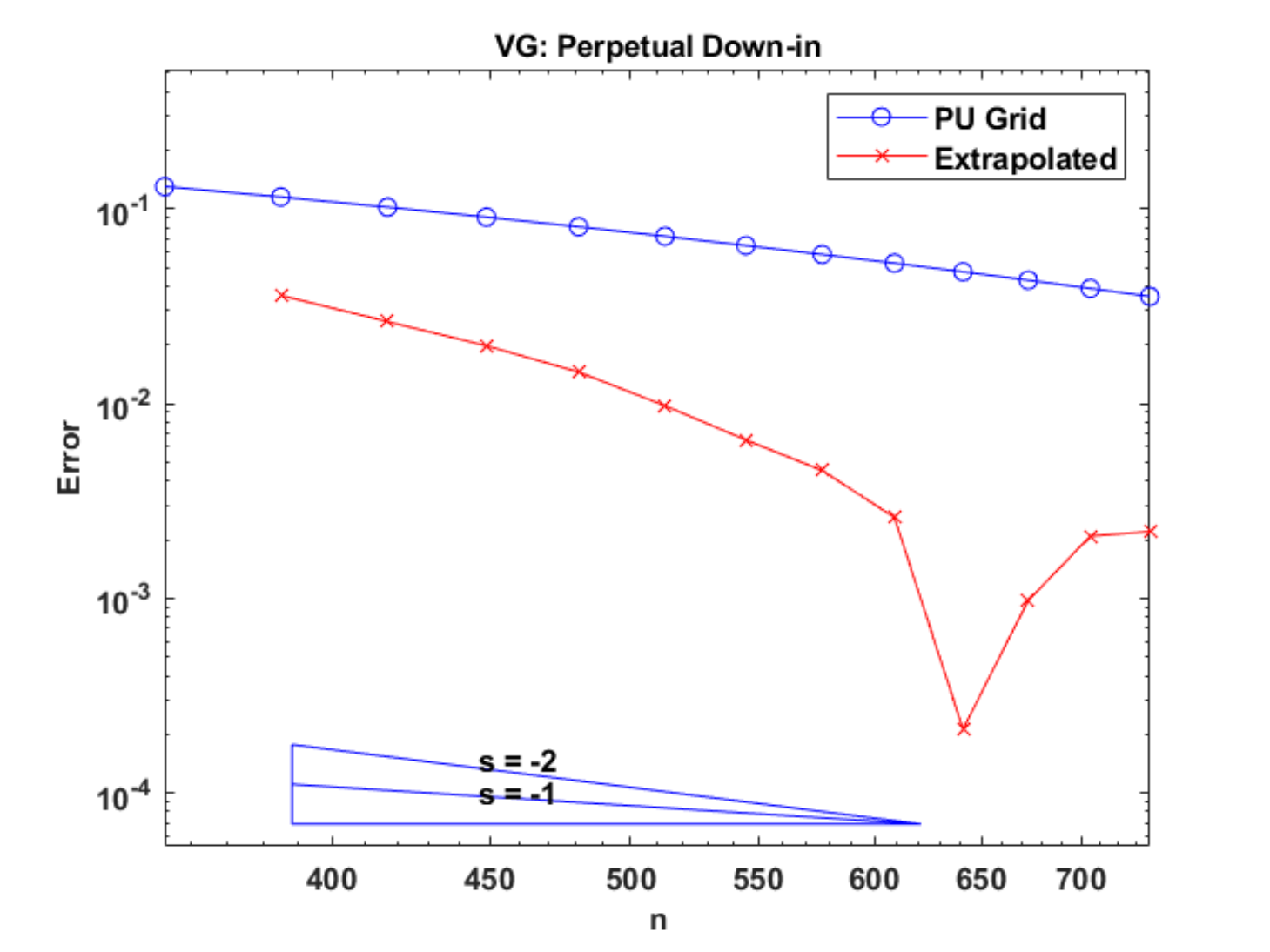}
		\end{minipage}%
	}%
	\subfigure{
		\begin{minipage}[t]{0.5\textwidth}
			\centering
			\includegraphics[width=\linewidth]{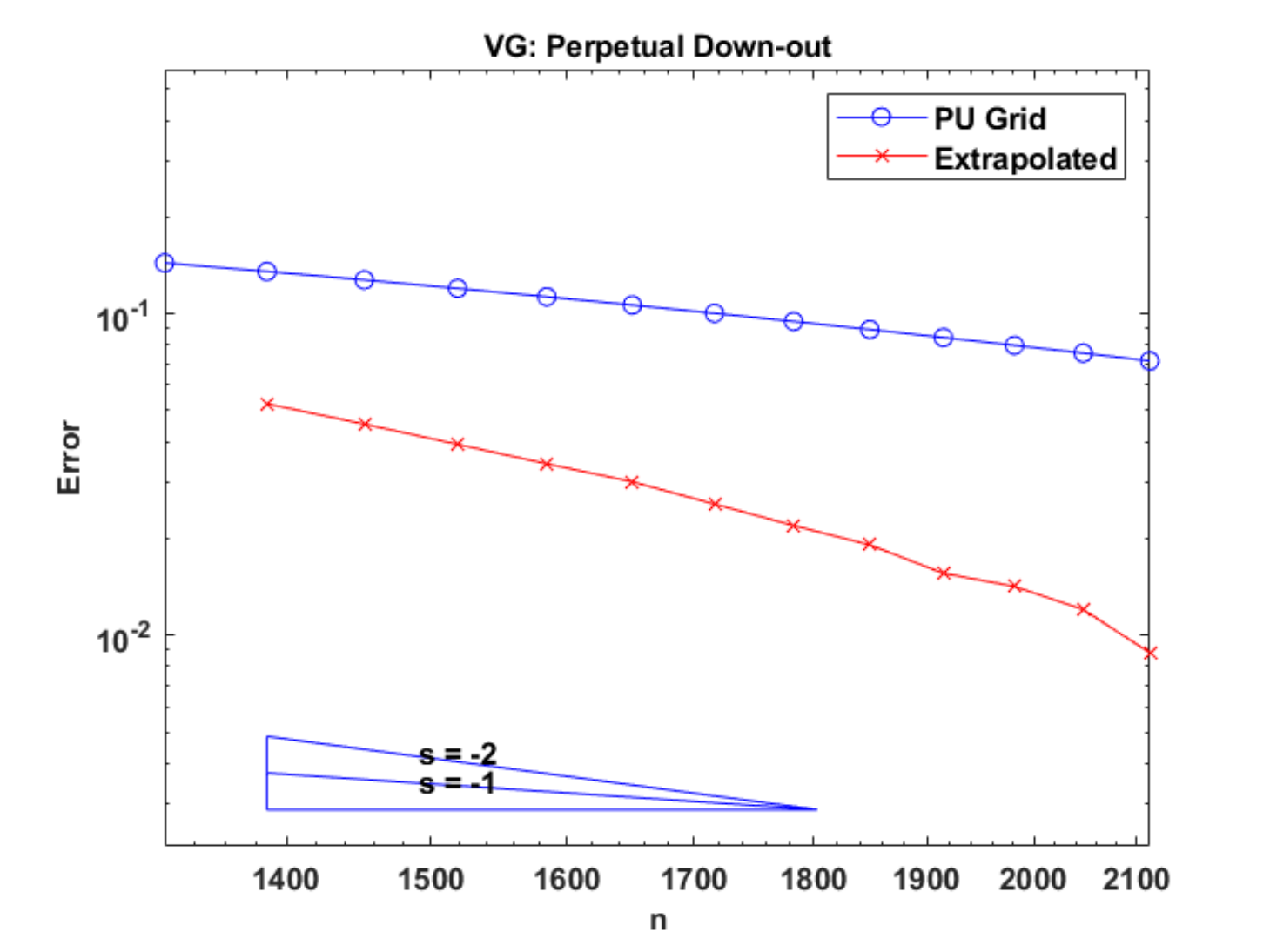}
		\end{minipage}%
	}%
	
	\subfigure{
		\begin{minipage}[t]{0.5\textwidth}
			\centering
			\includegraphics[width=\linewidth]{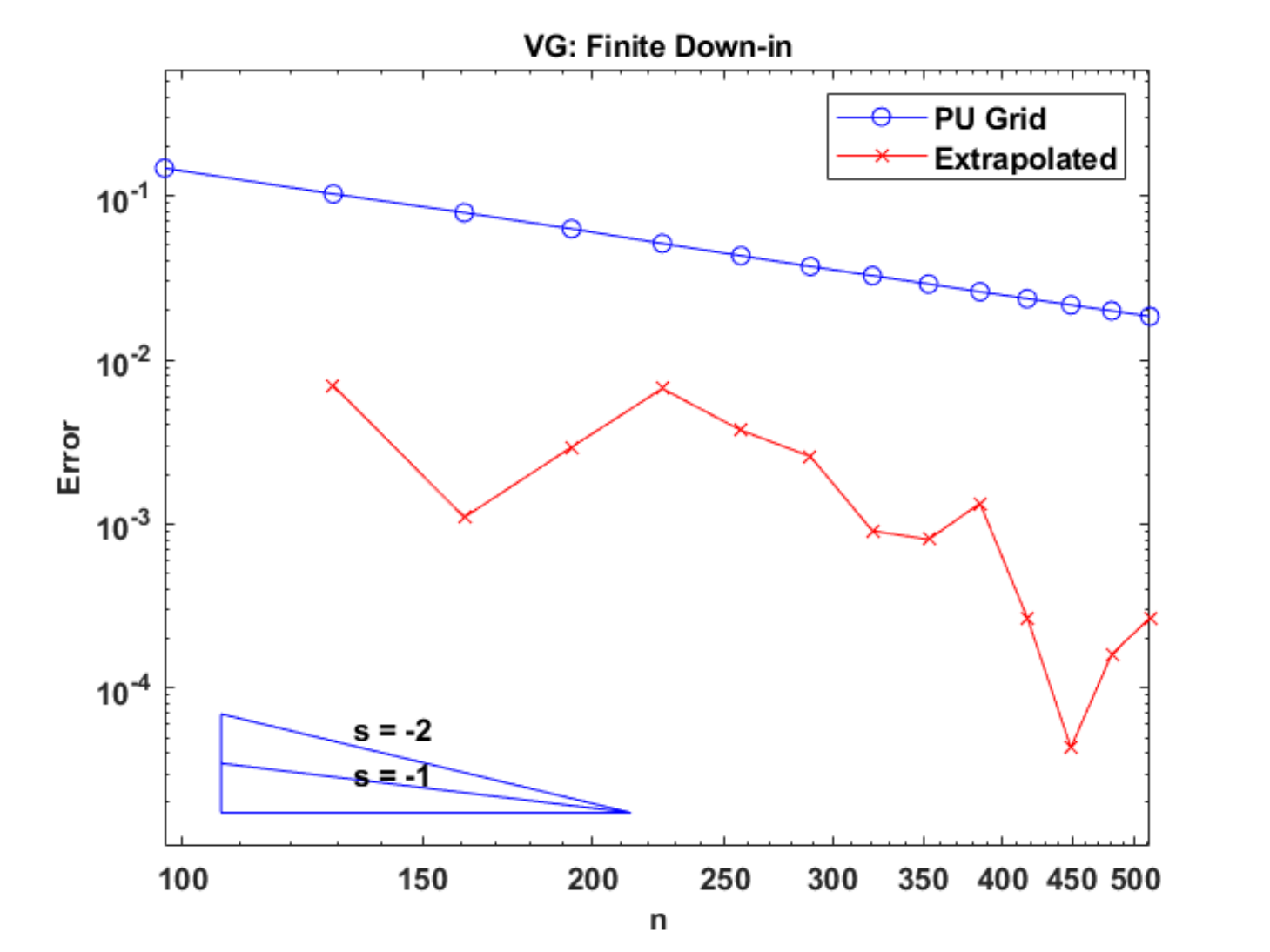}
		\end{minipage}%
	}%
	\subfigure{
		\begin{minipage}[t]{0.5\textwidth}
			\centering
			\includegraphics[width=\linewidth]{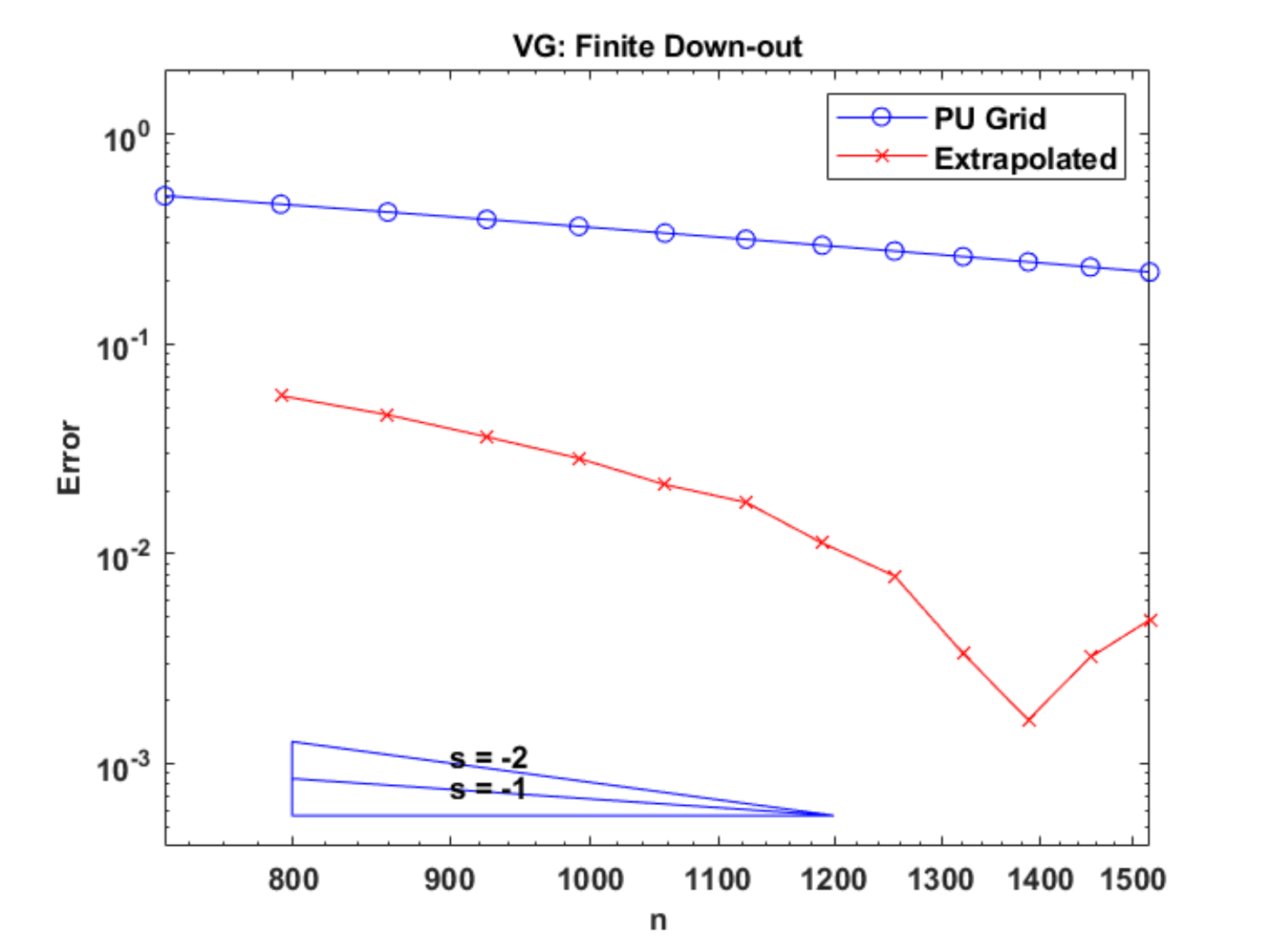}
		\end{minipage}%
	}%
	\centering
	\caption{Convergence of the CTMC approximation for perpetual/finite-maturity down-in/-out American Parisian option under the VG model.}
	\label{fig:error-VG-model}
\end{figure}

\section{Conclusions}
\label{sec:con}

In this paper, we have developed general methods for pricing perpetual/finite-maturity American Parisian down-in/-out options under time-inhomogeneous jump-diffusion models based on a temporal-spatial CTMC approximation that embeds the natural clock. For the down-in cases, we decompose the approximate option price under the CTMC model as the integration of vanilla American option prices w.r.t. the distribution function of the Parisian stopping time; the former is calculated by solving a linear complementary problem and we derive a series of linear and ODE systems to compute the latter. The down-out cases are more complicated due to the complex interaction between early-exercise and the Parisian stopping time. We tackle this challenge by incorporating an additional dimension that records the duration of current excursion and formulating a large linear complementary problem. For the finite-maturity American Parisian options in both the down-in/-out cases, the calculations can be done recursively from the maturity, while the recursion can be avoided leading to significant complexity reduction when the option is perpetual and the model is time-homogeneous. We proved the convergence of CTMC approximation under quite general assumptions on the path properties of the original model and continuity of the payoff function. Numerical experiments covering several representative diffusion, jump-diffusion and pure-jump models demonstrate the efficiency and accuracy of the proposed methods for all the four types of American Parisian options under consideration.

\section*{Competing Interests}

The authors declare no competing interests.

\section*{Acknowledgements}

Gongqiu Zhang was supported by National Natural Science Foundation of China Grant 12171408. Nian Yang was supported by National Natural Science Foundation of China Grant 72171109.
 
\appendices

\section{Proofs}
\label{app:proof}

\begin{proof}[Proof of Proposition \ref{prop:htx-linear-system}]
    We first consider the case $x<L$ and note that,
    \begin{align}\label{eq:htxsy-two-cases}
h(t,x; s,y) = \mathbb{P}_{t,x}\big[  \zeta_{\tau_{L,D}^-} =s, Y_{\tau_{L,D}^-} = y, \tau_L^+ < D \big]+\mathbb{P}_{t,x}\big[  \zeta_{\tau_{L,D}^-} =s, Y_{\tau_{L,D}^-} = y, \tau_L^+ \ge D \big].
		\end{align}
  For the first term, it represents the scenario where $Y$ crosses $L$ from below before time $D$. By the strong Markov property, the process restarts at $\tau_L^+$ from $Y_{\tau_L^+}$ and $\zeta_{\tau_L^+}$, and the clock $t-g^-_{L,t}$ is reset to zero. Therefore,
  \begin{align}
      &\mathbb{P}_{t,x}\big[  \zeta_{\tau_{L,D}^-} =s, Y_{\tau_{L,D}^-} = y, \tau_L^+ < D \big]\\
      &=1_{\{ x < L \}} \sum_{z \ge L, s' \ge t} \mathbb{P}_{t,x} \big[  \tau_L^+ < D, Y_{\tau_L^+} = z,\zeta_{\tau_L^+} = s' \big] h(s',z; s,y).
  \end{align}
For the second term in \eqref{eq:htxsy-two-cases}, $\tau_{L,D}^-=D$ when $\tau_L^+ \ge D$, and we have that,
\begin{align}
    \mathbb{P}_{t,x}\big[  \zeta_{\tau_{L,D}^-} =s, Y_{\tau_{L,D}^-} = y, \tau_L^+ \ge D \big]=1_{\{ x < L \}}  \mathbb{P}_{t,x} \big[  \tau_L^+ \ge D, Y_D = y, \zeta_D = s \big].
\end{align}
For the case when $x\ge L$, the length of excursion will not start to accumulate until $\tau_L^-$ and we restart the process at that time:
\begin{align}
    h(t,x;s,y)=1_{\{ x \ge L \}} \sum_{z < L, s' \ge t} \mathbb{P}_{t,x} \big[ \zeta_{\tau_L^-} = s', Y_{\tau_L^-} = z \big] h(s', z; s, y).
\end{align}
Putting the above results together, we obtain the linear system for $h(t,x;s,y)$. 
\end{proof}

\begin{proof}[Proof of Proposition \ref{prop:cfi-linear-system}]
     By the linear system for $h(t,x; s,y)$ in Proposition \ref{prop:htx-linear-system}, we have that:
     \begin{align}
         \widetilde{C}_{fi}(t,x)&=\sum_{(s,y) \in \mathbb{T} \times \mathbb{S}} e^{-rs} h(t,x; s,y) c_f(s,y) \\
         &=1_{\{ x < L \}} \sum_{z \ge L, s' \ge t} \mathbb{P}_{t,x} \big[  \tau_L^+ < D, Y_{\tau_L^+} = z,\zeta_{\tau_L^+} = s' \big] h(s',z; s,y)e^{-rs}c_f(s,y) \\
			&\quad + 1_{\{ x < L \}}  \mathbb{P}_{t,x} \big[  \tau_L^+ \ge D, Y_D = y, \zeta_D = s \big]e^{-rs}c_f(s,y) \\
			&\quad + 1_{\{ x \ge L \}} \sum_{z < L, s' \ge t} \mathbb{P}_{t,x} \big[ \zeta_{\tau_L^-} = s', Y_{\tau_L^-} = z \big] h(s', z; s, y)e^{-rs}c_f(s,y)\\
       &=1_{\{ x < L \}} \sum_{z \ge L, s' \ge t} \mathbb{P}_{t,x} \big[  \tau_L^+ < D, Y_{\tau_L^+} = z,\zeta_{\tau_L^+} = s' \big] \widetilde{C}_{fi}(s',z) \\
			&\quad + 1_{\{ x < L \}} \sum_{(s,y) \in \mathbb{T}\times \mathbb{S}}  \mathbb{P}_{t,x} \big[  \tau_L^+ \ge D, Y_D = y,\zeta_D = s \big] e^{-rs} c_f(s, y) \\
			&\quad + 1_{\{ x \ge L \}} \sum_{z < L, s' \ge t} \mathbb{P}_{t,x} \big[ \zeta_{\tau_L^-} = s', Y_{\tau_L^-} = z \big] \widetilde{C}_{fi}(s', z).
     \end{align}
 \end{proof}

\begin{proof}[Proof of Proposition \ref{prop:v-system}]
    For $t\le T$, $D>0$ and $x<L$, following the same arguments in the proof of Proposition 2.3 in \cite{zhang2023general}, we have 
    \begin{align}
        \partial_D v(D, t, x) &= \mathcal{G}^{\zeta, Y}v(D,t,x)\\
        &=\frac{v(D, t+\delta_t, x) - v(D, t, x)}{\delta_t} + \mathcal{G}_t^Y v(D, t, x).
    \end{align}
    We recall that $v(D, t, x) = \sum_{(s,y) \in \mathbb{T}\times \mathbb{S}}  \mathbb{P}_{t,x} \big[  \tau_L^+ \ge D, Y_D = y,\zeta_D = s \big] e^{-rs} c_f(s, y)$ and $\mathbb{P}_{t,x} \big[  \tau_L^+ \ge D, Y_D = y,\zeta_D = s \big] = 0$ for $D>0$, $x\ge L$ or $t>T$. Then $v(D,t,x)=0$ in these cases. And for $D=0$, $\mathbb{P}_{t,x} \big[  \tau_L^+ \ge D, Y_D = y,\zeta_D = s \big] = 1_{\{ x = y, s = t \}}$. Then $v(0,t,x)=e^{-rt}c_f(t,x)$.
\end{proof}

\begin{proof}[Proof of Proposition \ref{prop:h^+-system}]
For $x <L$, $t \le T$ and $D > 0$,
	\begin{align}
		\partial_D h^+(D,t,x, z) &= \mathcal{G}^{\zeta, Y} \mathbb{P}_{t,x} \big[  \tau_L^+ < D, Y_{\tau_L^+} = z,\zeta_{\tau_L^+} = t \big]\\
		&= \frac{\mathbb{P}_{t+\delta_t,x} \big[  \tau_L^+ < D, \widetilde{Y}_{\tau_L^+} = z,\zeta_{\tau_L^+} = t \big] - \mathbb{P}_{t,x} \big[ \tau_L^+ < D, Y_{\tau_L^+} = z,\zeta_{\tau_L^+} = t \big]}{\delta_t} \\
		&\quad + \mathcal{G}^{Y}_t \mathbb{P}_{t,x} \big[ \tau_L^+ < D, Y_{\tau_L^+} = z,\zeta_{\tau_L^+} = t \big] \\
		&= (\mathcal{G}_t^Y - 1/\delta_t) h^+(D,t,x, z).
	\end{align}
 The equations for other cases can be easily obtained by the definition of $h^+(D,t,x, z)$. By combining each pair of corresponding terms in the equations for $h_1^+(t, x, z)$ and $h_2^+(D, t, x, z)$, we see that $h_1^+(t, x, z)-h_2^+(D, t, x, z)$ satisfies the equations for $h^+(D,t,x, z)$. 
\end{proof}

\begin{proof}[Proof of Proposition \ref{prop:h^-system}]
    For $x \ge L$ and $t \le T$,
	\begin{align}
		&\mathcal{G}^{\zeta, Y}\mathbb{P}_{t,x} \big[ \zeta_{\tau_L^-} = t, Y_{\tau_L^-} = z \big] \\
		&= \frac{\mathbb{P}_{t+\delta_t,x} \big[\zeta_{\tau_L^-} = t, Y_{\tau_L^-} = z \big] - \mathbb{P}_{t,x} \big[ \zeta_{\tau_L^-} = t, Y_{\tau_L^-} = z \big]}{\delta_t} + \mathcal{G}_t^Y \mathbb{P}_{t,x} \big[ \zeta_{\tau_L^-} = t, Y_{\tau_L^-} = z \big]  \\
		&= -\frac{h^-(t,x,z)}{\delta_t} + \mathcal{G}_t^Y h^-(t,x,z) = 0.
	\end{align}
	For $x < L$, $\zeta_{\tau_L^-} = 0$ and hence $h^-(t, x, z) = 1_{\{ x = z \}}$ in this case.
\end{proof}

\begin{proof}[Proof of Proposition \ref{prop:u^+-system}]
    For $x < L$, $t \le T$ and $D>0$,
	\begin{align}
		&\partial_D u^+(D,t,x) \\
		&= \sum_{z \ge L, s' > t} \partial_D \mathbb{P}_{t,x} \big[ \tau_L^+ < D, Y_{\tau_L^+} = z,\zeta_{\tau_L^+} = s' \big] \widetilde{C}_{fi}(s',z) \\
		&= \sum_{z \ge L, s' > t}  \frac{ \mathbb{P}_{t+\delta_t,x} \big[ \tau_L^+ < D, Y_{\tau_L^+} = z,\zeta_{\tau_L^+} = s' \big] -  \mathbb{P}_{t,x} \big[ \tau_L^+ < D, Y_{\tau_L^+} = z,\zeta_{\tau_L^+} = s' \big]}{\delta_t}  \widetilde{C}_{fi}(s',z) \\
		&\quad + \sum_{z \ge L, s' > t} \mathcal{G}_t^{Y} \mathbb{P}_{t,x} \big[ \tau_L^+ < D, Y_{\tau_L^+} = z,\zeta_{\tau_L^+} = s' \big] \widetilde{C}_{fi}(s',z) \\
		&= \frac{1}{\delta_t} \sum_{z \ge L} h^+(D, t+\delta_t, x, z) \widetilde{C}_{fi}(t+\delta_t, z) + \frac{u^+(D,t+\delta_t,x) - u^+(D,t,x)}{\delta_t} + \mathcal{G}_t^Y u^+(D,t,x).
	\end{align}
	
	For $t >T$ or $D=0$, $1_{\{ \tau_L^+ < D \}} = 0$ and then	$u^+(D, t, x) = 0$. And for $x \ge L$ and $D >0$, $\tau_L^+ = 0$ and $\zeta_{\tau_L^+} = t$. Then $u^+(D, t, x) = 0$ in this case.
\end{proof}

\begin{proof}[Proof of Proposition \ref{prop:u^-system}]
    For $t \le T$ and $x \ge L$,
	\begin{align}
		&\sum_{z < L, s' > t} \mathcal{G}^{\zeta, Y} \mathbb{P}_{t,x} \big[ \zeta_{\tau_L^-} = s', Y_{\tau_L^-} = z \big] \widetilde{C}_{fi}(s', z) \\
		&= \sum_{z < L, s' > t} \frac{\mathbb{P}_{t+\delta_t,x} \big[ \zeta_{\tau_L^-} = s', Y_{\tau_L^-} = z \big] - \mathbb{P}_{t,x} \big[ \zeta_{\tau_L^-} = s', Y_{\tau_L^-} = z \big]}{\delta_t} \widetilde{C}_{fi}(s', z) \\
		&\quad + \sum_{z < L, s' > t} \mathcal{G}^Y_t \mathbb{P}_{t,x} \big[ \zeta_{\tau_L^-} = s', Y_{\tau_L^-} = z \big] \widetilde{C}_{fi}(s', z) \\
		&= \frac{1}{\delta_t}\sum_{z < L}\mathbb{P}_{t+\delta_t,x} \big[ \zeta_{\tau_L^-} = t+\delta_t, Y_{\tau_L^-} = z \big] \widetilde{C}_{fi}(s', z) + \frac{u^-(t+\delta_t,x)-u^-(t,x)}{\delta_t} + \mathcal{G}^Y_t u^-(t,x) \\
		&= \frac{1}{\delta_t}\sum_{z < L} h^-(t+\delta_t, x, z) \widetilde{C}_{fi}(s', z) + \frac{u^-(t+\delta_t,x)-u^-(t,x)}{\delta_t} + \mathcal{G}^Y_t u^-(t,x) \\
		&= 0.
	\end{align}
	We note that $\mathbb{P}_{t,x} \big[ \zeta_{\tau_L^-} = s', Y_{\tau_L^-} = z \big] = 1_{\{t = s', x = z\}}$ when $x < L $ and $\mathbb{P}_{t,x} \big[ \zeta_{\tau_L^-} = s', Y_{\tau_L^-} = z \big] = 0$ for all $s' > t$ when $t = T+\delta$. Therefore, $u^-(t,x) = 0$ in these cases.
\end{proof}

\begin{proof}[Proof of Lemma \ref{lem:perpetual-American-option-intermedia}] We have
    \[
    \begin{aligned}
        &\Big|c^{\hat{\delta}}(0,x)-c^{\hat{\delta}}(x)\Big|\\
        &=\sup_{\tau\in\mathcal{T}^X}\mathbb{E}_{0,x}\Big[e^{-r\tau^{\hat{\delta}}}f(X_{\tau^{\hat{\delta}}})\Big]-\sup_{\tau\in\mathcal{T}^X_{0,T}}\mathbb{E}_{0,x}\Big[e^{-r\tau^{\hat{\delta}}}f(X_{\tau^{\hat{\delta}}})\Big]\\
        &\leq  \sup_{\tau\in\mathcal{T}^X_{0,T}}\mathbb{E}_{0,x}\Big[\Big|e^{-r\tau^{\hat{\delta}}}f(X_{\tau^{\hat{\delta}}})\Big|\Big]+\sup_{\tau\in \mathcal{T}^X \setminus \mathcal{T}_{0,T}^X} \mathbb{E}_{0,x}\Big[\Big|e^{-r\tau^{\hat{\delta}}}f(X_{\tau^{\hat{\delta}}})\Big|\Big]-\sup_{\tau\in\mathcal{T}^X_{0,T}}\mathbb{E}_{0,x}\Big[\Big|e^{-r\tau^{\hat{\delta}}}f(X_{\tau^{\hat{\delta}}})\Big|\Big]\\
       &=  \sup_{\tau\in \mathcal{T}^X \setminus \mathcal{T}^X_{0,T}} \mathbb{E}_{0,x}\Big[\Big|e^{-r\tau^{\hat{\delta}}}f(X_{\tau^{\hat{\delta}}})\Big|\Big]
    \end{aligned}
    \]
    By Assumption \ref{assump:payoff-growth}, there exists a $T_\epsilon$, such that for any $\tau^{\hat{\delta}}>T>T_\epsilon$, $|e^{-r\tau^{\hat{\delta}}}f(X_{\tau^{\hat{\delta}}})|<\epsilon$, and
    \[
    \begin{aligned}
        &\Big|c^{\hat{\delta}}(0,x)-c^{\hat{\delta}}(x)\Big|
        \leq  \sup_{\tau\in \mathcal{T}^X \setminus \mathcal{T}^X_{0,T}} \mathbb{E}_{0,x}\Big[\Big|e^{-r\tau^{\hat{\delta}}}f(X_{\tau^{\hat{\delta}}})\Big|\Big]
        <  \epsilon.
    \end{aligned}
    \]
    Similarly, we can show that $\Big|\widetilde{c}^{\hat{\delta}}_f(0,x)-c_p^{\hat{\delta}}(x)\Big|< \epsilon$. So the lemma is proved.
\end{proof}

\begin{proof}[Proof of Proposition \ref{prop:perpetual-down-in-Amer-option-convergence}]
    We first prove the following claims:

    \begin{itemize}
    \item[(1)] $c^{\hat{\delta}}_p(x)\rightarrow c^{\hat{\delta}}(x)$ as $n\rightarrow \infty$ for any $\hat{\delta}$. 

    \item[(2)] There exists some constant $\widetilde{C}$ independent of $n$ and $\hat{\delta}$ such that 
    \[
    |c^{\hat{\delta}}(x)-c(x)|\leq \widetilde{C}{\hat{\delta}}^{\frac{1}{2}}.
    \]

    \item[(3)] There exists some constant $\widetilde{D}$ independent of $n$ and $\hat{\delta}$ such that
    \[
    |c_p^{\hat{\delta}}(x)-c_p(x)|\leq \widetilde{D}{\hat{\delta}}^{\frac{1}{2}}.
    \]
    \end{itemize}
 
    We will only show the proof for claim (1). The proofs of claim (2) and (3) are analogous to the proof of Theorem 3 in \cite{zhang2021pricing}. 

    Proof of (1): Let $\mathrm{H}$ be the filtration generated by the process $(X,Y)$. $\widetilde{\mathcal{T}}_{s,T}$ is the set of $\mathrm{H}$ stopping times taking values in $[s,T] \cup \{\infty\}$. $\tau^{\hat{\delta}}$ is the stopping time taking the form $\tau^{\hat{\delta}}=\inf\{ t\ge \tau: t\in \mathbb{T}^{\hat{\delta}}\}$ with $\mathbb{T}^{\hat{\delta}}=\{i\Hat{\delta}: i\in \mathbb{N}\}$  for some stopping time $\tau \in \widetilde{\mathcal{T}}_{s,T}$. We first show that $\Big|\widetilde{c}^{\hat{\delta}}_f(s,x)-c^{\hat{\delta}}(s,x)\Big|\rightarrow 0$ as $n\rightarrow \infty$.  
    \begin{align}
        \Big|\widetilde{c}^{\hat{\delta}}_f(s,x)-c^{\hat{\delta}}(s,x)\Big|
        &\leq  \bigg| \sup_{\tau\in \widetilde{\mathcal{T}}_{s,T}}\mathbb{E}_{s,x}\Big[ e^{-r(\tau^{\hat{\delta}}-s)}f(Y_{\tau^{\hat{\delta}}}) - e^{-r(\tau^{\hat{\delta}}-s)} f(X_{\tau^{\hat{\delta}}}) \Big]  \bigg|\\
        &\leq  \mathbb{E}_{s,x}\bigg[\sup_{t\in \mathbb{T}^{\hat{\delta}}\cap [s,T+\hat{\delta}]}  \Big|e^{-r(t - s)} f(Y_{t})-e^{-r(t-s)}f(X_{t})\Big| \bigg]\\
        &\leq  K\mathbb{E}_{s,x}\bigg[\sup_{t\in \mathbb{T}^{\hat{\delta}}\cap [s,T+\hat{\delta}]} \Big|X_{t}-Y_{t}\Big|\bigg]
    \end{align}
    By the assumption that $Y$ converges weakly to $X$ as $n\rightarrow \infty$, $Y_t$ converges to $X_t$ in distribution for any fixed $t\in \mathbb{T}^{\hat{\delta}}\cap [s,T+\hat{\delta}]$. By \cite{eriksson2015american}. Further with the uniform integrability of the collection $(Y_t, X_t, t\in \mathbb{T}^{\hat{\delta}}\cap [s,T+\hat{\delta}], n\in \mathbb{N})$, we have
    \[
    \mathbb{E}_{s,x}\Big[\big|Y_t-X_t\big|\Big]\rightarrow 0,
    \]
    as $n\rightarrow \infty$ for any $t\in \mathbb{T}^{\hat{\delta}}\cap [s,T+\hat{\delta}]$. Noting that $\mathbb{T}^{\hat{\delta}}\cap [s,T+\hat{\delta}]$ contains finite number of elements, we have
    \[
    K\mathbb{E}_{s,x}\Big[\sup_{t\in \mathbb{T}^{\hat{\delta}}\cap [s,T+\hat{\delta}]}\big|X_{t}-Y_{t}\big|\Big]\rightarrow 0
    \]
    as $n\rightarrow \infty$ which implies that $\Big|\widetilde{c}^{\hat{\delta}}_f(s,x)-c^{\hat{\delta}}(s,x)\Big|\rightarrow 0$.
    Then,
    \[
    \begin{aligned}
        &\Big|c^{\hat{\delta}}_p(x)-c^{\hat{\delta}}(x)\Big|\\
        &=\Big|c^{\hat{\delta}}_p(x)-c^{\hat{\delta}}(x)+\widetilde{c}^{\hat{\delta}}_f(0,x)-\widetilde{c}^{\hat{\delta}}_f(0,x)+c^{\hat{\delta}}(0,x)-c^{\hat{\delta}}(0,x)\Big|\\
        &\leq\Big|c^{\hat{\delta}}(0,x)-c^{\hat{\delta}}(x)\Big|+\Big|\widetilde{c}_f^{\hat{\delta}}(0,x)-c_p^{\hat{\delta}}(x)\Big|+\Big|\widetilde{c}_f^{\hat{\delta}}(0,x)-c^{\hat{\delta}}(0,x)\Big|.
    \end{aligned} 
    \]
    By Lemma \ref{lem:perpetual-American-option-intermedia}, for any $\epsilon>0$, there exists a $T_\epsilon$ independent of $n$, such that for any $T>T_\epsilon$,
    \[
    \text{max}\left \{ \Big|c^{\hat{\delta}}(0,x)-c^{\hat{\delta}}(x)\Big|, \Big|\widetilde{c}^{\hat{\delta}}_f(0,x)-c^{\hat{\delta}}_p(x)\Big| \right \}<\epsilon. 
    \]
    And for any fixed $T>T_\epsilon$, there exists an $N_\epsilon$ such that for $n>N_\epsilon$, 
    \[
    \Big|\widetilde{c}^{\hat{\delta}}_f(0,x)-c^{\hat{\delta}}(0,x)\Big|<\epsilon.
    \]
    Combining the results, we have $|c^{\hat{\delta}}_p(x)-c^{\hat{\delta}}(x)|<3\epsilon$. Since $\epsilon$ can be arbitrarily small, (1) is proved. 

    Combining the claims (1), (2) and (3), we get that,
    \[
    \begin{aligned}
    &\Big|c_p(x)-c(x)\Big| \\
    &\leq \Big|c_p(x)-c^{\hat{\delta}}_p(x) \Big|+\Big|c^{\hat{\delta}}_p(x) -c^{\hat{\delta}}(x)  \Big|+\Big|c^{\hat{\delta}}(x)-c(x)   \Big| \\
    &\le (\widetilde{C} + \widetilde{D}){\hat{\delta}}^{\frac{1}{2}} + \Big|c^{\hat{\delta}}_p(x) -c^{\hat{\delta}}(x)  \Big|.
    \end{aligned}
    \]
    We note that $\hat{\delta}$ can be arbitrarily small and for each $\hat{\delta}$, we can select a sufficiently large $N$ such that for all $n > N$, $\big|c^{\hat{\delta}}_p(x) -c^{\hat{\delta}}(x)\big| < \epsilon$ for any $\epsilon>0$. This concludes the proof.
\end{proof}

\begin{proof}[Proof of Theorem \ref{thm:perpetual-down-in-Amer-Pari-option-convergence}]
    Firstly, we note that:
    \[
    \begin{aligned}
     \sup_{\tau \in \mathcal{T}^Y}\mathbb{E}_x \big[ e^{-r\tau}  f( Y_\tau) 1_{\{ \tau \ge \tau_{L, D}^{Y,-} \}} \big]=\mathbb{E}_x\big[\sup_{\tau\in\mathcal{T}^Y}\mathbb{E}_x\big[e^{-r\tau}f(Y_{\tau})1_{\{ \tau\ge \tau_{L, D}^{Y,-} \}}|\tau_{L, D}^{Y,-}\big]  \big].
    \end{aligned}
    \]
For the quantity inside the expectation, $\sup_{\tau\in\mathcal{T}^Y}\mathbb{E}_x\big[e^{-r\tau}f(Y_{\tau})1_{\{ \tau\ge \tau_{L, D}^{Y,-} \}}|\tau_{L, D}^{Y,-}\big]$ is the approximate value function of a perpetual American option with initial time $\tau_{L, D}^{Y,-}$. Therefore, by Proposition \ref{prop:perpetual-down-in-Amer-option-convergence} and Proposition \ref{prop:parisian-stopping-time-weak-convergence}, as $n \to \infty$,
    \[
    \begin{aligned}
\mathbb{E}_x\big[\sup_{\tau\in\mathcal{T}^Y}\mathbb{E}_x\big[e^{-r\tau}f(Y_{\tau})1_{\{ \tau\ge \tau_{L, D}^{Y,-} \}}|\tau_{L, D}^{Y,-}\big]  \big]
        &\rightarrow \mathbb{E}_x\big[\sup_{\tau\in\mathcal{T}^X}\mathbb{E}_x\big[e^{-r\tau}f(X_{\tau})1_{\{ \tau\ge \tau_{L, D}^{X,-} \}}|\tau_{L, D}^{X,-}\big]\big] \\
        &=\sup_{\tau \in \mathcal{T}^X}\mathbb{E}_x\big[e^{-r\tau}f(X_{\tau})1_{\{ \tau\ge \tau_{L, D}^{X,-} \}}\big].
    \end{aligned}
    \]
\end{proof}

\begin{proof}[Proof of Proposition \ref{prop:finite-American-option-converge}]
To prove Proposition \ref{prop:finite-American-option-converge}, we first establish the following results: 
\begin{itemize}
    \item[(1)] There exists some constant $\widetilde{C}$ independent of $n$, $\delta_t$ and $\hat{\delta}$ such that
\[
\Big|c(s,x)-c^{\hat{\delta}} (s,x)\Big|\le \widetilde{C}\Hat{\delta}^{\frac{1}{2}}.
\]

    \item[(2)] There exists some constant $\widetilde{D}$ independent of $n$, $\delta_t$ and $\hat{\delta}$ such that
\[
\Big|c_f(s,x)-c_f^{\hat{\delta}} (s,x)\Big|\le \widetilde{D}\Hat{\delta}^{\frac{1}{2}}.
\]

    \item[(3)] For any fixed $\hat{\delta}$, as $\delta_t\rightarrow 0$ and $n\rightarrow \infty$,
\[
c_{f}^{\hat{\delta}} (s,x)\rightarrow c^{\hat{\delta}} (s,x).
\]
\end{itemize}

We will show the details for the proofs of (2) and (3). The proof of (1) is similar to that of (2). 

Proof of (2): We note that $\zeta$ is  $\delta_t$ times a Poisson process with intensity $1/\delta_t$ and $\zeta_t\rightarrow t$ almost surely as $\delta_t\rightarrow 0$. By definition, $0 \le \tau^{\hat{\delta}} - \tau \le \hat{\delta}$ and then it is not difficult to see that for any $\tau\in \mathcal{T}_s$, $\mathbb{E}_{s, x}\big[\zeta_{\tau^{\hat{\delta}}}-\zeta_{\tau} \big]\le \hat{\delta}$. Next, we analyze $c_f(s,x)-c_f^{\hat{\delta}} (s,x)$. By the triangle inequality, we have that
    \[
    \begin{aligned}
     & \left|c_f(s,x)-c_f^{\hat{\delta}} (s,x)\right| \\
&\leq  \sup _{\tau \in \mathcal{T}^{\zeta,Y}_{s}} \mathbb{E}_{s, x}\bigg[\Big|e^{-r (\zeta_{\tau}-s)} f\big(Y_{\tau}\big)-e^{-r (\zeta_{\tau^{\hat{\delta}}}-s)} f\left(Y_{\tau^{\hat{\delta}}}\right)\Big|\bigg] \\
&\leq  \sup _{\tau \in \mathcal{T}^{\zeta,Y}_{s}} \mathbb{E}_{s, x}\bigg[\Big|\left(e^{-r (\zeta_{\tau}-s)}-e^{-r (\zeta_{\tau^{\hat{\delta}}}-s)}\right) f\left(Y_{\tau}\right)\Big|+\Big|e^{-r (\zeta_{\tau^{\hat{\delta}}}-s)}\left(f\left(Y_{\tau}\right)-f\left(Y_{\tau^{\hat{\delta}}}\right)\right)\Big|\bigg] \\
&\leq  r \sup _{\tau \in \mathcal{T}^{\zeta, Y}_{s}} \mathbb{E}_{s, x}\big[\zeta_{\tau^{\hat{\delta}}}-\zeta_{\tau}  \big]\sup _{\tau \in \mathcal{T}^{\zeta,Y}_{s}} \mathbb{E}_{s, x}\Big[e^{-r (\zeta_{\tau}-s)}\big|f\left(Y_{\tau}\right)\big|\Big]\\
&\;\;\;\;+K\sup _{\tau \in \mathcal{T}^{Y}_{s}} \mathbb{E}_{s, x}\Big[\big|Y_{\tau}-Y_{\tau^{\hat{\delta}}}\big|\Big] \quad \text { (by Lipschitz continuity) } \\
&\leq  \Hat{\delta} r C+K\sup _{\tau \in \mathcal{T}^Y_{s}}\bigg(\mathbb{E}_{s, x}\Big[\big|Y_{\tau}-Y_{\tau^{\hat{\delta}}}\big|^2\Big]\bigg)^{\frac{1}{2}} \quad \text {(by Assumption \ref{assump:payoff-growth} and the Cauchy-Schwartz inequality) }\\
&=  \Hat{\delta} r C+K\sup _{\tau \in \mathcal{T}^Y_{s}}\bigg(\mathbb{E}_{s, x}\Big[\mathbb{E}_{\tau, Y_{\tau}}\Big[\big|Y_{\tau}-Y_{\tau^{\hat{\delta}}}\big|^2\Big]\Big]\bigg)^{\frac{1}{2}} \\
&\leq  \Hat{\delta} r C+K\sup _{\tau \in \mathcal{T}^Y_{s}}\bigg(\mathbb{E}_{s, x}\Big[\mathbb{E}_{\tau, Y_{\tau}}\Big[\langle Y\rangle_{\Hat{\delta}}\Big]\Big]\bigg)^{\frac{1}{2}} \\
&\leq  \Hat{\delta} r C+K\Hat{\delta}^{\frac{1}{2}} K_{Y}^{\frac{1}{2}}, \quad \text {(by Assumption \ref{assump:moment-quadratic-variation}} )
    \end{aligned}
    \]
    where $C$ is a constant independent of $n$ and $\delta_t$, $K$ is the Lipschitz constant of $f(x)$, and $K_Y$ is the Lipschitz constant of $\mathbb{E}_{t, x}\big[\langle Y\rangle_{\Hat{\delta}}\big]$ w.r.t. $\hat{\delta}$.

    Proof of (3): Let $\mathrm {H}$ and $\mathrm{F}^{\zeta,X}$ be the filtration generated by the processes $(X, Y)$ and $(\zeta, X)$, respectively. We introduce the following additional intermediate quantities:
    \[
    \hat{c}_f^{\hat{\delta}}(s,x)=\sup_{\tau\in \mathcal{T}^{\zeta, Y} _{s}}\mathbb{E}_{s,x}\Big[e^{-r(\tau^{\hat{\delta}} - s)} f(Y_{\tau^{\hat{\delta}}})1_{\{ \zeta_{\tau}\le T \}} \Big],
    \]
    \[
    \hat{c}^{\hat{\delta}}(s,x)=\sup_{\tau\in \mathcal{T}^{\zeta, X} _{s}}\mathbb{E}_{s,x}\Big[e^{-r(\tau^{\hat{\delta}} - s)} f(X_{\tau^{\hat{\delta}}})1_{\{ \zeta_{\tau}\le T \}} \Big],
    \]
    \[
    \Bar{c}_f^{\hat{\delta}}(s,x)=\sup_{\tau\in \mathcal{T}^{\zeta,Y} _{s}}\mathbb{E}_{s,x}\Big[e^{-r(\zeta_{\tau}^{\hat{\delta}} - s)} f(Y_{\tau^{\hat{\delta}}})1_{\{ \zeta_{\tau}\le T \}}1_{\{\tau\le \overline{T} \}} \Big],
    \]
    \[
    \tilde{c}_f^{\hat{\delta}}(s,x)=\sup_{\tau\in \mathcal{T}^{\zeta, Y} _{s}}\mathbb{E}_{s,x}\Big[e^{-r(\tau^{\hat{\delta}} - s)} f(Y_{\tau^{\hat{\delta}}})1_{\{ \zeta_{\tau}\le T \}}1_{\{\tau\le \overline{T} \}} \Big],
    \]
    \[
    \tilde{c}^{\hat{\delta}}(s,x)=\sup_{\tau\in \mathcal{T}^{\zeta, X} _{s}}\mathbb{E}_{s,x}\Big[e^{-r(\tau^{\hat{\delta}} - s)} f(X_{\tau^{\hat{\delta}}})1_{\{ \zeta_{\tau}\le T \}}1_{\{\tau\le \overline{T} \}} \Big],
    \]
    where $\mathcal{T}^{\zeta, X} _{s}$ and $\mathcal{T}^{\zeta, Y} _{s}$ are the sets of $\mathrm{F}^{\zeta, X}$ and $\mathrm{F}^{\zeta, Y}$ stopping times taking values in $[s,\infty]$.
     
    Since 
    \begin{align}
        &\Big|c_f^{\hat{\delta}} (s,x)-\hat{c}_f^{\hat{\delta}} (s,x)\Big|\\
         &=  \Big|c_f^{\hat{\delta}} (s,x)- \Bar{c}_f^{\hat{\delta}}(s,x)+\Bar{c}_f^{\hat{\delta}}(s,x)-\tilde{c}_f^{\hat{\delta}}(s,x)+\tilde{c}_f^{\hat{\delta}}(s,x)-\hat{c}_f^{\hat{\delta}} (s,x)  \Big|\\
        & \leq  \Big|c_f^{\hat{\delta}} (s,x)- \Bar{c}_f^{\hat{\delta}}(s,x)\Big| + \Big| \Bar{c}_f^{\hat{\delta}}(s,x)-\tilde{c}_f^{\hat{\delta}}(s,x)\Big| +\Big| \tilde{c}_f^{\hat{\delta}}(s,x)-\hat{c}_f^{\hat{\delta}} (s,x)\Big| \\
        &\leq \bigg| \sup_{\tau\in \mathcal{T}^{\zeta,Y} _{s}}\mathbb{E}_{s,x}\Big[e^{-r\zeta_{\tau^{\hat{\delta}}}} f(Y_{\tau^{\hat{\delta}}})-e^{-r\zeta_{\tau^{\hat{\delta}}}} f(Y_{\tau^{\hat{\delta}}})1_{\{ \tau\le \overline{T} \}}\Big]\bigg|\\
        &\;\;\;\;+\bigg| \sup_{\tau\in \mathcal{T}^{\zeta, Y} _{s}}\mathbb{E}_{s,x}\Big[e^{-r\zeta_{\tau^{\hat{\delta}}}} f(Y_{\tau^{\hat{\delta}}})1_{\{ \tau\le \overline{T} \}}-e^{-r\tau^{\hat{\delta}}} f(Y_{\tau^{\hat{\delta}}})1_{\{ \tau\le \overline{T} \}}\Big]\bigg|\\
        & \;\;\;\; + \bigg| \sup_{\tau\in \mathcal{T}^Y _{s}}\mathbb{E}_{s,x}\Big[e^{-r\tau^{\hat{\delta}}} f(Y_{\tau^{\hat{\delta}}})1_{\{ \tau\le \overline{T} \}}-e^{-r\tau^{\hat{\delta}}} f(Y_{\tau^{\hat{\delta}}}) \Big]\bigg|\\
        &\leq   \sup_{\tau\in \mathcal{T}^{\zeta,Y} _{s}\setminus \mathcal{T}^{\zeta,Y} _{s,\overline{T}}}\mathbb{E}_{s,x}\bigg[\Big|e^{-r\zeta_{\tau^{\hat{\delta}}}} f(Y_{\tau^{\hat{\delta}}})\Big|\bigg]+\mathbb{E}_{s,x}\bigg[\sup_{t\in \mathbb{T}^{\hat{\delta}}\cap [s,\overline{T}+\hat{\delta}]}  \Big|e^{-r\zeta_t} f(Y_{t})-e^{-rt}f(Y_{t})\Big| \bigg]\\
        & \;\;\;\;\; +\sup_{\tau\in \mathcal{T}^Y _{s}\setminus \mathcal{T}^Y _{s,\overline{T}}}\mathbb{E}_{s,x}\bigg[\Big|e^{-r\tau^{\hat{\delta}}} f(Y_{\tau^{\hat{\delta}}})\Big|\bigg], \label{eq:convergence-of-discount-factor}
    \end{align}
    we will show that for any $\epsilon>0$, there exist a $\overline{T}_{\epsilon}$ independent of $n$ and $\delta_t$, and there exists a $\delta_t^{\epsilon}$ independent of $n$, such that for any $\overline{T}>\overline{T}_{\epsilon}$ and $\delta_t<\delta_t^{\epsilon}$,  $|c_f^{\hat{\delta}} (s,x)- \hat{c}_f^{\hat{\delta}}(s,x)|<\epsilon$.
    For the first and last terms of \eqref{eq:convergence-of-discount-factor}, we can use the same arguments in the proof of Lemma \ref{lem:perpetual-American-option-intermedia} and show that for any $\epsilon>0$, there exists a $\overline{T}_{\epsilon}$ independent of $n$ and $\delta_t$ such that for any $\overline{T}>\overline{T}_{\epsilon}$, 
    \[
    \max\Big\{\sup_{\tau\in \mathcal{T}^{\zeta,Y} _{s}\setminus \mathcal{T}^{\zeta,Y} _{s, \overline{T}}}\mathbb{E}_{s,x}\bigg[\Big|e^{-r\zeta_{\tau^{\hat{\delta}}}} f(Y_{\tau^{\hat{\delta}}})\Big|\bigg], \sup_{\tau\in \mathcal{T}^Y _{s}\setminus \mathcal{T}^Y _{s,\overline{T}}}\mathbb{E}_{s,x}\bigg[\Big|e^{-r\tau^{\hat{\delta}}} f(Y_{\tau^{\hat{\delta}}})\Big|\bigg] \Big\}<\frac{\epsilon}{3}.
    \]
    And for the second term in \eqref{eq:convergence-of-discount-factor}, we fix $\overline{T}=\overline{T}_{\epsilon}$. We further note that inside the expectation the supremum is taken over a finite set $\mathbb{T}^{\hat{\delta}}\cap[s,\overline{T}+\hat{\delta}]$ which implies that there exists a $\delta_t^{\epsilon}$ independent of $n$ such that for any $\delta_t<\delta_t^{\epsilon}$, 
    \[
    \mathbb{E}_{s,x}\bigg[\sup_{t\in \mathbb{T}^{\hat{\delta}}\cap [s,\overline{T}+\hat{\delta}]}  \Big|e^{-r\zeta_t} f(Y_{t})-e^{-rt}f(Y_{t})\Big| \bigg]<\frac{\epsilon}{3}.
    \]
Combining the estimates above, we get that $|c_f^{\hat{\delta}} (s,x)- \hat{c}_f^{\hat{\delta}}(s,x)|<\epsilon$.

Moreover,
\begin{align}
   & \Big|\hat{c}_f^{\hat{\delta}}(s,x)-\hat{c}^{\hat{\delta}}(s,x)\Big|\\
   &\leq  \Big|\hat{c}_f^{\hat{\delta}}(s,x) -\tilde{c}_f^{\hat{\delta}}(s,x)\Big|+\Big|\tilde{c}_f^{\hat{\delta}}(s,x)-\tilde{c}^{\hat{\delta}}(s,x)\Big| +\Big|\tilde{c}^{\hat{\delta}}(s,x)-\hat{c}^{\hat{\delta}}(s,x)\Big| \\
   &\leq \bigg| \sup_{\tau\in \mathcal{T}^Y _{s}}\mathbb{E}_{s,x}\Big[e^{-r\tau^{\hat{\delta}}} f(Y_{\tau^{\hat{\delta}}})-e^{-r\tau^{\hat{\delta}}} f(Y_{\tau^{\hat{\delta}}})1_{\{ \tau\le \overline{T} \}}\Big]\bigg|+\bigg| \sup_{\tau\in \widetilde{\mathcal{T}} _{s}}\mathbb{E}_{s,x}\Big[e^{-r\tau^{\hat{\delta}}} f(Y_{\tau^{\hat{\delta}}})1_{\{ \tau\le \overline{T} \}}-e^{-r\tau^{\hat{\delta}}} f(X_{\tau^{\hat{\delta}}})1_{\{ \tau\le \overline{T} \}}\Big]\bigg|\\
        &\; \;\;\;+ \bigg| \sup_{\tau\in \mathcal{T}^X _{s}}\mathbb{E}_{s,x}\Big[e^{-r\tau^{\hat{\delta}}} f(X_{\tau^{\hat{\delta}}})1_{\{ \tau\le \overline{T} \}}-e^{-r\tau^{\hat{\delta}}} f(X_{\tau^{\hat{\delta}}}) \Big]\bigg|,
\end{align}
where $\widetilde{\mathcal{T}}_{s}$ is the set of $\mathrm {H}$ stopping times taking values in $[s,\infty]$.
Then we can use the same arguments as before and show that $\hat{c}_f^{\hat{\delta}}(s,x)\rightarrow \hat{c}^{\hat{\delta}}(s,x)$ as $n\rightarrow \infty$ for any fixed $\overline{T}$.

Now we show that $\hat{c}^{\hat{\delta}}(s,x)\rightarrow c^{\hat{\delta}}(s,x)$ as $\delta_t\rightarrow 0$. Let $T^{\delta_t}=\inf \{t\ge 0: \zeta_t=T \}$, which follows a Gamma distribution with mean $T$ and $T^{\delta_t}\rightarrow T$ as $\delta_t\rightarrow 0$. We note that,
\begin{align}
    \sup_{\tau\in \mathcal{T}^{\zeta,X} _{s}}\mathbb{E}_{s,x}\Big[e^{-r(\tau^{\hat{\delta}} - s)} f(X_{\tau^{\hat{\delta}}})1_{\{ \zeta_{\tau}\le T \}} \Big]&= \sup_{\tau\in \mathcal{T}^{\zeta,X} _{s}}\mathbb{E}_{s,x}\Big[e^{-r(\tau^{\hat{\delta}} - s)} f(X_{\tau^{\hat{\delta}}})1_{\{ \tau\le T^{\delta_t} \}} \Big]\\
    &=  \mathbb{E}_{s,x} \bigg[\sup_{\tau\in \mathcal{T}^{X} _{s}}\mathbb{E}_{s,x}\Big[e^{-r(\tau^{\hat{\delta}} - s)} f(X_{\tau^{\hat{\delta}}})1_{\{ \tau\le T^{\delta_t} \}}\Big| T^{\delta_t} \Big]   \bigg].
\end{align}
For the quantity above, $\sup_{\tau\in \mathcal{T}^X _{s}}\mathbb{E}_{s,x}\Big[e^{-r(\tau^{\hat{\delta}} - s)} f(X_{\tau^{\hat{\delta}}})1_{\{ \tau\le T^{\delta_t} \}}\Big| T^{\delta_t} \Big]$ is the value function of a finite-maturity Bermudan option with maturity $T^{\delta_t}$. As $\delta_t\rightarrow 0$, 
\begin{align}
    \mathbb{E}_{s,x} \bigg[\sup_{\tau\in \mathcal{T}^X _{s}}\mathbb{E}_{s,x}\Big[e^{-r(\tau^{\hat{\delta}} - s)} f(X_{\tau^{\hat{\delta}}})1_{\{ \tau\le T^{\delta_t} \}}\Big| T^{\delta_t} \Big]   \bigg] &\rightarrow  \mathbb{E}_{s,x} \bigg[\sup_{\tau\in \mathcal{T}^X _{s}}\mathbb{E}_{s,x}\Big[e^{-r(\tau^{\hat{\delta}} - s)} f(X_{\tau^{\hat{\delta}}})1_{\{ \tau\le T \}}\Big| T \Big]   \bigg]\\
    &=  \sup_{\tau\in \mathcal{T}^X _{s}}\mathbb{E}_{s,x}\Big[e^{-r(\tau^{\hat{\delta}} - s)} f(X_{\tau^{\hat{\delta}}})1_{\{ \tau\le T \}}\Big].
\end{align}
Now we are ready to establish the claim (3) by collecting the intermediate estimates above:
    \begin{align}
    \Big|c_f^{\hat{\delta}} (s,x)-c^{\hat{\delta}} (s,x)\Big|
    &= \Big|c_f^{\hat{\delta}} (s,x)-\hat{c}_f^{\hat{\delta}} (s,x)+\hat{c}_f^{\hat{\delta}} (s,x)-\hat{c}^{\hat{\delta}} (s,x)+\hat{c}^{\hat{\delta}} (s,x)-c^{\hat{\delta}} (s,x) \Big| \\
    &\leq  \Big|c_f^{\hat{\delta}} (s,x)-\hat{c}_f^{\hat{\delta}} (s,x)\Big|+ \Big|\hat{c}_f^{\hat{\delta}} (s,x)-\hat{c}^{\hat{\delta}} (s,x)\Big|+\Big|\hat{c}^{\hat{\delta}} (s,x)-c^{\hat{\delta}} (s,x) \Big| \to 0,
    \end{align}
    as $n\rightarrow \infty$ and $\delta_t\rightarrow 0$ for any fixed $\hat{\delta}$.

    Finally, we note that
     \begin{align}
     \big| c^{\hat{\delta}}_f(s,x)-c^{\hat{\delta}}(s,x)\big|
     \leq \Big|c_f^{\hat{\delta}} (s,x)-\hat{c}_f^{\hat{\delta}} (s,x)\Big|+ \Big|\hat{c}_f^{\hat{\delta}} (s,x)-\hat{c}^{\hat{\delta}} (s,x)\Big|+\Big|\hat{c}^{\hat{\delta}} (s,x)-c^{\hat{\delta}} (s,x) \Big|.
     \end{align}
   Then the proof is concluded by combining the claims (1), (2) and (3).
\end{proof}

\begin{proof}[Proof of Theorem \ref{thm:finite-down-in-Amer-Pari-option-converge}]
    We note that
    \[
    \begin{aligned}
        \sup_{\tau \in \mathcal{T}_t^{\zeta,Y}}\mathbb{E}_{t,x} \big[ e^{-r(\zeta_{\tau} - t)}  f( Y_\tau) 1_{\{ \tau_{L, D}^{Y,-} \le \tau, \zeta_{\tau} \le T \}} \big]
        &=\mathbb{E}_{t,x}\big[\sup_{\tau\in \mathcal{T}_t^{\zeta,Y}}\mathbb{E}_{t,x}\big[e^{-r(\zeta_{\tau} - t)}  f( Y_{\tau}) 1_{\{\zeta_{\tau} \le T \}}1_{\{\tau\ge \tau_{L, D}^{Y,-} \}}|\tau_{L, D}^{Y,-}\big] \big],
    \end{aligned}
    \]
    where $\mathcal{T}_t^{\zeta,Y}$ is the set of stopping times taking values in $[t,\infty]$. 
    
    For the above, $\sup_{\tau\in \mathcal{T}_t^{\zeta,Y}}\mathbb{E}_{t,x}\big[e^{-r(\zeta_{\tau} - t)}  f( Y_{\tau}) 1_{\{\zeta_{\tau} \le T \}}1_{\{\tau\ge \tau_{L, D}^{Y,-} \}}|\tau_{L, D}^{Y,-}\big]$ is the approximate value function of a finite-maturity American option with initial time $\tau_{L, D}^{Y,-}$. Then by Proposition \ref{prop:finite-American-option-converge} and Proposition \ref{prop:parisian-stopping-time-weak-convergence}, we have
    \[
    \begin{aligned}
        &\mathbb{E}_{t,x}\big[\sup_{\tau\in \mathcal{T}_t^{\zeta,Y}}\mathbb{E}_{t,x}\big[e^{-r(\zeta_{\tau} - t)}  f( Y_{\tau}) 1_{\{\zeta_{\tau} \le T \}}1_{\{\tau\ge \tau_{L, D}^{Y,-} \}}|\tau_{L, D}^{Y,-}\big] \big] \\
        &\rightarrow \mathbb{E}_{t,x}\big[\sup_{\tau\in \mathcal{T}_t^X}\mathbb{E}_{t,x}\big[e^{-r(\tau - t)}  f( X_{\tau}) 1_{\{\tau \le T \}}1_{\{\tau\ge \tau_{L, D}^{X,-} \}}|\tau_{L, D}^{X,-}\big] \big]\quad (n\rightarrow \infty, \delta_t\rightarrow 0)\\
        &=\sup_{\tau \in \mathcal{T}_t^X}\mathbb{E}_{t,x} \big[ e^{-r(\tau - t)}  f( X_{\tau}) 1_{\{\tau \le T \}}1_{\{\tau\ge \tau_{L, D}^{X,-} \}} \big].
    \end{aligned}
    \]
\end{proof}

\begin{proof}[Proof of Theorem \ref{thm:finite-down-out-Amer-Pari-option-converge}]
We will first prove the following claims:

(1) There exists some constant $\widetilde{C}$ independent of $\hat{\delta}$, $n$, $\delta_d$ and $\delta_t$ such that
\[
\Big|\overline{C}_{fo}(s,d,x)-\overline{C}^{\hat{\delta}}_{fo}(s,d,x)\Big|\le \widetilde{C}\hat{\delta}^{\frac{1}{2}}.
\]

(2) There exists some constant $\widetilde{D}$ independent of $\hat{\delta}$, $n$, $\delta_d$ and $\delta_t$ such that
\[
\Big|C_{fo}(s,d,x)-C_{fo}^{\hat{\delta}}(s,d,x)\Big|\le \widetilde{D}\hat{\delta}^{\frac{1}{2}}.
\]

(3) We consider an arbitrary fixed $\hat{\delta}$. Let $\widetilde{\mathrm {H}}$ be the filtration generated by the process $(X, Y,D )$. Let $\tau^{\bar{\delta},-}=\sup\{ t< \tau: t\in \mathbb{T}^{\bar{\delta}}\}$ with $\mathbb{T}^{\bar{\delta}}=\left \{i\bar{\delta}: i\in \mathbb{N}  \right \} $ for some stopping time $\tau$. We set $\bar{\delta}=\frac{T-s}{2^m}$ for some $m\in \mathbb{N}$. For the quantities
\[
\widehat{C}_{fo}^{\hat{\delta}} (s,d,x)=\sup_{\tau\in \widetilde{\mathcal{T}} _{s}}\mathbb{E}_{s,d,x}\Big[e^{-r(\tau^{\hat{\delta}}-s)}f(X_{\tau^{\hat{\delta}}})1_{\{\tau\le T  \}}1_{\{ \tau\le \widetilde{\tau}_{L, D}^{Y,-}\}}\Big],
\]
\[
{C}_{fo}^{\hat{\delta},-} (s,d,x)=\sup_{\tau\in \widetilde{\mathcal{T}} _{s}}\mathbb{E}_{s,d,x}\Big[e^{-r(\tau^{\hat{\delta}}-s)}f(X_{\tau^{\hat{\delta}}})1_{\{\tau\le T  \}}1_{\{ \tau^{\bar{\delta},-}\le \widetilde{\tau}_{L, D}^{Y,-}\}}\Big]
\]
and
\[
\overline{C}_{fo}^{\hat{\delta},-} (s,d,x)=\sup_{\tau\in \widetilde{\mathcal{T}}^X _{s}}\mathbb{E}_{s,d,x}\Big[e^{-r(\tau^{\hat{\delta}}-s)}f(X_{\tau^{\hat{\delta}}})1_{\{\tau\le T  \}}1_{\{ \tau^{\bar{\delta},-}\le \widetilde{\tau}_{L, D}^{X,-}\}}\Big], 
\]
where $\widetilde{\mathcal{T}}^X _{s}$ is the set of $\widetilde{\mathrm{F}}^X$ stopping times taking values in $[s,\infty]$. For any $\epsilon>0$, there exists a $\bar{\delta}_{\epsilon}$ independent of $\delta_d$, $\hat{\delta}$ and $n$ such that for $\bar{\delta}<\bar{\delta}_{\epsilon}$, 
\[
\max\Big\{\big|\overline{C}_{fo}^{\hat{\delta},-} (s,d,x)-\overline{C}_{fo}^{\hat{\delta}} (s,d,x)\big|, \big|\widehat{C}_{fo}^{\hat{\delta}} (s,d,x)- {C}_{fo}^{\hat{\delta},-} (s,d,x) \big|   \Big\}<\epsilon.
\]
Moreover, for any fixed $\bar{\delta}$, 
\[
{C}_{fo}^{\hat{\delta},-} (s,d,x)\rightarrow \overline{C}_{fo}^{\hat{\delta},-} (s,d,x)
\]
as $\delta_d\rightarrow 0$ and $n\rightarrow \infty$. 

(4) For any fixed $\hat{\delta}$, as $\delta_t \to 0$ and $n\rightarrow \infty$, 
\[
C^{\hat{\delta}}_{fo}(s,d,x)\rightarrow \widehat{C}_{fo}^{\hat{\delta}} (s,d,x).
\] 

We will show the proof details for (3) and (4). The proofs of (1) and (2) are similar to those in the proof of Proposition \ref{prop:finite-American-option-converge}.

Proof of (3):  
We first show that $\overline{C}_{fo}^{\hat{\delta},-} (s,d,x)\rightarrow \overline{C}_{fo}^{\hat{\delta}} (s,d,x)$ as $\bar{\delta}\rightarrow 0$.
\begin{align}
    &\Big|\overline{C}_{fo}^{\hat{\delta}} (s,d,x)-\overline{C}_{fo}^{\hat{\delta},-} (s,d,x) \Big|\\
    &\leq \bigg|\sup_{\tau\in \widetilde{\mathcal{T}}^X _{s}}\mathbb{E}_{s,d,x}\Big[e^{-r(\tau^{\hat{\delta}}-s)}f(X_{\tau^{\hat{\delta}}})1_{\{\tau\le T  \}}1_{\{ \tau\le \widetilde{\tau}_{L, D}^{X,-}\}}-e^{-r(\tau^{\hat{\delta}}-s)}f(X_{\tau^{\hat{\delta}}})1_{\{\tau\le T  \}}1_{\{ \tau^{\bar{\delta},-}\le \widetilde{\tau}_{L, D}^{X,-}\}}\Big]\bigg|\\
    &\leq  C \bigg|\sup_{t\in [s,T]}\mathbb{E}_{s,d,x}\Big[1_{\{ t\le \widetilde{\tau}_{L, D}^{X,-}\}}-1_{\{ t^{\bar{\delta},-}\le \widetilde{\tau}_{L, D}^{X,-}\}}\Big]\bigg|\\
    &= C \sup_{t\in [s,T]}\bigg|\mathbb{P}\Big(t\le \widetilde{\tau}_{L, D}^{X,-}\Big)-\mathbb{P}\Big(t^{\bar{\delta},-}\le \widetilde{\tau}_{L, D}^{X,-}\Big)\bigg|\\
   & =  C \sup_{t\in [s,T]\cap \mathbb{T}^{\bar{\delta}}}\bigg(\mathbb{P}\Big(t-\bar{\delta}\le \widetilde{\tau}_{L, D}^{X,-}\Big)-\mathbb{P}\Big(t\le \widetilde{\tau}_{L, D}^{X,-}\Big)\bigg)\\
    &=  C \sup_{t\in [s,T]\cap \mathbb{T}^{\bar{\delta}}}\mathbb{P}\Big(t-\bar{\delta}\le \widetilde{\tau}_{L, D}^{X,-}<t\Big).
\end{align}
The reason for the above is that $\mathbb{P}\big(t\le \widetilde{\tau}_{L, D}^{X,-}\big)$ decreases as $t$ increases and hence by the definition of $t^{\bar{\delta},-}$, the maximum value of the difference in probabilities occurs in $t\in [s,T]\cap \mathbb{T}^{\bar{\delta}}$. Therefore, for any $\epsilon>0$, there exists an $m_{\epsilon}\in \mathbb{N}$ such that for all $\frac{T-s}{2^m}=\bar{\delta}<\bar{\delta}_{\epsilon}=\frac{T-s}{2^{m_{\epsilon}}}$, 
\begin{align}
\sup_{t\in [s,T]\cap \mathbb{T}^{\bar{\delta}}}\mathbb{P}\Big(t-\bar{\delta}\le \widetilde{\tau}_{L, D}^{X,-}<t\Big)<\sup_{t\in [s,T]\cap \mathbb{T}^{\bar{\delta}_{\epsilon}}}\mathbb{P}\Big(t-\bar{\delta}_{\epsilon}\le \widetilde{\tau}_{L, D}^{X,-}<t\Big)\le \epsilon,
\end{align}
which implies that claim to prove. With the same arguments, we can show that ${C}_{fo}^{\hat{\delta},-} (s,d,x)\rightarrow \widehat{C}_{fo}^{\hat{\delta}} (s,d,x)$ as $\bar{\delta}\rightarrow 0$. 

Then we will show that for any fixed $\bar{\delta}$, ${C}_{fo}^{\hat{\delta},-} (s,d,x)\rightarrow \overline{C}_{fo}^{\hat{\delta},-} (s,d,x)$ as $\delta_d\rightarrow 0$ and $n\rightarrow \infty$. The arguments are as follows:
\begin{align}
    &\Big|{C}_{fo}^{\hat{\delta},-} (s,d,x)-\overline{C}_{fo}^{\hat{\delta},-} (s,d,x)  \Big|\\
    &\leq  \bigg|\sup_{\tau\in \widetilde{\mathcal{T}} _{s}}\mathbb{E}_{s,d,x}\Big[ e^{-r(\tau^{\hat{\delta}}-s)}f(X_{\tau^{\hat{\delta}}})1_{\{\tau\le T  \}}1_{\{ \tau^{\bar{\delta},-}\le \widetilde{\tau}_{L, D}^{Y,-}\}}-e^{-r(\tau^{\hat{\delta}}-s)}f(X_{\tau^{\hat{\delta}}})1_{\{\tau\le T  \}}1_{\{ \tau^{\bar{\delta},-}\le \widetilde{\tau}_{L, D}^{X,-}\}}\Big]\bigg|\\
    &\leq  C \sup_{t\in [s-\bar{\delta},T-\bar{\delta}]\cap \mathbb{T}^{\bar{\delta}}}\bigg| \mathbb{P}\Big(t\le \widetilde{\tau}_{L, D}^{Y,-}\Big)- \mathbb{P}\Big( t\le \widetilde{\tau}_{L, D}^{X,-} \Big) \bigg| \rightarrow  0,
\end{align}
as $\delta_d\rightarrow 0$ and $n\rightarrow \infty$ by Proposition \ref{prop:down-out-parisian-time-weak-convergence}. $\widetilde{\mathcal{T}} _{s}$ is the set of $\widetilde{\mathrm {H}}$ stopping times taking values in $[s,\infty]$.

Proof of (4): Let $\widetilde{\mathrm{H}}^{\zeta}$ be the filtration generated by the process $(\zeta, X, Y, D)$. We introduce the following intermediate quantities:
\[
\breve{C}_{fo}^{\hat{\delta}} (s,d,x)=\sup_{\tau \in \widetilde{\mathcal{T}}^{\zeta,Y}_{s}} \mathbb{E}_{s,d,x} \Big[ e^{-r(\tau^{\hat{\delta}} - s)} f(Y_{\tau^{\hat{\delta}}})1_{\{ \zeta_{\tau}\le T \}}1_{\{ \tau\le \widetilde{\tau}_{L, D}^{Y, -} \}} \Big],
\]
\[
\overrightarrow{C}_{fo}^{\hat{\delta}} (s,d,x)=\sup_{\tau \in \widetilde{\mathcal{T}}^\zeta_{s}} \mathbb{E}_{s,d,x} \Big[ e^{-r(\tau^{\hat{\delta}} - s)} f(X_{\tau^{\hat{\delta}}})1_{\{ \zeta_{\tau}\le T \}}1_{\{ \tau\le \widetilde{\tau}_{L, D}^{Y, -} \}} \Big],
\]
where $\widetilde{\mathcal{T}}^{\zeta,Y}_{s}$ and $\widetilde{\mathcal{T}}^\zeta_{s}$ are the sets of $\widetilde{\mathrm{F}}^{\zeta, Y}$ stopping times and $\widetilde{\mathrm{H}}^{\zeta}$ stopping times, respectively, taking values in $[s,\infty]$. With the same arguments in the proof of Proposition \ref{prop:finite-American-option-converge}, we can show that $C^{\hat{\delta}}_{fo}(s,d,x)\rightarrow \breve{C}_{fo}^{\hat{\delta}} (s,d,x)$ as $\delta_t\rightarrow 0$, and $\breve{C}_{fo}^{\hat{\delta}} (s,d,x)\rightarrow \overrightarrow{C}_{fo}^{\hat{\delta}} (s,d,x)$ as $n\rightarrow \infty$. 

Next we show that $\overrightarrow{C}_{fo}^{\hat{\delta}} (s,d,x)\rightarrow \widehat{C}_{fo}^{\hat{\delta}} (s,d,x)$ as $\delta_t \rightarrow 0$.
\begin{align}
    \overrightarrow{C}_{fo}^{\hat{\delta}} (s,d,x)
    &=\sup_{\tau \in \widetilde{\mathcal{T}}^\zeta_{s}} \mathbb{E}_{s,d,x} \Big[ e^{-r(\tau^{\hat{\delta}} - s)} f(X_{\tau^{\hat{\delta}}})1_{\{ \tau\le T^{\delta_t} \}}1_{\{ \tau\le \widetilde{\tau}_{L, D}^{Y, -} \}} \Big]\\
    &= \mathbb{E}_{s,d,x} \bigg[ \sup_{\tau \in \mathcal{T}^X_{s}} \mathbb{E}_{s,d,x} \Big[ e^{-r(\tau^{\hat{\delta}} - s)} f(X_{\tau^{\hat{\delta}}})1_{\{ \tau\le T^{\delta_t} \}}1_{\{ \tau\le \widetilde{\tau}_{L, D}^{Y, -} \}} \Big|T^{\delta_t}, \widetilde{\tau}_{L, D}^{Y, -}  \Big]  \bigg]\\
    &=  \mathbb{E}_{s,d,x} \bigg[ \sup_{\tau \in \mathcal{T}^X_{s}} \mathbb{E}_{s,d,x} \Big[ e^{-r(\tau^{\hat{\delta}} - s)} f(X_{\tau^{\hat{\delta}}})1_{\{ \tau\le \widetilde{\tau}_{L, D}^{Y, -} \}} \Big|T^{\delta_t}, \widetilde{\tau}_{L, D}^{Y, -}  \Big]1_{\{ \widetilde{\tau}_{L, D}^{Y, -}\le T^{\delta_t} \}}\\
    &\;\;\;\; + \sup_{\tau \in \mathcal{T}^X_{s}} \mathbb{E}_{s,d,x} \Big[ e^{-r(\tau^{\hat{\delta}} - s)} f(X_{\tau^{\hat{\delta}}})1_{\{ \tau\le T^{\delta_t}\}} \Big|T^{\delta_t}, \widetilde{\tau}_{L, D}^{Y, -}  \Big]1_{\{ \widetilde{\tau}_{L, D}^{Y, -}> T^{\delta_t} \}} \bigg]\\
    &= \mathbb{E}_{s,d,x} \bigg[ \sup_{\tau \in \mathcal{T}^X_{s}} \mathbb{E}_{s,d,x} \Big[ e^{-r(\tau^{\hat{\delta}} - s)} f(X_{\tau^{\hat{\delta}}})1_{\{ \tau\le \widetilde{\tau}_{L, D}^{Y, -} \}} \Big|T^{\delta_t}, \widetilde{\tau}_{L, D}^{Y, -}  \Big]1_{\{ \widetilde{\tau}_{L, D}^{Y, -}\le T^{\delta_t} \}}\bigg]\\
    &\;\;\;\; + \mathbb{E}_{s,d,x} \bigg[\sup_{\tau \in \mathcal{T}^X_{s}} \mathbb{E}_{s,d,x} \Big[ e^{-r(\tau^{\hat{\delta}} - s)} f(X_{\tau^{\hat{\delta}}})1_{\{ \tau\le T^{\delta_t}\}} \Big|T^{\delta_t}, \widetilde{\tau}_{L, D}^{Y, -}  \Big]1_{\{ \widetilde{\tau}_{L, D}^{Y, -}> T^{\delta_t} \}} \bigg]\\
    &= \mathbb{E}_{s,d,x} \Big[ V\big( \widetilde{\tau}_{L, D}^{Y, -}, s,d,x\big)1_{\{ \widetilde{\tau}_{L, D}^{Y, -}\le T^{\delta_t} \}}\Big]+\mathbb{E}_{s,d,x} \Big[ V\big( T^{\delta_t}, s,d,x\big)1_{\{ T^{\delta_t}<  \widetilde{\tau}_{L, D}^{Y, -} \}}\Big]
\end{align}
Here $V\big(T,s,d,x\big)=\sup_{\tau \in \mathcal{T}^X_{s}} \mathbb{E}_{s,d,x} \Big[ e^{-r(\tau^{\hat{\delta}} - s)} f(X_{\tau^{\hat{\delta}}})1_{\{ \tau\le T \}} \Big]$ is the value function of a finite maturity Bermudan option with maturity $T$. Suppose that $\mathbb{P}\big[\widetilde{\tau}_{L, D}^{Y, -}=T \big]=0$ for any $T>0$. Because $T^{\delta_t}\rightarrow T$ as $\delta_t\rightarrow 0$, we have that 
\[
\mathbb{E}_{s,d,x} \Big[ V\big( \widetilde{\tau}_{L, D}^{Y, -}, s,d,x\big)1_{\{ \widetilde{\tau}_{L, D}^{Y, -}\le T^{\delta_t} \}}\Big]\rightarrow \mathbb{E}_{s,d,x} \Big[ V\big( \widetilde{\tau}_{L, D}^{Y, -}, s,d,x\big)1_{\{ \widetilde{\tau}_{L, D}^{Y, -}\le T \}}\Big]
\]
and
\[
\mathbb{E}_{s,d,x} \Big[ V\big( T^{\delta_t}, s,d,x\big)1_{\{ T^{\delta_t}<  \widetilde{\tau}_{L, D}^{Y, -} \}}\Big]\rightarrow \mathbb{E}_{s,d,x} \Big[ V\big( T, s,d,x\big)1_{\{ T<  \widetilde{\tau}_{L, D}^{Y, -} \}}\Big]
\]
as $\delta_t\rightarrow 0$. Note that 
\begin{align}
   &\mathbb{E}_{s,d,x} \Big[ V\big( \widetilde{\tau}_{L, D}^{Y, -}, s,d,x\big)1_{\{ \widetilde{\tau}_{L, D}^{Y, -}\le T \}}\Big]+\mathbb{E}_{s,d,x} \Big[ V\big( T, s,d,x\big)1_{\{ T<  \widetilde{\tau}_{L, D}^{Y, -} \}}\Big]\\
    &=  \mathbb{E}_{s,d,x} \bigg[ \sup_{\tau \in \mathcal{T}_{s}^X} \mathbb{E}_{s,d,x} \Big[ e^{-r(\tau^{\hat{\delta}} - s)} f(X_{\tau^{\hat{\delta}}})1_{\{ \tau\le \widetilde{\tau}_{L, D}^{Y, -} \}} \Big|T, \widetilde{\tau}_{L, D}^{Y, -}  \Big]1_{\{ \widetilde{\tau}_{L, D}^{Y, -}\le T \}}\\
    &\;\;\;\;  + \sup_{\tau \in \mathcal{T}^X_{s}} \mathbb{E}_{s,d,x} \Big[ e^{-r(\tau^{\hat{\delta}} - s)} f(X_{\tau^{\hat{\delta}}})1_{\{ \tau\le T  \}}\Big|T, \widetilde{\tau}_{L, D}^{Y, -}\Big]1_{\{ \widetilde{\tau}_{L, D}^{Y, -}> T\}}\bigg]\\
    &=  \mathbb{E}_{s,d,x} \bigg[ \sup_{\tau \in \mathcal{T}^X_{s}} \mathbb{E}_{s,d,x} \Big[ e^{-r(\tau^{\hat{\delta}} - s)} f(X_{\tau^{\hat{\delta}}})1_{\{ \tau\le T  \}}1_{\{ \tau\le \widetilde{\tau}_{L, D}^{Y, -} \}} \Big|T, \widetilde{\tau}_{L, D}^{Y, -}  \Big]1_{\{ \widetilde{\tau}_{L, D}^{Y, -}\le T \}}\\
    & \;\;\;\; + \sup_{\tau \in \mathcal{T}^X_{s}} \mathbb{E}_{s,d,x} \Big[ e^{-r(\tau^{\hat{\delta}} - s)} f(X_{\tau^{\hat{\delta}}})1_{\{ \tau\le T  \}}1_{\{ \tau\le \widetilde{\tau}_{L, D}^{Y, -} \}}\Big|T, \widetilde{\tau}_{L, D}^{Y, -}\Big]1_{\{ \widetilde{\tau}_{L, D}^{Y, -}> T\}}\bigg]\\
    &=  \mathbb{E}_{s,d,x} \bigg[ \sup_{\tau \in \mathcal{T}^X_{s}} \mathbb{E}_{s,d,x} \Big[ e^{-r(\tau^{\hat{\delta}} - s)} f(X_{\tau^{\hat{\delta}}})1_{\{ \tau\le T  \}}1_{\{ \tau\le \widetilde{\tau}_{L, D}^{Y, -} \}} \Big|T, \widetilde{\tau}_{L, D}^{Y, -}  \Big]\bigg]\\
    &=  \widehat{C}_{fo}^{\hat{\delta}} (s,d,x).
\end{align}
Therefore, we have that $\overrightarrow{C}_{fo}^{\hat{\delta}} (s,d,x)\rightarrow \widehat{C}_{fo}^{\hat{\delta}} (s,d,x)$ as $\delta_t \rightarrow 0$. 

We further note that
\begin{align}
    &\Big| \widehat{C}_{fo}^{\hat{\delta}} (s,d,x)-C^{\hat{\delta}}_{fo}(s,d,x) \Big|\\
    &\leq  \Big|\widehat{C}^{\hat{\delta}}_{fo}(s,d,x)-\overrightarrow{C}_{fo}^{\hat{\delta}} (s,d,x)\Big| + \Big| \overrightarrow{C}_{fo}^{\hat{\delta}} (s,d,x)-\breve{C}_{fo}^{\hat{\delta}} (s,d,x)\Big| +  \Big|\breve{C}_{fo}^{\hat{\delta}} (s,d,x)-{C}_{fo}^{\hat{\delta}} (s,d,x)\Big|.
\end{align}
Combining the results obtained to the above and we conclude that
\[
\Big| \widehat{C}_{fo}^{\hat{\delta}} (s,d,x)-C^{\hat{\delta}}_{fo}(s,d,x) \Big|\rightarrow 0
\]
as $\delta_t\rightarrow 0$ and $n\rightarrow \infty$ for any fixed $\hat{\delta}$.

Finally, we note that
    \[
    \begin{aligned}
    &\Big|C_{fo}(s,d,x)-\overline{C}_{fo}(s,d,x)\Big|\\
    &\leq  \Big|C_{fo}(s,d,x)-C^{\hat{\delta}}_{fo}(s,d,x)\Big|+\Big|C^{\hat{\delta}}_{fo}(s,d,x)-\widehat{C}_{fo}^{\hat{\delta}} (s,d,x)\Big|+\Big|\widehat{C}_{fo}^{\hat{\delta}} (s,d,x)-{C}_{fo}^{\hat{\delta},-} (s,d,x)\Big|\\
    &  \; \; \; \;    + \Big|{C}_{fo}^{\hat{\delta},-} (s,d,x)-\overline{C}_{fo}^{\hat{\delta},-} (s,d,x)\Big|+\Big|\overline{C}_{fo}^{\hat{\delta},-} (s,d,x)-\overline{C}_{fo}^{\hat{\delta}} (s,d,x)\Big|+\Big|\overline{C}_{fo}^{\hat{\delta}} (s,d,x)-\overline{C}_{fo}(s,d,x)\Big|.
    \end{aligned}
    \]
    Then the proof of Theorem \ref{thm:finite-down-out-Amer-Pari-option-converge} is concluded by combining the claims (1), (2), (3) and (4).  
\end{proof}

\begin{proof}[Proof of Lemma \ref{lem:perpetual-down-out-American-option-intermedia}]
    \[
    \begin{aligned}
        &\big| \overline{C}_{fo}(0,d,x)-\overline{C}^{\hat{\delta}}_{po}(d,x)\big|\\
        &=\sup _{\tau \in \widetilde{\mathcal{T}}^X}\mathbb{E}_{d, x}\big[e^{-r\tau^{\hat{\delta}}} f(X_{\tau^{\hat{\delta}}})1_{\{ \tau\le \widetilde{\tau}_{L, D}^{X,-} \}} \big]-\sup _{\tau \in \widetilde{\mathcal{T}}^X_{0,T}} \mathbb{E}_{0,d, x}\left[e^{-r\tau^{\hat{\delta}}} f\left(X_{\tau^{\hat{\delta}}}\right) 
 1_{\{ \tau\le \widetilde{\tau}_{L, D}^{X,-} \}} \right]\\
 &\leq  \sup _{\tau \in \widetilde{\mathcal{T}}^X_{0,T}} \mathbb{E}_{0,d, x}\left[e^{-r\tau^{\hat{\delta}}} f\left(X_{\tau^{\hat{\delta}}}\right) 
 1_{\{ \tau\le \widetilde{\tau}_{L, D}^{X,-} \}} \right]+\sup _{\tau \in \widetilde{\mathcal{T}}^X\setminus \widetilde{\mathcal{T}}^X_{0,T}} \mathbb{E}_{0,d, x}\left[e^{-r\tau^{\hat{\delta}}} f\left(X_{\tau^{\hat{\delta}}}\right) 
 1_{\{ \tau\le \widetilde{\tau}_{L, D}^{X,-} \}} \right]\\
 &\;\;\;\;-\sup _{\tau \in \widetilde{\mathcal{T}}^X_{0,T}} \mathbb{E}_{0,d, x}\left[e^{-r\tau^{\hat{\delta}}} f\left(X_{\tau^{\hat{\delta}}}\right) 
 1_{\{ \tau\le \widetilde{\tau}_{L, D}^{X,-} \}} \right]\\
 &=\sup _{\tau \in \widetilde{\mathcal{T}}^X\setminus \widetilde{\mathcal{T}}^X_{0,T}} \mathbb{E}_{0,d, x}\left[e^{-r\tau^{\hat{\delta}}} f\left(X_{\tau^{\hat{\delta}}}\right) 
 1_{\{ \tau\le \widetilde{\tau}_{L, D}^{X,-} \}} \right].
    \end{aligned}
    \]

    By Assumption \ref{assump:payoff-growth}, there exists a $T_\epsilon$, such that for any $\tau^{\hat{\delta}}>T_\epsilon>T$, $\big| e^{-r\tau^{\hat{\delta}}} f\left(X_{\tau^{\hat{\delta}}}\right)\big|<\epsilon$, and 
    \[
    \begin{aligned}
        \big| \overline{C}_{fo}(0,d,x)-\overline{C}^{\hat{\delta}}_{po}(d,x)\big|
        &\leq  \sup _{\tau \in \widetilde{\mathcal{T}}^X\setminus \widetilde{\mathcal{T}}^X_{0,T}} \mathbb{E}_{0,d, x}\left[e^{-r\tau^{\hat{\delta}}} f\left(X_{\tau^{\hat{\delta}}}\right) 
 1_{\{ \tau\le \widetilde{\tau}_{L, D}^{X,-} \}} \right]\\
 &< \epsilon \sup _{\tau \in \widetilde{\mathcal{T}}^X\setminus \widetilde{\mathcal{T}}^X_{0,T}} \mathbb{E}_{0,d, x}\left[
 1_{\{ \tau\le \widetilde{\tau}_{L, D}^{X,-} \}} \right] \leq \epsilon.
    \end{aligned}
    \]

    Using similar arguments, we can show that $\big| \widetilde{C}^{\hat{\delta}}_{fo}(0,d,x)-C^{\hat{\delta}}_{po}(d,x)\big|<\epsilon$. This concludes the proof.
\end{proof}

\begin{proof}[Proof of Theorem \ref{thm:perpetual-down-out-Amer-Pari-option-converge}]
    We will first prove the following claims: 
    
    (1) There exists some constant $\widetilde{C}$ independent of $\hat{\delta}$, $n$ and $\delta_d$  such that
\[
\left|\overline{C}_{po}(d, x)-\overline{C}^{\hat{\delta}}_{po}(d, x)\right|\le \widetilde{C}\hat{\delta}^{\frac{1}{2}}
\]

(2) There exists some constant $\widetilde{D}$ independent of $\hat{\delta}$, $n$ and $\delta_d$ such that
\[
\left|C_{po}(d, x)-C^\delta_{po}(d, x)\right|\le \widetilde{D}\hat{\delta}^{\frac{1}{2}}
\]

(3) For any fixed $\hat{\delta}$, as $n\rightarrow \infty$ and $\delta_d\rightarrow 0$, we have 
\[
C_{po}^{\hat{\delta}}(d, x)\rightarrow \overline{C}_{po}^{\hat{\delta}}(d, x).
\]  

We will show details for (3). The proofs of (1) and (2) are similar to the proof of Proposition \ref{prop:perpetual-down-in-Amer-option-convergence}. 

Proof of (3): 
With the same arguments in the proof of Theorem \ref{thm:finite-down-out-Amer-Pari-option-converge}, we can show that $\big|\widetilde{C}^{\hat{\delta}}_{fo}(s,d,x)-\overline{C}^{\hat{\delta}}_{fo}(s,d,x)  \big|\rightarrow 0$ as $n\rightarrow \infty$ and $\delta_d\rightarrow 0$. Then we have 
\[
\begin{aligned}
    &\big| C_{po}^{\hat{\delta}}(d, x)- \overline{C}_{po}^{\hat{\delta}}(d, x) \big|\\
    &=  \big| C_{po}^{\hat{\delta}}(d, x)-\overline{C}_{po}^{\hat{\delta}}(d, x)+\widetilde{C}^{\hat{\delta}}_{fo}(0,d,x)-\widetilde{C}^{\hat{\delta}}_{fo}(0,d,x)+\overline{C}^{\hat{\delta}}_{fo}(0,d,x)-\overline{C}^{\hat{\delta}}_{fo}(0,d,x) \big|\\
    &\leq  \big|\overline{C}^{\hat{\delta}}_{fo}(0,d,x)- \overline{C}_{po}^{\hat{\delta}}(d, x) \big|+\big| \widetilde{C}^{\hat{\delta}}_{fo}(0,d,x)- C_{po}^{\hat{\delta}}(d, x)\big|+\big| \widetilde{C}^{\hat{\delta}}_{fo}(0,d,x)-\overline{C}^{\hat{\delta}}_{fo}(0,d,x)   \big|.
\end{aligned}
\]
By Lemma \ref{lem:perpetual-down-out-American-option-intermedia}, with the same arguments in the proof of Proposition \ref{prop:perpetual-down-in-Amer-option-convergence}, the claim (3) is proved.

Finally we note that,
\[
\begin{aligned}
    &\big| C_{po}(d, x)-\overline{C}_{po}(d, x) \big|\\
    &\leq  \big| C_{po}(d, x)-C^{\hat{\delta}}_{po}(d,x)  \big|+\big| C^{\hat{\delta}}_{po}(d,x)-\overline{C}^{\hat{\delta}}_{po}(d,x) \big|+\big|\overline{C}^{\hat{\delta}}_{po}(d,x)-\overline{C}_{po}(d, x)  \big|.
\end{aligned}
\]
Then the proof of Theorem \ref{thm:perpetual-down-out-Amer-Pari-option-converge} is concluded by combining the results in the claims (1), (2) and (3). 

\end{proof}

\bibliographystyle{chicagoa}

\end{document}